%% file: tsta_2019_08.tex
\documentclass[11pt]{amsart}

\usepackage{amsfonts, amsmath, amssymb}
\usepackage{amsaddr}
\usepackage{graphicx,color}
\usepackage{listings}
\usepackage{algorithm}
\usepackage{algpseudocode}
\usepackage{cancel} 
\usepackage{wrapfig}
\usepackage{color}
\usepackage{tikz}
\usepackage{url}
\usepackage{stmaryrd}

\newtheorem{theorem}{Theorem}[section]
\newtheorem{corollary}[theorem]{Corollary}

\newtheorem{lemma}[theorem]{Lemma}
\newtheorem{proposition}[theorem]{Proposition}

\theoremstyle{definition}
\newtheorem{definition}[theorem]{Definition}

\newtheorem{example}[theorem]{Example}

\newcommand{\LLL}{\mathfrak{L}}

\newcommand{\RRR}{\mathfrak{R}}
\newcommand{\ARRR}{\mathfrak{R}^t}
\newcommand{\IARRR}{\mathfrak{R}_{\infty}^{t}}
\newcommand{\PARRR}{\mathfrak{R}_\textnormal{per}^{t}}

\newcommand{\TTTSSS}{\mathbf{TS}}

\newcommand{\ta}{\mathrm{TA}}
\newcommand{\dta}{\mathrm{DTA}}
\newcommand{\nta}{\mathrm{NTA}}

\newcommand{\ntaeps}{\mathrm{eNTA}}

\newcommand{\Actions}{\Sigma}

\newcommand{\eActions}{\Sigma_{\epsilon}}

\newcommand{\Locs}{\mathcal{Q}}
\newcommand{\Clocks}{\mathcal{C}}
\newcommand{\ResetClocks}{\mathcal{C}_{rst}}

\newcommand{\clk}{c}

\newcommand{\Trans}{\mathcal{T}}
\newcommand{\Guards}{\mathcal{G}}

\newcommand{\ClockVal}{\mathcal{V}}
\newcommand{\zerov}{\textbf{0}}

\newcommand{\Naturals}{\mathbb{N}}
\newcommand{\ZNaturals}{\mathbb{N}_0}
\newcommand{\Integers}{\mathbb{Z}}

\newcommand{\PReals}{\mathbb{R}_{\geq 0}}

\newcommand{\frc}[1]{\left\{{#1}\right\}}

\newcommand{\lcm}{\mathrm{lcm}}

\newcommand{\powerset}[1]{\mathcal P \left({#1}\right) }

\usepackage{color}

\begin{document}

\title{The Timestamp of Timed Automata}
\author{Amnon Rosenmann}
\address{Graz University of Technology, Steyrergasse 30, A-8010 Graz, Austria}
\email{rosenmann@math.tugraz.at}


\date{}
\maketitle

\begin{abstract}
Let $\ntaeps$ be the class of non-deterministic timed automata with silent transitions. Given $A \in \ntaeps$, we effectively compute its timestamp: the set of all pairs (time value, action) of all observable timed traces of $A$. We show that the timestamp is eventually periodic and that one can compute a simple deterministic timed automaton with the same timestamp as that of $A$. As a consequence, we have a partial method, not bounded by time or number of steps, for the general language non-inclusion problem for $\ntaeps$. We also show that the language of $A$ is periodic with respect to suffixes. 
\end{abstract}
\section{Introduction}
\label{sec:intro}
Timed automata ($\ta$) are finite automata extended with clocks that measure the time that elapsed since past events in order to control the triggering of future events.
They were defined by Alur and Dill in their seminal paper \cite{ta} as abstract models of real-time systems and were implemented in tools like UPPAAL \cite{uppaal}, Kronos \cite{Kronos98}, RED \cite{Wang04} and PRISM\cite{Prism11}.

A fundamental problem in this area is the reachability problem, which in its basic form asks whether a given location of a timed automaton is reachable from the initial location.
The set of states of the system (i.e., locations and valuation to the clocks) is, in general, an infinite uncountable set.
However, through the construction of a region automaton, which contains finitely-many equivalence classes of regions \cite{ta}, the reachability problem becomes a decidable problem (though of complexity PSPACE-complete).

Research on the reachability problem went beyond the above basic question.
In \cite{CY92} it is shown that the problem of the minimum and maximum reachability time is also PSPACE-complete.
In another work, \cite{CJ99}, which is more of a theoretical nature, the authors show that some problems on the relations between states may be defined in the decidable theory of the domain of real numbers equipped with the addition operation.
In particular, the reachability problem between any two states is decidable.
For other aspects of the reachability problem, also in the context of variants and extensions of timed automata (e.g. with game and probability characteristics) we refer to \cite{CY92},\cite{member-ta-ha}, \cite{TY01},  \cite{WZP03}, \cite{control_ta}, \cite{HP06}, \cite{CHKM11}, \cite{HOW12}.
In this paper we generalize the reachability problem in another direction.
We show that the problem of computing the set of all time values on which any observable transition occurs (and thus, a location is reached by an observable transition) is solvable.
This set, called the \emph{timestamp} of the automaton $A$ and denoted $\TTTSSS(A)$, is more precisely defined to be the set of all pairs $(t, a)$ that appear in the observable timed traces of $A$.
Note that for this definition it does not matter whether we consider infinite runs or finite ones.

We show that the timestamp is in the form of a union of action-labeled open intervals with integral end-points, and action-labeled points of integral values.
When the timestamp is unbounded in time then it is eventually periodic.
 
The set of languages defined by the class $\dta$ of deterministic timed automata is strictly included in the set of languages defined by the class $\nta$ of non-deterministic timed automata \cite{ta}, \cite{Finkel06}, and the latter is strictly included in the set of languages defined by the class $\ntaeps$ of non-deterministic timed automata with silent transitions \cite{ta-eps}.
The fundamental problem of inclusion of the language accepted by a timed automaton $A$ (e.g. the implementation) in the language accepted by the timed automaton $B$ (e.g. the specification) is undecidable for the class $\nta$ but decidable for the class $\dta$.
On the other hand, for special sub-classes or modifications it was shown that decidability exists (see \cite{bbbb,ta-eps,era,updatable-ta,Ouak_one_clk,Ouak_time_bound,bounded-time,Dec_TA_Survey,bound-det} for a partial list).
However, the abstraction (or over-approximation) represented in the form of a timestamp is a discrete object, in which questions like inclusion of timestamps or universality are decidable.
In fact, we show that for any given non-deterministic timed automaton with silent transitions, one can construct a simple deterministic timed automaton having the same timestamp. 

The computation of the timestamp is done through the construction of a periodic augmented region automaton $\PARRR(A)$.
It is a region automaton augmented with a global non-resetting clock $t$ and containing periodic regions and periodic transitions: they are defined modulo a time period $L \in \Naturals$.
This kind of abstraction demonstrates a periodic nature which is absent, in general, from timed traces: there are timed automata with no timed traces that are eventually periodic (see Example~\ref{ex:non-period}).
Periodic transitions were introduced in \cite{periodic}, where it was shown that they increase the expressiveness of $\dta$, though they are less expressive than silent transitions.

The construction of the periodic automaton is preceded by defining the infinite augmented region automaton $\IARRR(A)$, in which the values of the clock $t$ are unbounded.
Then, after exhibiting the existence of a pattern that repeats itself every $L$ time units, we fold the infinite automaton into a finite one according to this periodic structure.

Our construction shows that the language of a timed automaton $A \in \ntaeps$ is periodic with respect to suffixes: for every run $\varrho$ with suffix $\varsigma$ that occurs after passing a fixed computable time there are infinitely-many runs of $A$ with the same suffix $\varsigma$, but with the suffix shifted in time by multiples of $L$.
Note that this result does not follow from the pumping lemma, which does not hold in general in timed automata \cite{pumping}. 

In Section~\ref{sec:ta} basic definitions concerning timed automata are given.
Then, in Section~\ref{sec:sing_path} we describe the trail and timestamp of a single path of a timed automaton, more from a geometric than from an algebraic point of view, after treating the absolute-time clock $t$ as part of the system.
The augmented and infinite augmented region automaton, $\ARRR(A)$ and $\IARRR(A)$, are presented in Section~\ref{sec:IARA}, and then, in Section~\ref{sec:per}, we explore the time-periodicity in them, so that $\IARRR(A)$ can be folded into the finite periodic augmented region automaton $\PARRR(A)$ (Section~\ref{sec:APARA}).
In last section (Section~\ref{sec:timestamp}) we construct the entire eventually periodic timestamp.
As for the general language inclusion problem in $\ntaeps$, the timestamp, or better - the more informative automaton $\PARRR$, may serve as a tool in demonstrating the non-inclusion relation between the languages of two members of $\ntaeps$s.
\section{Timed Automata with Silent Transitions}
\label{sec:ta}
A timed automaton is an abstract model aiming at capturing the temporal behavior of real-time systems.
It is a finite automaton extended with a finite set of clocks defined over $\PReals$, the set of non-negative real numbers.
It consists of a finite set of \emph{locations} $q$ with a finite set of \emph{transitions} $\tau$ between the locations, while time, measured by the clocks, is continuous.
A transition at time $t$ can occur only if the condition expressed as a \emph{transition guard}
is satisfied at $t$.
The transition is immediate - no clock is advancing in time. However, some of the clocks may be reset to zero.

There are two sorts of transitions: \emph{observable} transitions, which can be traced by an outside observer, and \emph{silent} transitions, which are inner transitions and thus cannot be observed from the outside.
There are finitely-many types of observable transitions, each type labeled by a unique \emph{action} $a \in \Actions$, whereas all the silent transitions have the same label $\epsilon$.
In $\nta$, the class of non-deterministic timed automata, there exist states in which two transitions from the same location $q$ can be taken at the same time and with the same action but to two different locations $q'$ and $q''$.
When this situation cannot happen, the TA is deterministic.

Let $\ZNaturals := \Naturals \cup \{0\}$ and let $\powerset{S}$ be the power set of a set $S$.
A transition guard is a conjunction of constraints of the form $\clk \sim n$, where $\clk$ is a clock, $\sim \ \in \{<,\leq, =,\geq, >\}$ and $n \in \ZNaturals$.
A formal definition of $\ntaeps$ is as follows.
\begin{definition}[$\ntaeps$]
\label{def:ntaeps}
A \emph{non-deterministic timed automaton with silent transitions} $A \in \ntaeps$ is a tuple $(\Locs, q_0, \eActions, \Clocks,
\Trans)$, where:
\begin{enumerate}
\item $\Locs$ is a finite set of locations and $q_0$ is the initial location;
\item $\eActions = \Sigma \cup \{\epsilon\}$ is a finite set of transition labels, called actions, where $\Actions$ refers to the observable actions and $\epsilon$ represents a silent transition;
\item $\Clocks$ is a finite set of clock variables;
\item $\Trans \subseteq \Locs \times \eActions \times \Guards \times \powerset{\Clocks} \times \Locs$ is a finite set of transitions of the form $(q, a, g, \ResetClocks, q')$, where:
\begin{enumerate}
\item $q,q' \in \Locs$ are the source and the target locations respectively;
\item $a \in \eActions$ is the transition action;
\item $g \in  \Guards$ is the \emph{transition guard};
\item $\ResetClocks \subseteq \Clocks$ is the subset of clocks to be reset.
\end{enumerate}
\end{enumerate}
\end{definition}

A clock \emph{valuation} $v$ is a function $v:\Clocks \to \PReals$. 
We denote by $\ClockVal$ the set of all clock valuations and by
$\textbf{d}$ the valuation which assigns the value $d$ to every clock.
Given a valuation $v$ and $d \in \PReals$, we define $v+d$ to be the valuation $(v+d)(\clk) := v(\clk)+d$ for every $ \clk \in \Clocks$.
The valuation $v[\ResetClocks]$, $\ResetClocks \subseteq \Clocks$, is defined to be $v[\ResetClocks](c) = 0$ for $c \in \ResetClocks$ and $v[\ResetClocks](c) = v(c)$ for $c \notin \ResetClocks$.

The \emph{semantics} of $A \in \ntaeps$ is given by the \emph{timed transition system}
$\llbracket A \rrbracket = (S, s_0, \PReals, \eActions, T)$, where:
\begin{enumerate}
\item $S = \{(q,v) \in \Locs \times \ClockVal \}$ is the set of states, with
$s_0 = (q_0, \zerov)$ the initial state;
\item $T \subseteq S \times (\eActions \cup \PReals) \times S$ is the transition relation.
The set $T$ consists of
\begin{enumerate}
\item \emph{Timed transitions (delays):} $(q,v) \xrightarrow{d} (q, v+d)$, where $d \in \PReals$;
\item \emph{Discrete transitions (jumps):} $(q,v) \xrightarrow{a} (q',v')$, where $a \in \eActions$ and there exists a transition $(q, a, g, \ResetClocks, q')$ in $\Trans$, such that 
for each clock $c$, $v(c)$ satisfies the constraints of $g$ regarding $c$, and $v' = v[\ResetClocks]$.
\end{enumerate}
\end{enumerate}

A (finite) \emph{run} $\varrho$ of $A \in \ntaeps$ is a sequence of alternating timed and discrete transitions of the form 
$$(q_{0}, \zerov) \xrightarrow{d_{1}} (q_{0}, \textbf{d}_1) \xrightarrow{a_{1}} (q_{1}, v_{1}) \xrightarrow{d_{2}} \cdots 
\xrightarrow{d_{k}} (q_{k-1}, v_{k-1} + d_{k}) \xrightarrow{a_{k}} (q_{k}, v_{k})
$$
and \emph{duration} $T=\sum_{j=1}^{k}d_j$.
The run $\varrho$ of $A$ induces the \emph{timed trace} (\emph{timed word})
$$
\lambda = (t_{1}, a_{1}), (t_{2}, a_{2}), \ldots, (t_{k}, a_{k}),
$$
with $a_i \in \eActions$ and $t_{i} = \Sigma_{j=1}^{i} d_j$. 
From the latter we can extract the \emph{observable timed trace} (\emph{observable timed word}), which is obtained by deleting from $\lambda$ all the pairs containing silent transitions.
%
Note that when the TA is deterministic then each timed trace refers to a unique run.
We remark that we did not include for a location $q$ the location invariants in the definition of timed automata since these invariants can be incorporated in the guards of the transitions to $q$ (for the clocks that are not reset at the transitions) and in those emerging from $q$.
We also do not distinguish between accepting and non-accepting locations as they do not change the analysis and results concerning the reachability problems that are dealt with here.
Thus, the \emph{language} $\LLL(A)$ of $A$ refers here to the set of observable timed traces of $A$ without restricting it to those observable timed traces of runs that end in acceptable locations.
\section{The Trail and Timestamp of a Single Path}
\label{sec:sing_path}
In this section we describe the trail and timestamp of a single path of a TA.
Given a timed automaton $A \in \ntaeps$ over $s$ clocks $x_1, \ldots, x_s$, we add to it a non-resetting global clock $t$ that displays absolute time.
A finite \emph{path} in $A$ has the form $\gamma = q_0 \tau_1 q_1 \tau_2 \cdots \tau_n q_n$ of alternating locations and transitions, with $q_0$ the initial location and $\tau_i$ a transition between $q_{i-1}$ and $q_i$, $i=1,\ldots,n$, that is, a path here refers to the standard definition in a directed graph.
A run of the TA induces a \emph{trajectory} in the non-negative part of the $t x_1 \cdots x_s$-space that is a piecewise-linear curve (the discontinuity is the clocks reset).
\begin{definition}[Trajectory of a run]
	Let $\{ t, x_1, \ldots, x_s \}$ be an ordered set of clocks of $A \in \ntaeps$. 
	Let $\varrho$ be a run of duration $T$ of $A$.
	The \emph{trajectory} of $\varrho$ is the set of points $(t, x_1, \ldots, x_s)$ in the $t x_1 \cdots x_s$-space visited during $\varrho$, where $0 \leq t \leq T$.
\end{definition}
Next, we define the \emph{trail} of a path.
%
%
\begin{definition}[Trail of a path]
	The \emph{trail of a path} $\gamma$
	is the union of the trajectories of all feasible runs along $\gamma$, that is, runs that follow the locations and discrete transitions of $\gamma$.
\end{definition}
%
The \emph{trail legs}, the parts of the trail between clocks reset, are in the form of \emph{zones} \cite{zones}, a conjunction of diagonal constraints $x_i - x_j < n_{ij}$ or $x_i - x_j \leq n_{ij}$, $n_{ij} \in \Integers$, bounded by transition constraints $x_i \sim n_i$, where $\sim \ \in \{<,\leq, =,\geq, >\}$, $n_i \in \ZNaturals$.
Each trail leg can be further partitioned into \emph{simplicial trails}, which are (possibly unbounded) parallelotopes consisting of a sequence of \emph{regions} \cite{ta} arranged along the directional vector ${\bf 1} = (1,1, \ldots, 1)$.
Each region ${\bf n} + \Delta$ is in the form of an open (unless it is a point) simplex $\Delta$ that is a hyper-triangle of dimension $0 \leq d \leq s+1$.
The simplex $\Delta$ is characterized by the fractional values $\{x_i\}$ of the clock variables, and each point in the simplex satisfies the same fixed ordering of the form
\begin{equation}
\label{eq:simplex}
0 \preceq_1 \frc{x_{i_1}} \preceq_2 \frc{x_{i_2}} \preceq_3 \cdots \preceq_s \frc{x_{i_s}} < 1,
\end{equation} 
 where $\preceq_i \, \in \{=, < \}$.
The integral point ${\bf n} \in \ZNaturals^{s+1}$ consists of the integral parts of the values of the clocks $x_0, x_1, \ldots, x_s$, and it indicates the lowest point in 
the $x_0 \cdots x_s$-space of the boundary of the region.  
Each region has a unique \emph{immediate time-successor}, which is the next region along the directional vector ${\bf 1}$, as long as no clock is reset on an event.

When the simplicial trail $S$ is $k$-dimensional then the immediate time-successor of an open $k$-simplex (a simplex of dimension $k$, $1 \leq k \leq r+1$) is a $(k-1)$-simplex and vice-versa, where each $(k-1)$-simplex is a face of its neighbouring $k$-simplices.
A region which is in the form of a $k$-simplex refers to the case where the fractional parts of the clocks are all non-zero, and then its immediate time-successor is a $(k-1)$-simplex, in which the integral part of the clock with maximal fractional part is increased by 1 while its fractional part is set to zero. The order between the other clocks remains as before.
The switch from a $k-1$-simplex into a $k$-simplex occurs when a clock of fractional part 0 turns into a positive fractional part and the order of the fractional parts of the clocks as well as their integral values remains as before.  
Thus, at each switch there is a cyclic shift in the fractional parts of the clocks, which results in a periodic sequence of simplices along a simplicial trail.

Let $ d_i \geq 0$ be the feasible \emph{duration} of the $i$-th event along a path $\gamma$.
That is, $d_i = M_i - m_i$, where $M_i$ is the supremum, over the runs along $\gamma$, of the time at which the $i$-th event of the run occurs, and $m_i$ is defined as the infimum of the same set.
In case of an automaton with a single clock $x$, if $x$ resets on this transition then the size of the temporal part of the timestamp of the $i$-th event increases by $d_i$, resulting in an increase in the width of the parallelogram that represents the trail of $\gamma$ after the $i$-th event, and possibly increasing the dimension of the trail from 1 to 2.
Otherwise, the width remains as before.
In case of multiple clocks, the dimension of the trail can increase, decrease or stay the same after an event with reset of clocks: clocks with the same fractional part can be separated, resulting in an increase of the dimension, while clocks whose fractional parts become identical (namely, $0$) contribute to a decrease of the dimension.

Let us look at a simple example of the trail and timestamp of a path in an automaton with a single clock.
\begin{example}
	In Fig.~\ref{fig:trail_and_ts_path}(a) a TA is drawn, and in Fig.~\ref{fig:trail_and_ts_path}(b) we see the trail and timestamp of the finite path $\gamma: (0) \xrightarrow{a} (1) \xrightarrow{b} (2) \xrightarrow{a} (3) \xrightarrow{a} (2)$, where '$a$-timestamp' refers to the projection on the $t$-axis of the elements $(t,a)$ of the timestamp, and similarly for '$b$-timestamp'.
	The first event occurs when $x = 1$ and the timestamp is $\{1\} \times \{a\}$.
	Then $x$ resets and the trail (a straight line of slope $1$) continues from the $t$-axis.
	Event 2 occurs when $1 \leq x \leq 3$ with timestamp $[2,4] \times \{b\}$ and a reset of $x$.
	After that event the trail is 2-dimensional (a parallelogram).
	Event $3$ occurs when $1 < x < 2$ without clock reset, and the orthogonal projection to the $t$-axis gives the timestamp $(3,6) \times \{a\}$ (here $(3,6)$ is the open interval $3 < t < 6$).
	The fourth event happens when $x = 3$ and its timestamp is $[5,7] \times \{a\}$.
	The timestamp of $\gamma$ is the union of the above sets, that is, $S_1 \times \{a\} \cup S_2 \times \{b\}$, with $S_1 = \{1\} \cup (3,7]$ and $S_2 = [2,4]$.
	\begin{figure}[htb]
		\centering
		\scalebox{0.5}{ 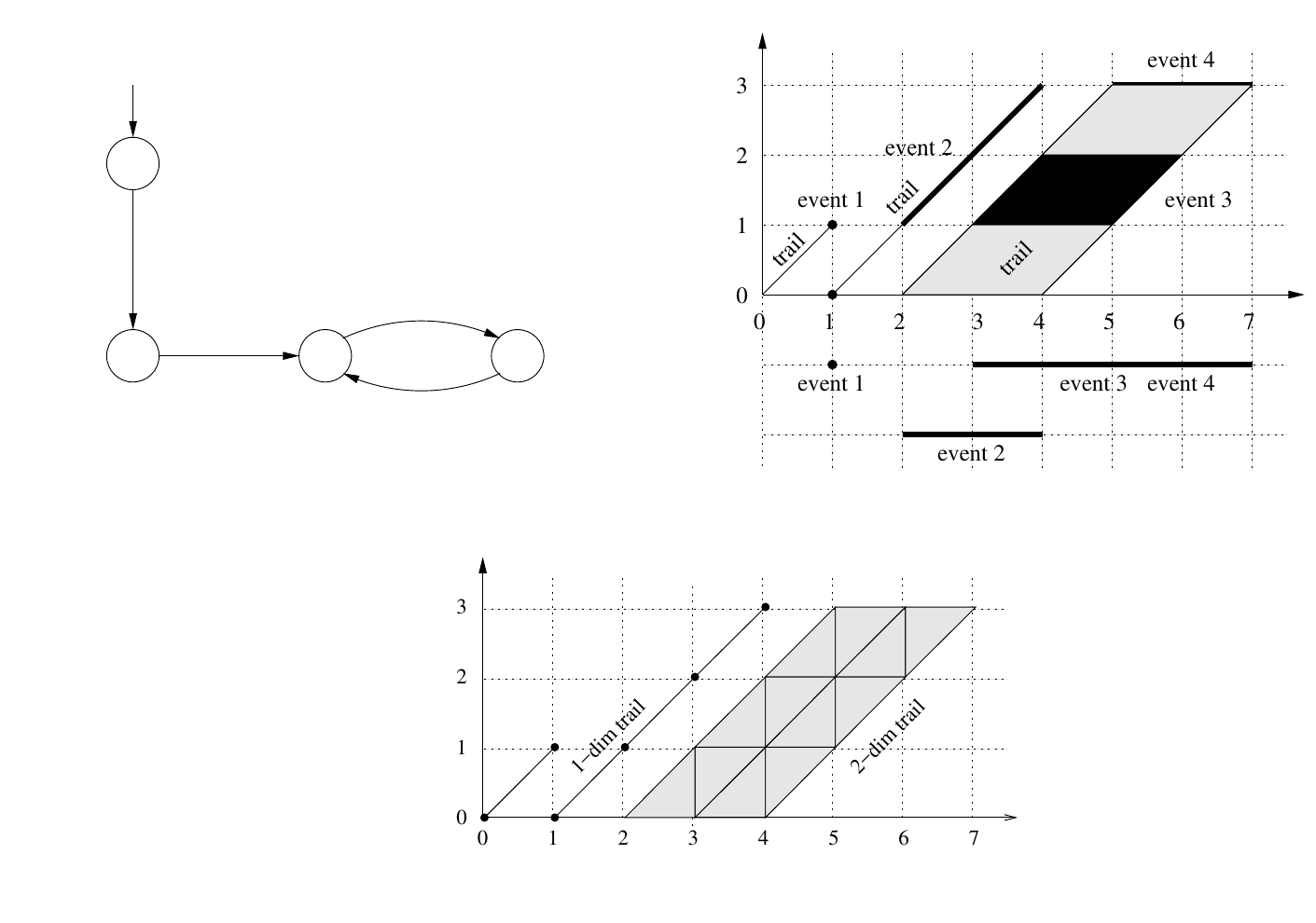 }
		\caption{Trail, timestamp and regions of a path (single clock)}
		\label{fig:trail_and_ts_path}
	\end{figure} 
\end{example}
\begin{definition}[Timestamp of a run]
	The \emph{timestamp of a run} $\varrho$ is the set of pairs $(t_i, a_i) \in \PReals \times \Actions$ of the observable timed trace induced by $\varrho$.
\end{definition}
A finite path in $A$ has the form $\gamma = q_0 \tau_1 q_1 \tau_2 \cdots \tau_n q_n$ of alternating locations and transitions, and we always assume that $q_0$ is the initial location.
Such a path is an abstraction of a run since the temporal part is omitted.
Given a path $\gamma$ in $A$, there may be many possible runs along $\gamma$, and we say that $\gamma$ is \emph{feasible} when there is at least one run along it.
\begin{definition}[Timestamp of a path]
	The \emph{timestamp of a feasible path} $\gamma$ of $A$ is the union of the timestamps of all runs $\varrho$ along $\gamma$.
\end{definition}
Each instance of a transition along $\gamma$ is an \emph{event}.
That is, a transition is a static object which joins two locations of the TA, whereas an event refers to a specific occurrence of a transition within the path $\gamma$.
Hence, several events along a path may refer to the same transition of the TA.
\begin{definition}[Timestamp of an event in a path]
	The \emph{timestamp of an event} in a path $\gamma$ is the union of the timestamps of that event of all runs along $\gamma$.
	It is the part of the timestamp of the path that refers to that event.
\end{definition}
\begin{proposition}
\label{pr:mult_ev_timestamp}
The timestamp of each event is either a labeled integral point or a labeled (open, closed or half-open) interval between points $m$ and $n$, $m < n$, $m \in \ZNaturals$ and $n \in \Naturals \cup \infty$.
\end{proposition}
\begin{proof} 
	The trail of each path is composed of simplices as in \eqref{eq:simplex} residing on the integral grid.
	The intersection of such a simplex $\Delta$ with a domain satisfying a transition constraint of the form $x_i \sim n_i$, where $\sim \ \in \{<,\leq, =,\geq, >\}$, $n_i \in \ZNaturals$ is either the whole of $\Delta$ or the empty set.
	A possible reset of clocks $x_i$ during an event results in mapping $\Delta$ to another simplex $\Delta'$, which may be of smaller dimension.
	Thus, it suffices to show that the timestamp of a single simplex $\Delta$ is of the required form.
	But the temporal part of the timestamp of $\Delta$ is the set $n + S$, were $n \in \ZNaturals$ and $S$ is the set of values of the clock $t=x_0$ in $\Delta$.
	Since $S$ is either $\{0\}$ or the open interval $(0,1)$, we get that the timestamp of $\Delta$ is either an action-labeled integral point $\{n\}$ or an action-labeled open unit interval $(n,n+1)$.

	Another way of proving the claim is via linear programming.
	Suppose that a path $\gamma$ contains $r$ events and that the time of event $i$, $0\leq i \leq r$, is recorded by the variable $t_i$. Then we can represent $t_i$ as satisfying equalities and inequalities over the integers: instead of referring to a regular clock $x$ in the constraint of the $i$-th transition along $\gamma$, we refer to the variable $t_j$, where the $j$-th transition along $\gamma$ was the last time that the clock $x$ was reset. The result then follows by the fact that the corresponding maximum and minimum linear programming problems have integer solutions.	
\end{proof}	
\begin{definition}[Timestamp of a timed automaton]
	The timestamp $\TTTSSS(A)$ of a timed automaton $A$ is the set of all pairs $(t,a)$, such that an observable transition with action $a$ occurs at time $t$ in some run of $A$.
\end{definition}
\section{Augmented and Infinite Augmented Region Automaton}
\label{sec:IARA}
\subsection{Infinite Augmented Region Automaton}
Given a (finite) timed automaton $A$, the region automaton $\RRR(A)$ \cite{ta} is a finite \emph{discretized} version of $A$, such that time is abstracted and both automata define the same untimed language.
Each vertex in $\RRR(A)$ records a location $q$ in $A$ and a region $r$, which is either in the form of a simplex (as described in Section~\ref{sec:sing_path}) or an unbounded region, in which the value of at least one of the clocks is $\top$, meaning that it passed the maximal integer value $M$ that appears in the transition guards.
The regions partition the space of clock valuations into equivalence classes, where two valuations belong to the same equivalence class if and only if they agree on the clocks with $\top$ value and on the integral parts and the order among the fractional parts of the other clocks.
The edges of $\RRR(A)$ are labeled by the transition actions, and they correspond to the actual transitions that occur in the runs of $A$.
Using the time-successor relation over the clock regions (see \cite{ta}), the region automaton can be effectively constructed.
As shown in \cite{ta}, through the region automaton the questions of reachable locations and states of $A$ and the actions along the (possibly infinitely-many) paths that lead to these locations, i.e. the untimed language of $A$, become decidable.

Now we define the \emph{infinite augmented region automaton} $\IARRR(A)$.
First, we add to $A$ a clock $t$ that measures absolute time, does not appear in the transition guards, is never reset to $0$ and does not affect the runs and timed traced of $A$.
Next, we construct the region automaton augmented with $t$.
The construction is similar to the construction of the standard region automaton with respect to the regular clocks (all clocks except for $t$) and the maximal bound $M$, that is, the time regions of each regular clock $x_i$ are $\{0\}$, $(0,1)$, $\{1\},(1,2), \ldots, M, >M$, the latter being unbounded and refers to all values of $x$ greater than $M$.
The integration of the clock $t$ is as follows. The construction of regions is as usual by considering the integral parts and the order of the fractional parts of all clocks, including $t$. The only difference is that the integral part of $t$ is in $\ZNaturals$ and not bounded by $M$. Thus, the infinitely-many time-regions associated with $t$ are the alternating point and open unit interval: $\{0\}$, $(0,1)$, $\{1\}$, $(1,2), \ldots$ (see Fig.~\ref{fig:APTA_a}(b)).
Hence,  $\IARRR(A)$ contains information about absolute time that is lacking from the standard region automaton.
\begin{definition}[Infinite augmented region automaton]
	\label{def:inf_aug_region_automaton}
	Given $A \in \ntaeps$ extended with the clock $t$ that measures absolute time, a corresponding \emph{infinite augmented region automaton} $\IARRR(A)$ is a tuple $(V, v_0, E, \eActions)$, where:
	\begin{enumerate}
		\item $V$ is an infinite (in general) set of vertices of the form $(q, {\bf n}, \Delta)$, where $q$ is a location of $A$ and the pair $({\bf n}, \Delta)$ is a region, with
		\begin{equation}
			{\bf n} = (n_0, n_1, \ldots, n_s) \in \ZNaturals \times \{0, 1, \ldots, M, \top \}^{s}
		\end{equation}
		containing the integral parts of the clocks $t, x_1, \ldots, x_s$, and $\Delta$
		is the simplex defined by the order of the fractional parts of the clocks.
		\item $v_0 = (q_0, {\bf 0}, {\bf 0})$ is the initial vertex with $q_0$ the initial location of $A$ and with all clocks having integral part and fractional part equal to 0.
		\item $E$ is the set of edges.
		There is an edge
		\begin{equation}
		(q, r) \xrightarrow{a} (q',r')
		\end{equation}
		labeled with $a$ in $\IARRR(A)$ if and only if there is a run of $A$  which contains a timed transition followed by a discrete transition of the form
		\begin{equation}
		(q, v) \xrightarrow{d} (q, v + d) \xrightarrow{a} (q', v'),
		\label{eq:regtrans}
		\end{equation}
		such that the clock valuation $v$ over $t, x_1, \ldots, x_s$ represents a point in the region $r$ and the clock valuation $v'$ represents a point in the region $r'$.
		\item $\eActions = \Sigma \cup \{\epsilon\}$ is the finite set of actions that are edge labels.
	\end{enumerate}
\end{definition}
We note that there may be infinitely-many edges going-out of the same region in $\IARRR(A)$ (see Fig.~\ref{fig:APTA_a}(b)).
\begin{proposition}
	\label{prop:finite_IARA}
	For each positive integer $n$, one can effectively construct the part of $\IARRR(A)$ which contains all regions with $t \leq n$ and all in-coming edges of these regions. 
\end{proposition}
\begin{proof}
	There are finitely-many regions obeying the constraint $t \leq n$.
	These regions and their in-coming edges can be constructed the same way as a standard region automaton is constructed, starting with the initial location and proceeding step by step according to the immediate time-successor regions (which include the clock $t$) and according to the transitions of $A$.
	Indeed, the additional clock $t$ is only responsible for a finer partition of regions, but its introduction does not affect the transition guards of $A$.  
	Note also that since the clock $t$ never resets, there are no edges connecting regions with $t > n$ to regions with $t \leq n$.
	Hence, the number of edges of $\IARRR(A)$ restricted to $t \leq n$ is finite.   
\end{proof}

The benefit of introducing the clock $t$ into the region automaton is that we can know approximately at what absolute time an action occurs. 
For example, suppose that $A$ has a single clock $x$ and that $x$ is reset on a transition from location $q$ to location $q'$.
Then, in the corresponding region automaton the information about the time spent at location $q$ before moving to $q'$ is lost.
In $\IARRR(A)$, however, if we take the absolute time at which an action occurs to be $n+0.5$ when entering a region whose time-region (the value of $t$) is the open interval $(n, n+1)$, and the absolute time $n$ when entering a region whose time-region is exactly $t = n$, then it is possible to construct from it an (infinite) approximate timed automaton with a single clock and which differs from $A$ by at most $0.5$ time units at each action.

The timestamp of the TA $A$, denoted $\TTTSSS(A)$, is the union of the timestamps of all observable transitions of $A$, that is, the set of all pairs $(t,a)$, such that an observable transition with action $a$ occurs at time $t$ in some run of $A$.
We define also the timestamp of $\IARRR(A)$.
\begin{definition}[Timestamp of $\IARRR(A)$]
	The timestamp of $\IARRR(A)$, denoted $\TTTSSS(\IARRR(A))$, is the union of sets $s \times a$, where $s$ is a time-region of $t$ (an integral point $\{n\}$ or an open unit interval $(n,n+1)$) that is part of a region of a vertex of $\IARRR(A)$ and $a \in \Sigma$ is a label of an edge of $\IARRR(A)$ that is directed towards this vertex.
\end{definition}

\begin{proposition}
	\label{prop:eq_timestamp}
	$\TTTSSS(A) = \TTTSSS(\IARRR(A))$.
\end{proposition}
\begin{proof}
	By definition of the infinite augmented region automaton $\IARRR(A)$, its regions are exactly the clock-regions which are visited by runs of the TA $A$ extended with the clock $t$.
	In particular, the time-regions of $\IARRR(A)$ are the time-regions that are visited by the runs on the extended TA.
	Thus, $\TTTSSS(A) \subseteq \TTTSSS(\IARRR(A))$.
	By Proposition~\ref{pr:mult_ev_timestamp}, this is an equality since for each open interval $(n,n+1)$ representing absolute time that is visited in some run of $A$ on an action $a$, the set of all runs of $A$ cover all the points of this interval with the same action $a$.
\end{proof}
\subsection{Augmented Region Automaton}
A second construction is the augmented region automaton, denoted $\ARRR(A)$, in which we consider only the fractional part of $t$ and ignore its integral part.
$\ARRR(A)$ is a finite folding of $\IARRR(A)$, obtained by identifying vertices that contain the same data except for the integral part of $t$, and the corresponding edges. Thus, $t$ has only two time-regions: $\{0\}$ and $(0,1)$.
%
As a compensation, we assign weights to the edges of $\ARRR(A)$, as explained below.
  
\begin{definition}[Augmented region automaton]
	\label{def:aug_region_automaton}
	Given a non-deterministic timed automaton with silent transitions $A \in \ntaeps$, extended with the absolute-time clock $t$, a corresponding (finite) \emph{augmented region automaton} $\ARRR(A)$ is a tuple $(V, v_0, E, \eActions, W^*)$, where:
	\begin{enumerate}
		\item $V$ is the set of vertices.
		Each vertex is a triple $(q, {\bf n}, \Delta)$,	where $q$ is a location of $A$ and the pair $({\bf n}, \Delta)$ is a region, with
		\begin{equation}
			{\bf n} = (n_1, \ldots, n_s) \in \{0, 1, \ldots, M, \top \}^{s}
		\end{equation}
		containing the integral parts of the clocks $x_1, \ldots, x_s$, and $\Delta$
		is the simplex defined by the fractional parts of the clocks $t, x_1, \ldots, x_s$.
		\item $v_0 = (q_0, {\bf 0}, {\bf 0})$ is the initial vertex.
		\item $E$ is the set of edges.
		There is an edge $(q, r) \xrightarrow{a} (q',r')$ labeled with action $a$ if and only if there is a run of $A$  which contains a timed transition followed by a discrete transition of the form $(q, v) \xrightarrow{d} (q, v + d) \xrightarrow{a} (q', v')$, such that, when ignoring the integral part of the time measured by $t$, the clock valuation $v$ represents a point in the region $r$ and the clock valuation $v'$ represents a point in the region $r'$.
		
		\item $\eActions = \Sigma \cup \{\epsilon\}$ is the finite set of actions.
		\item $W^*$ is the set of weights on the edges.
		Each weight $m$, possibly marked with $`*`$, is $m = \lfloor t_1 \rfloor - \lfloor t_0 \rfloor \in [0..M]$, where $\lfloor t_1 \rfloor$ is the integral part of the value of $t$ in the target location and $\lfloor t_0 \rfloor$ - in the source location in the corresponding run of $A$.
	\end{enumerate}
\end{definition}
	There may be more than one edge between two vertices of $\ARRR(A)$, each one with a distinguished weight.
	A marked weight $m^*$ represents infinitely-many consecutive values $m, m+1, m+2, \ldots$ as weights between the same two vertices, with $m$ being the minimal value of such a sequence.
	It refers to a transition to or from a region $r$ in which all regular clocks have passed the maximal integer $M$ appearing in a transition guard.

\begin{example}
	In Fig.~\ref{fig:APTA_a}(a) we see a very simple TA $A$ containing a transition to an unbounded region.
	The corresponding infinite augmented region automaton $\IARRR(A)$ is shown in Fig.~\ref{fig:APTA_a}(b).
	Each vertex of $\IARRR(A)$ is represented by a rounded rectangle containing the original location of $A$ (circled, on the left), the integral values of $t$ and of $x$ (in the top of the rectangle) and the simplex (in the bottom).
	Notice that when the value of $x$ is greater than $M=0$ it is marked by $\top$ and its fractional part is ignored.
	To the left of $\IARRR(A)$ we see the discretization of time $t$ into time-regions, and each vertex of $\IARRR(A)$ is drawn in the level of its time-region.
	In  Fig.~\ref{fig:APTA_a}(c) the augmented region automaton $\ARRR(A)$ is shown.
	Here the integral part of the value of $t$ is ignored.
	The edge labeled by $0^*$ represents the infinitely-many differences in the integral parts of the values of $t$: $0, 1, 2, \ldots$.
	Similarly, the edge labeled with $1^*$ refers to the differences $1, 2, 3, \ldots$.
	\begin{figure}[t]
		\centering
		\scalebox{0.45}{ 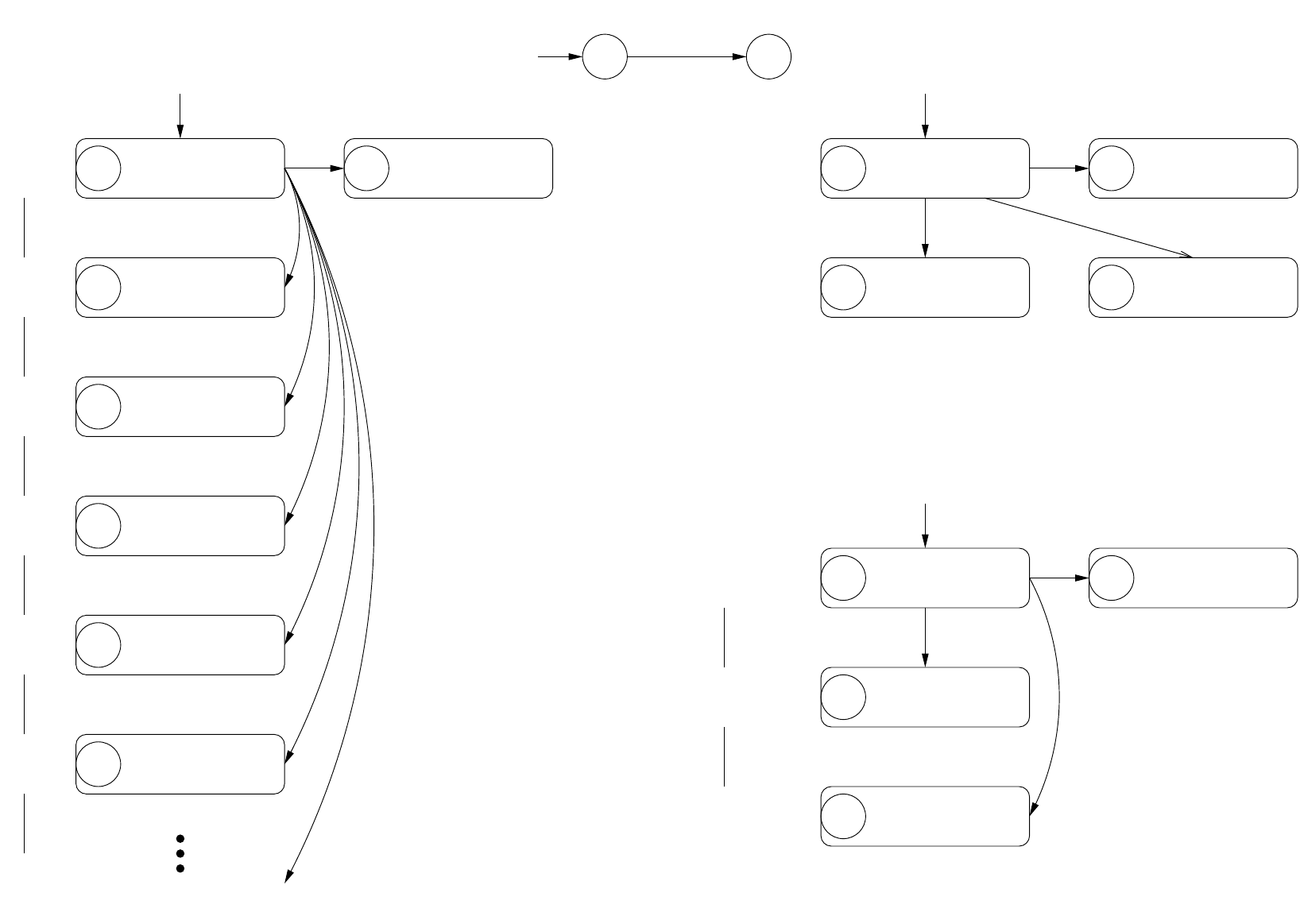 }
		\caption{$(a)$ $A \in \ta$; $(b)$ The infinite augmented region automaton $\IARRR(A)$;  $(c)$ The augmented region automaton $\ARRR(A)$; $(d)$ A periodic augmented region automaton $\PARRR(A)$. Each rectangle represents a vertex containing the location of $A$ (circled, left), the integral values of $t$ and $x$ (top) and the simplex (bottom).}
		\label{fig:APTA_a}
	\end{figure}
\end{example}

	The languages $\LLL({\ARRR(A)})$ of $\ARRR(A)$ and $\LLL({\IARRR(A)})$ of $\IARRR(A)$  consist of all observable timed traces but, in contrast to the language $\LLL({A})$ of $A$, in each pair $(t_i,a_i)$ the time $t_i$ is not exact: it is either an exact integer $n$ or an arbitrary value of an interval $(n,n+1)$ that satisfies $t_i \geq t_{i-1}$.
	Thus, $\LLL({\ARRR(A)})$ and  $\LLL({\IARRR(A)})$ are less abstract than the untimed language $\LLL(\RRR(A))$ of the region automaton $\RRR(A)$ but are more abstract than $\LLL({A})$: one cannot, in general, distinguish between a transition that occurs without any time delay, e.g. when $x_i \geq 0$, and a transition that demands a time delay, e.g. when $x_i > 0$.
	When comparing $\LLL({\ARRR(A)})$ and $\LLL({\IARRR(A)})$ then, since $\ARRR(A)$ may be obtained from $\IARRR(A)$, it is clear that $\LLL({\ARRR(A)})$ cannot be less abstract than $\LLL({\IARRR(A)})$.	  
	But, in fact, these region automata are equally informative: for each positive integer $n$, one can effectively construct $\IARRR(A)$ up to time $t=n$, as in Proposition~\ref{prop:finite_IARA}, by unfolding $\ARRR(A)$ and recovering absolute time $t$ by summing up the weights of the edges along the taken paths.
	Indeed, since the transitions in $A$ do not rely on $t$, by taking the quotient of $\IARRR(A)$ by 'forgetting' the integral part of $t$, the only loss of information is the time difference in $t$ between the target and source regions, but then this information is regained in the form of weight on the corresponding edge of $\ARRR(A)$.
	Thus, we have the following.
\begin{proposition}
	\label{prop:eq_info}
	 $\LLL({\ARRR(A)}) = \LLL({\IARRR(A)})$.
\end{proposition}

As with $\IARRR(A)$, we can construct from $\ARRR(A)$ an approximate automaton, this time a finite and deterministic one, which approximates $A$ with a maximal error of $1/2$ time units at each observed transition.
This automaton has only one clock and this clock resets at every transition.
The maximal error $\epsilon = 1/2$ could be further reduced to $1/n$ by allowing transitions in the approximate automaton to occur at times $p/n$, $p \in \ZNaturals$ and only on such times.

\section{Eventual Periodicity}
\label{sec:per}
In this section we address the main topic of this paper: exploring the time-periodic property of $\ta$.
In addition to demonstrating its existence, we show how one can actually compute the parameters of a period.
\subsection{Non-Zeno Cycles in $\ARRR(A)$}
$\ARRR(A)$ is in the form of a finite connected directed graph with an initial vertex.
Every edge of $\ARRR(A)$ corresponds to a feasible transition in $A$ (contained in a run of $A$).
In what follows, a `path' in $\ARRR(A)$ is a directed path that starts at the initial vertex $g_0$, unless otherwise stated.
\begin{definition}[Duration of a path]
Given a path $\gamma$ in $\ARRR(A)$, its minimal integral \emph{duration}, or simply duration, $d(\gamma) \in \ZNaturals$ is the sum of the weights on its edges, where a weight $m^*$ is counted as $m$.
\end{definition}
\begin{definition}[(Non)-Zeno cycle]
A  cycle of $\ARRR(A)$ of duration $0$ is called a \emph{Zeno cycle} .
Otherwise, it is a \emph{non-Zeno cycle}.
\end{definition}
A path is called \emph{simple} if no vertex of it repeats itself, and we let $D$ be the maximal duration of a simple path in $\ARRR(A)$.
\begin{lemma}
\label{lem:reg_time}
There exists a minimal positive integer $t_{\textnormal{nz}} \leq D+1$, the non-Zeno threshold time, such that every path $\gamma$ of $\ARRR(A)$ that is of (minimal) duration $t_{\textnormal{nz}}$ or more contains a vertex belonging to some non-Zeno cycle.
\end{lemma}
\begin{proof}
Indeed, if $\ARRR(A)$ does not contain non-Zeno cycles then we can take $t_{\textnormal{nz}} = D+1$ and the claim holds vacuously.
So, suppose that $\ARRR(A)$ contains non-Zeno Cycles.
Then each path of duration $D+1$ must contain non-Zeno cycles because otherwise the Zeno cycles could have been removed, without changing the duration of the path, resulting in a simple path of duration $D+1$ - a contradiction.
\end{proof}
In order to compute $t_{\textnormal{nz}}$ we can explore the simple paths of $\ARRR(A)$, say in a breadth-first manner, up to the time $t_0$ in which each such path either cannot be extended to a path of a larger duration or any extension of it hits a vertex belonging to some non-Zeno cycle.
Then $t_{\textnormal{nz}}=t_0 +1$, which may be much smaller than $D+1$.
\subsection{A Period of $\ARRR(A)$}
A set $S$ is \emph{minimal} with respect to some property if for every element $e \in S$ the set $S \smallsetminus \{e\}$ does not satisfy the property.
\begin{definition}[Covering set of non-Zeno cycles]
A set $C$ of non-Zeno cycles of $\ARRR(A)$ is called a \emph{covering set of non-Zeno cycles} if every path $\gamma$ of $\ARRR(A)$ whose duration $d(\gamma)$ is at least $t_{\textnormal{nz}}$ intersects a cycle in $C$ in a common vertex.
\end{definition}
Without loss of generality, we may assume that a covering set of non-Zeno cycles is minimal.
\begin{definition}[Period of $\ARRR(A)$]
A time \emph{period} (or just period) $L$ of $\ARRR(A)$ is a common multiple of the set of durations $d(\pi)$, $\pi \in C$, for
some fixed (minimal) covering set of non-Zeno cycles $C$.
For convenience, we also set $L$ to be greater than $M$, unless $\ARRR(A)$ does not contain non-Zeno cycles, in which case we define $L$ to be 0.
\end{definition}

We remark that if we want to compute a minimal period $L>M$ we need to conduct a thorough exploration of the duration of cycles in $\ARRR(A)$, taking into account their common factors, but this computation is not needed for the results presented here.

\subsection{Eventual Periodicity of $\IARRR(A)$}
Let $t_{\textnormal{nz}}, C, L$ be as above, with $C$ fixed.
We denote by $\IARRR(A)|_{t \geq n}$ the subgraph of $\IARRR(A)$ that starts at time-level $n$, that is, the set of vertices of $\IARRR(A)$ with absolute time $t \geq n$ and their out-going edges.
\begin{definition}[$L$-shift in time]
	Given a subgraph $G$ of $\IARRR(A)$, an \emph{$L$-shift in time} of $G$, denoted $G+L$, is the graph obtained by adding the value $L$ to each value of the integral part of the clock $t$ in $G$ and leaving the rest of the data unaltered.
	We also denote by $V(G)+L$ the $L$-shift in time for the set of vertices of $G$, with $v+L$ in case $V = \{v\}$.
\end{definition}
\begin{lemma}
	\label{lem:forward_period}
	If $\IARRR(A)$ is not bounded in time then 
	$$\IARRR(A)|_{t \geq t_\textnormal{nz}} + L \subseteq \IARRR(A)|_{t \geq t_\textnormal{nz}+L}.$$
\end{lemma}
\begin{proof}
	First we show that the inclusion holds for the set of vertices of the above subgraphs.
	Let $\gamma$ be a path of $\IARRR(A)$ which terminates in a vertex $v_1 \in \IARRR(A)|_{t \geq t_\textnormal{nz}}$.
	Let $\gamma' = p(\gamma)$ be the image of $\gamma$ under the projection to $\ARRR(A)$.
	If $\gamma$ contains an edge $e_1$ whose image $e'_1=p(e_1)$ is labeled by a marked weight $m^*$ then we can replace $e_1$ by another edge $e_2 \in p^{-1}(e'_1)$ whose delay is greater by $L$ than the delay of $e_1$.
	So, suppose that $e_1$ starts in the vertex $u_1$ and terminates in $w_1$. Then $e_2$ starts in $u_1$ and terminates in the vertex $w_2=w_1+L$ and then the path continues as in $\gamma$ but with an $L$-shift in time, terminating in the vertex $v_2=v_1 +L$.
	Otherwise, no edge of $\gamma'$ has a marked weight.
	Since $d(\gamma) \geq t_{\textnormal{nz}}$ then by Lemma~\ref{lem:reg_time} and the definition of $L$, $\gamma'$ contains a vertex $v'$ that belongs to a non-Zeno cycle $\pi$ and whose duration is a factor of $L$.
	Hence, by a 'pumping' argument, we can extend $\gamma'$ with $L / d(\pi)$ cycles of $\pi$ that start and end in $v'$ and then reach the vertex $v_2 = v_1 +L$ in the pre-image in $\IARRR(A)$ of this extended path.
	
	The inclusion of the out-going edges follows from the fact that the out-going edges do not depend on the value of $t$.
\end{proof}
Let us denote by $V_k$, $k = 0,1,2, \ldots$, the set of vertices
$$V_k = V(\IARRR(A)|_{t \geq t_\textnormal{nz}+kL}) \smallsetminus V(\IARRR(A)|_{t \geq t_\textnormal{nz}+(k+1)L}).$$
\begin{theorem}
\label{th:eventual_period}
If the infinite augmented region automaton $\IARRR(A)$ is not bounded in time then it is eventually periodic: there exists an integral time $t_\textnormal{per} > 0$ such that 
$$\IARRR(A)|_{t \geq t_\textnormal{per}} + L = \IARRR(A)|_{t \geq t_\textnormal{per}+L}.$$
\end{theorem}
\begin{proof}
By Lemma~\ref{lem:forward_period}, $V_k + L \subseteq V_{k+1}$, for $k \geq 0$.
But there is a bound on the number of possible vertices of $V_k$ since $t$ is bounded, hence the sequence $V_k$ eventually stabilizes.
The result then follows since for the out-going edges the same argument given in the proof of Lemma~\ref{lem:forward_period} holds also here.
\end{proof}
When $\IARRR(A)$ is finite then we can set $t_\textnormal{per}$ to be $t_\textnormal{max} + 1$, where $t_\textnormal{max}$ is the maximal integral time of $\IARRR(A)$.
By the following proposition, a possible value for $t_\textnormal{per}$ can be effectively computed when $\IARRR(A)$ is infinite.
\begin{proposition}
\label{prop:cons_eq}
if $|V_k| = |V_{k+1}| = |V_{k+2}|$ for some $k$ then we can set
$t_\textnormal{per} = t_\textnormal{nz}+kL$.
\end{proposition}
\begin{proof}
The equalities $|V_k| = |V_{k+1}| = |V_{k+2}|$ are equivalent to $V_{k+1} = V_k +L$ and $V_{k+2} = V_{k+1} +L$. 
By induction, it suffices to show that these equalities imply that $|V_{k+2}| = |V_{k+3}|$, that is, $V_{k+3} = V_{k+2} +L$.
Let $v \in V_{k+3}$.
We need to show that there exists $v' \in V_{k+2}$ such that $v=v'+L$.

Suppose that $v$ is reached by an edge from a vertex $u \in V_{k+1} \cup V_{k+2}$.
Since $V_{k+2} = V_{k+1} +L = V_k + 2L$, there exists a vertex $u'=u-L \in V_{k} \cup V_{k+1}$ and this vertex is connected to a vertex  $v' = v-L \in V_{k+2}$.

Otherwise, $v$ is reached by an edge $e_1$ from a vertex $u$ in $V_k$ or earlier, and the time difference $d$ between $u$ and $v$ is greater than $2L$.
This implies that the projection $p(e_1) \in \ARRR(A)$ is of unbounded time delay $m^*$.
Since $L>M$ then $d-L>M$.
Hence, there is another edge $e_2$ in $\IARRR(A)$, that is also a pre-image of $p(e_1)$, and which joins $u$ to a vertex $v' \in V_{k+2}$, where $v=v'+L$.
\end{proof}
\begin{example}
	This example refers to the TA of Fig.~\ref{fig:complex2} (a).
	In order to make the analysis of its time-periodic structure simpler, we changed the guard on the transition from location $1$ to location $2$ to be simpler (Fig.~\ref{fig:complex} (a)), so that in the resulting infinite augmented region automaton $\IARRR(A)$ (Fig.~\ref{fig:complex} (b)) we can clearly see two different cycles of period $6$ (circled in dotted lines) (the edges with label $c$ are only partly shown).
	We then added the original guard between locations $1$ and $2$ (Fig.~\ref{fig:complex2} (a)).
	In the additional part in $\IARRR(A)$ (Fig.~\ref{fig:complex2} (b)) we see two more cycles, one of period $11$ and one of period $5$.
	We can still use a period of length $6$ for this more complex automaton, but the existence of cycles of other lengths results in a longer time until reaching the repeated periodic part of the entire automaton.
	\begin{figure}[t]
		\centering
		\scalebox{0.4}{ 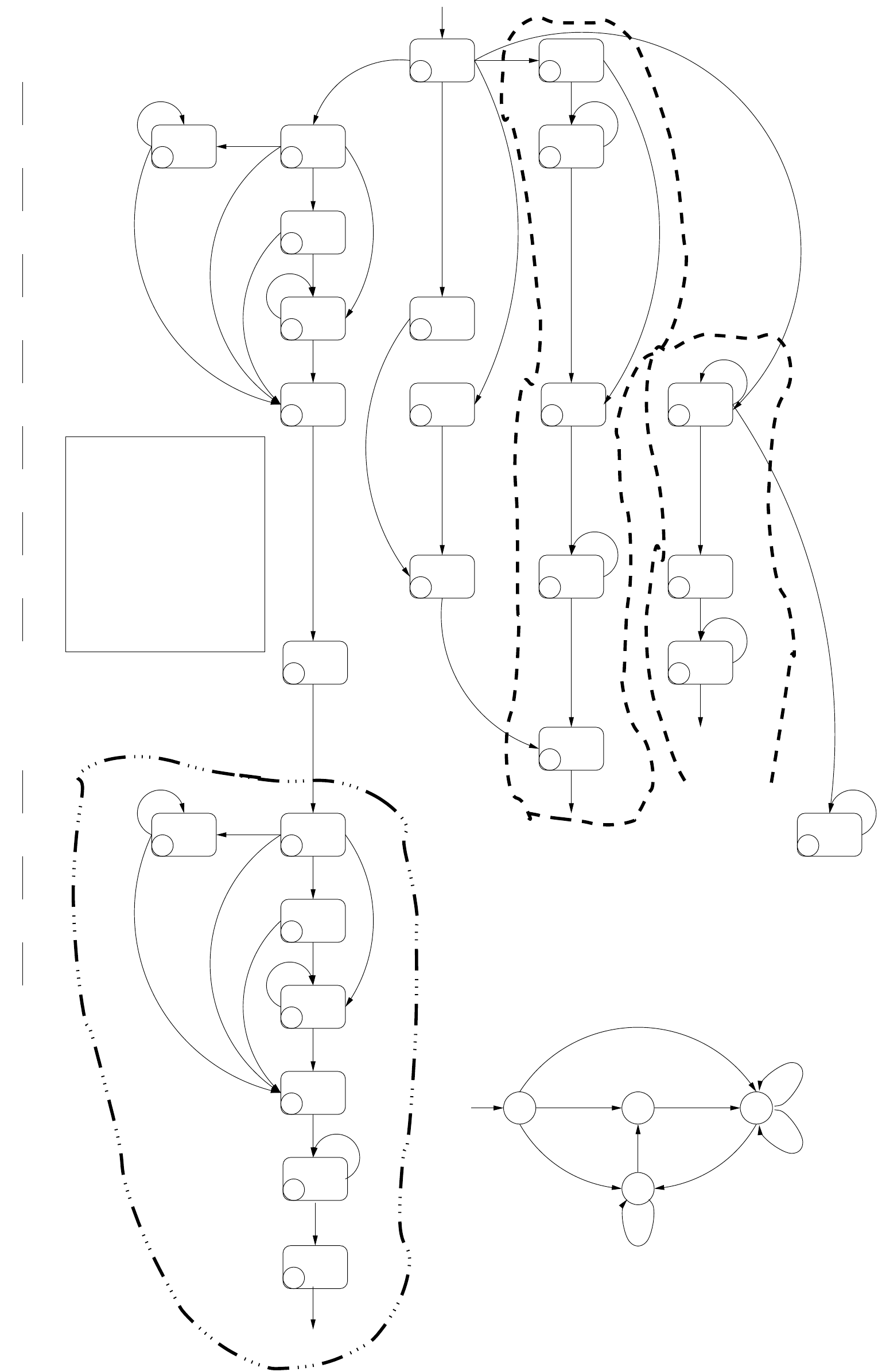 }
		\caption{(a) The simplified $A \in \ntaeps$; (b) $\IARRR(A)$ with period $6$}
		\label{fig:complex}
	\end{figure}
	\begin{figure}[t]
		\centering
		\scalebox{0.4}{ 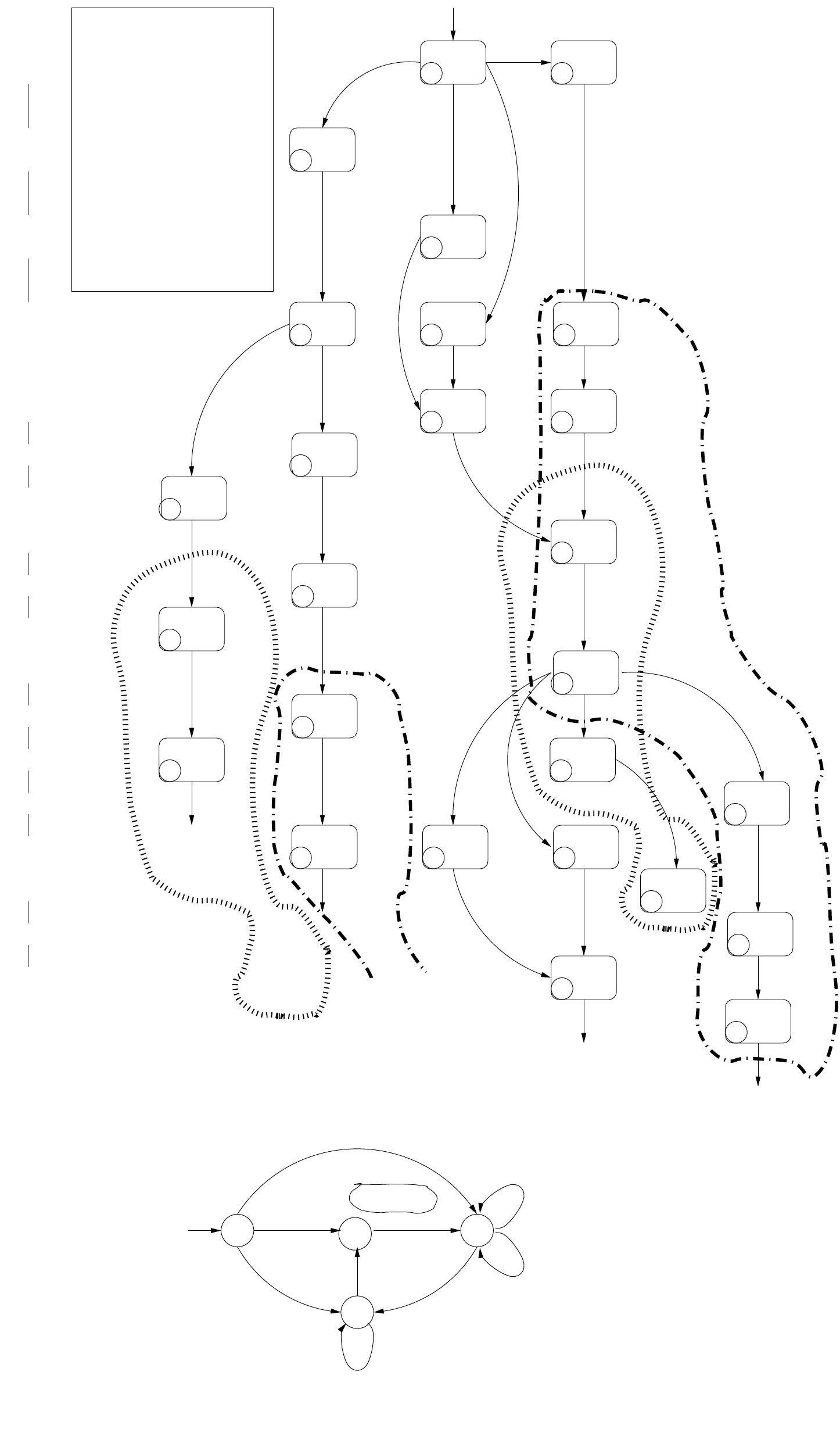 }
		\caption{(a) The original $A \in \ntaeps$; (b) The additional part of $\IARRR(A)$ with cycles of lengths $11$ and $5$}
		\label{fig:complex2}
	\end{figure}
\end{example}

As is known, a TA may be totally non-periodic in the sense that no single timed trace of it is eventually periodic (see Example~\ref{ex:non-period}).
However, a special kind of periodicity, which we call \emph{suffix-periodicity}, holds between different timed traces, as shown in the following theorem.
\begin{theorem}
\label{th:lang_eventual_period}
If $A \in \ntaeps$ is not bounded in time then its language $\LLL(A)$ is suffix-periodic:
if $t_r > t_\textnormal{per}$ and
$$\lambda = (t_1, a_1),\ldots,  (t_{r-1}, a_{r-1}), (t_r, a_r), (t_{r+1}, a_{r+1}),\ldots, (t_{r+m}, a_{r+m})$$
is an observable timed trace of $\LLL(A)$ then, for each $k \in L\Integers$, if $t_r + k > t_\textnormal{per}$ then there exists an observable timed trace
$\lambda' \in \LLL(A)$ such that $$\lambda' = (t'_1, a'_1),\ldots,  (t'_s, a'_s), (t_{r}+k, a_r), (t_{r+1}+k, a_{r+1}),\ldots, (t_{r+m}+k, a_{r+m}).$$
\end{theorem}
\begin{proof}
Suppose that $\lambda$ is the observable timed trace of some run $\varrho$ of $A$.
This run corresponds to a path in $\IARRR(A)$ whose $r$-th transition reaches a vertex $v$ with some time-region $\alpha$ with $t_r \in \alpha$. 
By Theorem~\ref{th:eventual_period} there exists a path $\gamma$ in $\IARRR(A)$ which reaches a vertex $u=v+k$.
That is, if $v  = (q, (n_0, n_1, \ldots, n_s), \Delta)$ then $u$ is identical to $v$ except for the integral part of $t$, which is increased by $k$: 
$u =  (q, (n_0 + k, n_1, \ldots, n_s), \Delta)$, or, in other words, the time-region $\alpha'$ of $u$ is a translate by $k$ of the time-region $\alpha$ of $v$.
Hence, since $t_r \in \alpha$ then $t_r + k \in \alpha'$.
As we saw in Section~\ref{sec:sing_path}, the trail of the path $\gamma$ (the union of the trajectories along $\gamma$) is composed of regions in the form of simplices.
Thus, for every value of the time-region $\alpha'$, in particular for $t_r + k$,   there exists a run $\varrho'$ of $A$ which reaches location $q$ at the exact time $t_{r}+k \in \alpha'$ on an observable action $a_r$.
From that time on, the run $\varrho'$ can imitate the behavior of $\varrho$ by keeping a time difference $k$ in the taken transitions.
The result then follows.
\end{proof}
\section{Periodic Augmented Region Automaton}
\label{sec:APARA}
After revealing the periodic structure of $\IARRR(A)$, it is natural to fold it into a finite graph according to this period, which we call periodic augmented region automaton and denote by $\PARRR(A)$.
The construction of $\PARRR(A)$ is done by first taking the subgraph of $\IARRR(A)$ of time $t < t_\textnormal{per} + L$ and then folding the infinite subgraph of $\IARRR(A)$ of time $t \geq t_\textnormal{per} + L$ onto the subgraph of time $t_\textnormal{per} \leq t < t_\textnormal{per} + L$, which becomes the periodic subgraph.
Thus, each vertex of the periodic subgraph represents infinitely-many vertices of $\IARRR(A)$.
Similarly, the out-going edges of the periodic subgraph are periodic edges.
In addition, some of the edges of $\PARRR(A)$ are marked with ($*$) or ($*+$), as explained below.
For an edge $e$, we denote by $\iota(e)$ and $\tau(e)$ the initial, resp. terminal, vertex of $e$. 
\begin{definition}[Periodic augmented region automaton]
	\label{def:per_aug_region_automaton}
	Given an infinite augmented region automaton $\IARRR(A)$ with period $L$ and periodicity starting time $t_\textnormal{per}$, a finite projection $p(\IARRR(A))$ of it, called \emph{periodic augmented region automaton} and denoted $\PARRR(A)$, is a tuple $(V, v_0, E, \eActions, B)$, where:
	\begin{enumerate}
		\item $V$ is the set of vertices, with $v_0 = (q_0, {\bf 0}, {\bf 0})$ the initial vertex.
		For each $v \in \PARRR(A)$, if $u \in p^{-1}(v) \subseteq \IARRR(A)$ then $u$ equals $v$ in all fields, except possibly for the integral part of $t$.
		If $v.\lfloor t \rfloor < t_\textnormal{per}$ then $u=v$ and $v$ is a \emph{regular} vertex.
		Otherwise, $v$ is a \emph{periodic} vertex, $v.\lfloor t \rfloor$ is written as $n + L \ZNaturals$, for some $t_\textnormal{per} \leq n < t_\textnormal{per} + L$, $p^{-1}(v)$ is infinite and $\{u.\lfloor t \rfloor \, | \, p(u)=v\} = \{n+kL \, | \, k = 0, 1, 2, \ldots \}$.
		\item $E$ is the set of edges, which are the projected edges of $\IARRR(A)$ under the map $p$.
		Each edge joining two vertices of $\IARRR(A)$ is mapped to an edge with the same action label that joins the projected vertices. 
		Some of the edges are marked with a symbol of $B = \{(*),(*+)\}$.
		The description below is technical and refers to the different types of edges that occur when folding $\IARRR(A)$: whether the source of the edge is a regular \upshape{(\textbf{R})} or a periodic \upshape{(\textbf{P})} vertex (in the latter case the preimage in $\IARRR(A)$ contains infinitely-many edges, one from each of the preimage vertices), whether it is unmarked \upshape{(\textbf{U})} or marked \upshape{(\textbf{M})} (in the latter case there are infinitely-many edges starting from each of the vertices in the preimage source vertices), and finally the plus sign (\textbf{+}) represents the case where in the preimage the target vertices are not of value $n$ but $n+L$.
		\begin{itemize}
			\item \textnormal{\textbf{UR :}} (unmarked, regular) If $e \in \PARRR(A)$ is unmarked and $\iota(e)$ is regular then $\iota(e).\lfloor t \rfloor = n_1 < t_\textnormal{per}$, $\tau(e).\lfloor t \rfloor = n_2$ or   $\tau(e).\lfloor t \rfloor = n_2 + L \ZNaturals$ and $p^{-1}(e)=\{e'\}$, with $\iota(e').\lfloor t \rfloor=n_1$ and $\tau(e').\lfloor t \rfloor=n_2$.
			\item \textnormal{\textbf{UP :}} (unmarked, periodic) If $e \in \PARRR(A)$ is unmarked and $\iota(e)$ is periodic then $\iota(e).\lfloor t \rfloor = n_1 + L\ZNaturals$, $\tau(e).\lfloor t \rfloor = n_2 + L\ZNaturals$, $t_\textnormal{per} \leq n_1, n_2 < t_\textnormal{per} + L$ and the preimage of $e$ in $\IARRR(A)$ are the infinitely-many edges satisfying the following.
			If $n_1 \leq n_2$ then $p^{-1}(e) = \{e' \,|\, \iota(e').\lfloor t \rfloor = n_1 +kL, \tau(e').\lfloor t \rfloor = n_2 +kL, k=0,1,2,\ldots\}$, 
			and if $n_1 > n_2$ then $p^{-1}(e) = \{e' \,|\, \iota(e').\lfloor t \rfloor = n_1 +kL, \tau(e').\lfloor t \rfloor = n_2 +(k+1)L, k=0,1,2,\ldots\}$.
			\item \textnormal{\textbf{MR :}} (marked, regular) If $e \in \PARRR(A)$ is marked with $`(*)`$ and $\iota(e)$ is regular, with $\iota(e).\lfloor t \rfloor = n_1$ and $\tau(e).\lfloor t \rfloor = n_2$ or   $n_2 + L \ZNaturals$, then $p^{-1}(e) = \{e' \,|\, \iota(e').\lfloor t \rfloor = n_1, \tau(e').\lfloor t \rfloor = n_2 +kL, k=0,1,2,\ldots\}$, that is, infinitely-many edges starting from the same vertex.
			\item \textnormal{\textbf{MP :}} (marked, periodic) If $e \in \PARRR(A)$ is marked with $`(*)`$ and $\iota(e)$ is periodic, with $\iota(e).\lfloor t \rfloor = n_1 + L \ZNaturals$ and $\tau(e).\lfloor t \rfloor = n_2 + L \ZNaturals$, then its preimage in $\IARRR(A)$ contains all the edges according to both rules \textnormal{\textbf{UP}} and \textnormal{\textbf{MR}}. 
		\item \textnormal{\textbf{MP+ :}} (marked, periodic, shifted) If $e \in \PARRR(A)$ is marked with $`(*+)`$
		then the same rules that apply to an edge marked with $`(*)`$ hold, except that the target vertices are of $L$-shift in time compared to those of an edge marked with $`(*)`$. 
		\end{itemize}
		\item $\eActions = \Sigma \cup \{\epsilon\}$ is the finite set of actions.
	\end{enumerate}
\end{definition}
We remark that instead of periodic time interval of type $[a,b)$ we can define it analogously to be of type $(a,b]$ as in Fig.~\ref{fig:APTA_a}(d), where the periodic time is $(0,1]$.
\begin{example}
	\label{ex:non-period}
	The TA shown in Fig.~\ref{fig:ta3}(a) is taken from \cite{ta}, where it demonstrates non-periodicity: the time difference between an $a$-transition and the following $b$-transition is strictly decreasing along a run.
	However, the periodicity among the collection of timed traces is seen in the periodic augmented region automaton, where the period here is of size 1, and the vertices in times $(2,3)+\ZNaturals$ and $3+\ZNaturals$ are periodic.
	Notice also that there are edges marked with ($*$) which represent infinitely-many edges with the same source.
	\begin{figure}[htb]
		\centering
		\scalebox{0.5}{ 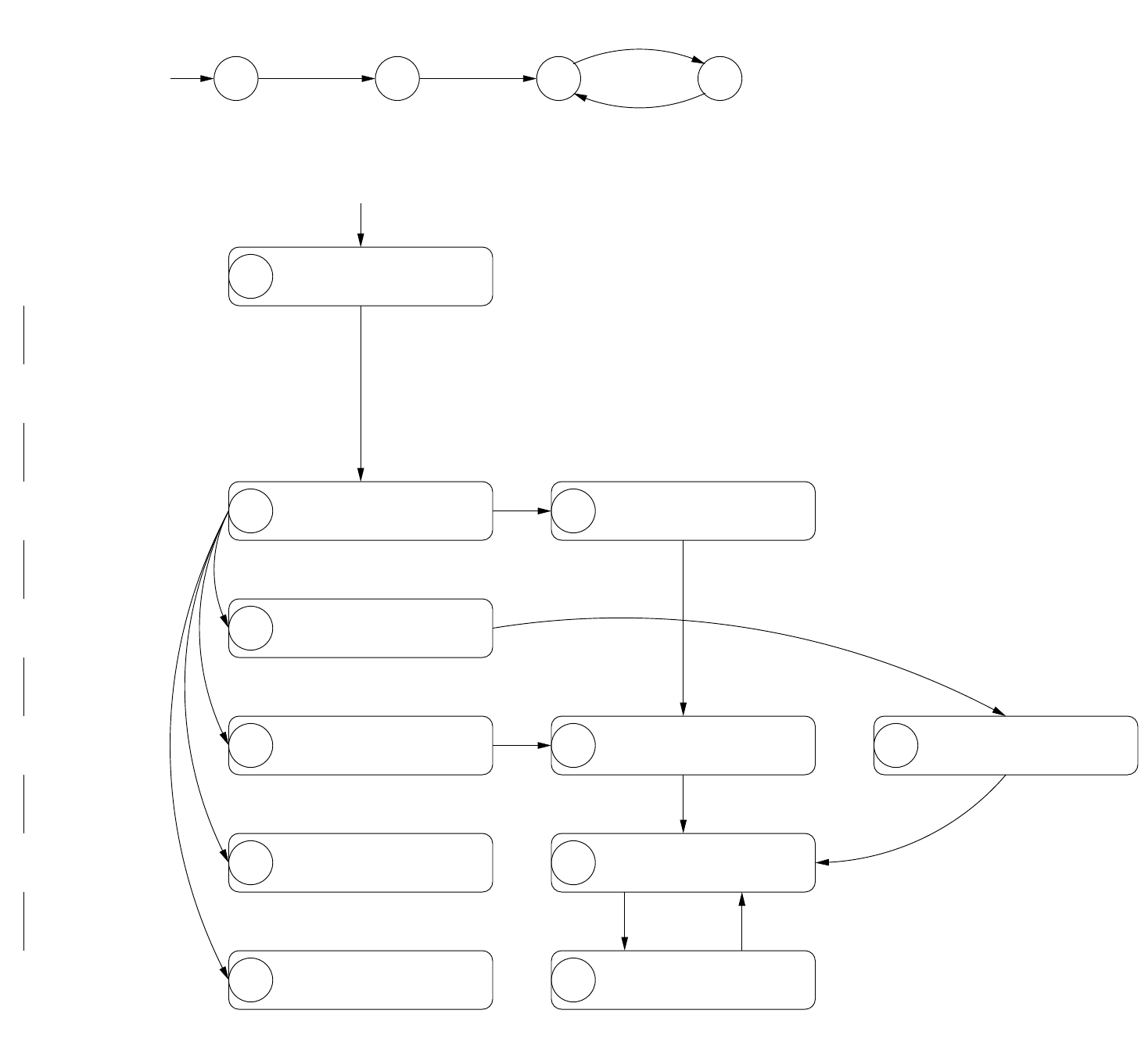 }
		\caption{a) $A \in \ta$ ; b) $\PARRR(A)$, a periodic augmented region automaton of $A$}
		\label{fig:ta3}
	\end{figure}
\end{example} 
\begin{proposition}
	\label{prop:well_def}
	$\PARRR(A)$ is well-defined and as informative as $\IARRR(A)$.
\end{proposition}
\begin{proof}
	Clearly, since $\PARRR(A)$ may be obtained from $\IARRR(A)$ then it cannot be more informative.
	It suffices then to show that for each positive integer $n$, $\IARRR(A)$ can be effectively constructed from $\PARRR(A)$ up to time $t=n$.
	Well, for $t < t_\textnormal{per}$, $\PARRR(A)$ is identical to $\IARRR(A)$.
	Then, by Theorem~\ref{th:eventual_period}, the graph of $\IARRR(A)$ becomes periodic in the sense that the subgraph of time $t_\textnormal{per} \leq t < t_\textnormal{per} + L$ repeats itself, except for the integral part of $t$, which progresses indefinitely in $\IARRR(A)$ but can be expressed modulo the period $L$, as is done in $\PARRR(A)$.
	Indeed, since the transitions in $A$ do not rely on $t$, by taking the quotient of $\IARRR(A)$ modulo $L$ from time $t \geq t_\textnormal{per}$, the only loss of information is the exact time difference in $t$ between the target and source regions.
	But due to the periodicity in $\IARRR(A)$, this information can be finitely presented.
	Hence, since the edges of $\IARRR(A)$ whose initial vertices are of time $t \geq t_\textnormal{per} + L$ are translates of similar edges that start at time $t_\textnormal{per} \leq t < t_\textnormal{per} + L$, it suffices to examine the latter.
		
	So, let $e$ be an edge of $\IARRR(A)$ which joins a vertex $u$ of integral time $\lfloor t \rfloor = n_1$, $t_\textnormal{per} \leq n_1 < t_\textnormal{per} + L$, with a vertex $v$ of integral time 
	$\lfloor t \rfloor = n_2$, and suppose that $n_1 \leq n_2 \mod L$.
	Suppose also that $u$ is not joined to a vertex $v' = v-L$. 
	Then, since $L > M$ there are only two cases: either $n_2 - n_1 < L$ or $L \leq n_2 - n_1 < 2L$.
	In order to distinguish between these cases, the latter case is marked by a plus sign that is added to the corresponding edge of $\PARRR(A)$ from a vertex of integral time $\lfloor t \rfloor = n_1 + L\ZNaturals$ to a vertex of integral time $\lfloor t \rfloor = t_\textnormal{per} + ((n_2 - t_\textnormal{per}) \mod L) + L\ZNaturals$.
	When $u$ is also connected to a vertex $v' = v-L$ then we let $v"$ be of minimal integral time modulo $L$ to which $u$ is connected, that is $v" = v-iL$, for some $i > 0$, and there is no edge from $u$ to $v" - L$ (here $v, v'$ ans $v"$ are identical except for the integral time of $t$).
	If $v"$ is of integral time $n$ then necessarily $u$ is connected to infinitely-many vertices of integral time $n + kL$, $k \geq 0$, and all these edges are captured in $\PARRR(A)$ by marking with $(*)$ the edge from the corresponding vertex of integral time $\lfloor t \rfloor = n_1 + L\ZNaturals$ to a vertex of integral time $\lfloor t \rfloor = t_\textnormal{per} +((n - t_\textnormal{per}) \mod L) + L\ZNaturals$. 
	
	The case where $n_1 > n_2 \mod L$ is handled similarly.   
	It is now clear that in order to construct $\IARRR(A)$ up to time $t=n$ we only need to unfold $\PARRR(A)$ up to this time by obeying the above rules.
\end{proof}
\subsection{Complexity}
Let $N = N(\ARRR(A))$ denote the number of vertices in the augmented region automaton $\ARRR(A)$.
If $\kappa$ denotes the number of clocks, including the absolute clock $t$, $\lambda$ the number of locations in $A$ and $\mu=M+2$, where $M$ is the maximal integer appearing in a guard of $A$, then
\begin{equation}
N \leq \lambda (2\mu)^{\kappa} \kappa!.
\end{equation}
Indeed, the number of combinations of the integral values of the clocks is bounded by $\mu^{\kappa}$ (in fact, $t$ is assigned a single value), there are $\kappa!$ different orderings of the fractional parts of the clocks $\{x_i\}$, and the term $2^{\kappa}$ refers to all possibilities of inequality or equality between each pair of adjacent $\{x_i\}, \{x_j\}$ in an ordering.

Let us look now at the number of vertices in $\PARRR(A)$.
At each time-level the number of vertices is bounded by $N$.
Since $t_{\textnormal{nz}} \leq MN$ then there are at most $MN^2$ vertices of time $t \leq t_{\textnormal{nz}}$.
After passing $t_{\textnormal{nz}}$ we have the subgraphs $\bar{G}_k$ of time length $L$, where $L$ is the period.
Each such subgraph has at most $N L$ vertices.
Since the number of vertices in the subgraphs forms an almost increasing sequence (until an equality occurs two consecutive times), the number of vertices from time $t_{\textnormal{nz}}$ to time $t_\textnormal{per}$ is bounded by $(N L)^2$. 
Thus, the number $N(\PARRR(A))$ of vertices in $\PARRR(A)$ satisfies
\begin{equation}
N(\PARRR(A)) \leq (L^2 + M)N^2 (1+o(1))
\label{eq:nr_vertices}
\end{equation}
as $N \to \infty$.

The largest factor in \eqref{eq:nr_vertices} may come from the period $L$, so let us compute an upper bound of $L$.
$L$ is the least common multiple of the durations $d(\pi)$ of cycles that form a covering set of non-Zeno cycles.
For each such cycle $\pi$, $d(\pi) \leq MN$ since  the length of a simple cycle is bounded by the number $N$ of vertices in $\ARRR(A)$ and the time difference between two vertices along a path is at most $M$.
Thus, a bound on $L$ is given by the least common multiple of $1,2, \ldots, MN$, which is by the prime number theorem
\begin{equation}
L \leq \lcm(1,2, \ldots, MN) = e^{MN(1+o(1))}
\label{eq:L_bound}
\end{equation}
as $MN \to \infty$.

\begin{example}
	When computing the period $L$, in the worst case the numbers $d(\pi)$ are pairwise prime and the vertices of the cycles $\pi$ form a disjoint union of sets which (almost) covers the set of vertices of $\ARRR(A)$.
	So, suppose that $\ARRR(A)$ is in the form of $n$ simple cycles, where each cycle is connected to the initial vertex by an additional edge.
	Suppose also that the length of cycle $i$ is $p_i$, the $i$-th prime number, $i = 1, \ldots, n$.
	Let us assume that $M=1$ and each edge is of weight 1.
	The number of vertices in $\ARRR(A)$ is $N = 1 + \sum_{i = 1}^{n} p_i \sim (1/2) n^2 \log n$.
	Then $L= \lcm \{ p_1, \ldots, p_n \} = \prod_{i = 1}^{n} p_i = e^{n \log n (1+o(1))}$, the primorial $p_n\#$.
	This upper bound is closer to $e^{M\sqrt{N}}$ than to the bound $e^{MN}$ of \eqref{eq:L_bound}. 
\end{example}
\section{The Timestamp}
\label{sec:timestamp}
Recall that the timestamp $\TTTSSS(A)$ of a timed automaton $A$ is the set of all pairs $(t,a)$, such that an observable transition with action $a$ occurs at time $t$ in some run of $A$.
\begin{theorem}
\label{th:timestamp_eventual_period}
The timestamp of a TA $A$ is a union of action-labeled integral points and open unit intervals with integral end-points.
It is either finite or forms an eventually periodic (with respect to time $t$) subset of $\PReals \times \Actions$ and is effectively computable.
\end{theorem}
\begin{proof}
	By Theorem~\ref{th:lang_eventual_period}, if the timestamp is not finite then it becomes periodic, with period $L$, after time $t = t_\textnormal{per}$.
	Thus, if it can effectively be computed up to time $t_\textnormal{per}+L$, then in order to find whether there is an observable transition with action $a$ at time $t_\textnormal{per}+L+t$ one only needs to check the timestamp at time $t_\textnormal{per}+ (t \mod L)$.
	
	By Proposition~\ref{pr:mult_ev_timestamp}, the timestamp up to time $t_\textnormal{per}+L$ is a finite number of labeled integral points and open intervals between integral points and by Proposition~\ref{prop:finite_IARA}, it is effectively computable.	
\end{proof}

The timestamp of a TA is an abstraction of its language: it does not preserve the timestamps of single timed traces.
However, the timestamp is eventually periodic and computable, hence the timestamp inclusion problem is decidable.
Thus, due to the general undecidability of the language inclusion problem in non-deterministic timed automata, one may use the timestamp for refutation purpose.
\begin{corollary}
Given two timed automata $A,B \in \ntaeps$ over the same alphabet (action labels), the question of non-inclusion of their timestamps is decidable, 
thus providing a decidable sufficient condition for the (in general, undecidable) question of non-inclusion of their languages: $\LLL(A) \nsubseteq \LLL(B)$.
\end{corollary}

The timestamp is easily extracted from $\PARRR$ (in fact, it is enough to take the subgraph of $\IARRR$ up to level $t_\textnormal{per}+ L$).
We just form the union of the time-regions up to level $t_\textnormal{per}+ L$, where each time-region is either a point $\{n\}$ or an open interval $(n,n+1)$, along with the labels of the actions of the in-going edges.
The timestamp in the interval $t_\textnormal{per} \leq t < t_\textnormal{per}+ L$ then repeats itself indefinitely. 
\begin{definition}
	For each $a \in \Actions$, let $A_a$ be the restriction of $A$ to $a$-actions, obtained by substituting each $b \in \Actions \smallsetminus \{a\}$ with $\epsilon$, representing the silent transition.
\end{definition}
Thus, the language of $A_a$ is the 'censored' language of $A$, which is the outcome of deleting from each word (timed trace) all pairs $(b,t)$, $b \neq a$.
\begin{example}
	The timestamp of the $a$-transitions of the automaton of Fig.~\ref{fig:ta3} is $\TTTSSS(A_a) = \Naturals$, and that of the $b$-transitions is $\TTTSSS(A_b) = [1,\infty)$.
\end{example}
\subsection{Timestamp Automaton}
Given a TA $A$, one can effectively construct a deterministic TA $\tilde{A}$, called a \emph{timestamp automaton} of $A$ with the same timestamp as that of $A$.
Such as automaton is decomposable into the timestamp automata of the automata $A_a$.
\begin{definition}[Timestamp automaton]
	Given a timed automaton $A \in \ntaeps$, a timestamp automaton $\tilde{A}$ is a deterministic (finite) timed automaton with a single clock and with timestamp identical to that of $A$.
	It is the union of the timestamp automata $\tilde{A}_a$, $a \in \Actions$, having a common initial vertex.
	Each $\tilde{A}_a$ has the form of a single path $\tilde{\gamma}_a$ of positive length, which may end in a loop $\tilde{\pi}_a$, thus giving $\tilde{A}$ the form of a bouquet.
\end{definition}
\begin{theorem}
	\label{th:timestamp_contruct}
	Given a timed automaton $A \in \ntaeps$, one can effectively construct a timestamp automaton $\tilde{A}$.
\end{theorem}
\begin{proof}
We construct $\tilde{A}_a$ by following the ordered connected components (intervals) of the timestamp $\TTTSSS(A_a)$ (here 'interval' includes also singletons $\{n\}$).
To each such time interval corresponds the next transition guard in $\tilde{\gamma}_a$, where the lower and upper constraint on the clock $x$ in the transition guard are exactly the left and right end-points of the interval.
In case $\TTTSSS(A_a)$ contains a finite number of intervals (possibly the last interval of infinite length) then we are done.

Otherwise, $\TTTSSS(A_a)$ contains infinitely-many intervals, which form an eventually periodic sequence with respect to the sizes of the intervals and the distances among them. 
Then we need to attach a loop at the end of $\tilde{\gamma}_a$.
We distinguish between two cases.

Case (i): The periodic part of $\TTTSSS(A_a)$ contains an integral point $n$ (not necessarily as an isolated point).
Then we first split the interval, say $[a,b)$, to which $n$ belongs into disjoint intervals $[a,n), \{n\}, (n,b)$, such that the point $n$ belongs to a singleton. 
Then we extend $\tilde{\gamma}_a$ until reaching $\{n\}$, so that the last transition of $\tilde{\gamma}_a$ is constrained to $x = n$ while resetting $x$. 
From that point begins the loop $\tilde{\pi}_a$, which obeys the same rules as applied to $\tilde{\gamma}_a$, with $x$ being reset only when finishing the loop (see Fig.~\ref{fig:ts_autom} (a)).

Case (ii): The periodic part of $\TTTSSS(A_a)$ does not contain an integral point, that is, it is a union of open unit intervals $(n,n+1)$.
Then, if necessary, we split the last interval before starting the loop into two with the second component a unit open interval (we know that this last interval is not a singleton).
This unit interval refers to the last transition of $\tilde{\gamma}_a$ and we reset $x$ on that transition. 
Then, all transitions within the loop $\tilde{\pi}_a$ are forced to occur at integral times, with $x$ being reset when completing the whole loop (see Fig.~\ref{fig:ts_autom} (b)) (hence, in both cases the clock $x$ is reset in each $\tilde{A}_a$ only on a transition to the vertex $v_a \in \tilde{\gamma}_a \cap \tilde{\pi}_a$).
The idea is that if we enter the loop at a fractional time, say $c = 0.3$, then all the next transitions will take place at times $n + 0.3$, but since $c$ can be arbitrarily chosen within the open interval $(0,1)$ then the set of all runs will cover the entire timestamp.
\end{proof}
%
\begin{example}
\label{ex:ts_autom}
Let $A$ be a TA with timestamp
$$
\begin{array}{l}
\TTTSSS(A_a) = (1,3] \cup \{5\} \cup (6 + ([0,2) \cup \{3\} \cup (8,18)) + 21\ZNaturals) \times \{ a \}, \\
\TTTSSS(A_b) = [0,1] \cup (2,4) \cup \{5\} \cup (6 + ((0,1) \cup (1,2) \cup (5,6) \cup (8,9)) + 10\ZNaturals) \\
\hspace{17 mm} \times \{ b \}, \\
\TTTSSS(A_c) = [1,4] \cup \{6\} \cup (10, \infty) \times \{ c \}.
\end{array}
$$
Then a possible timestamp automaton of $A$ is given in Fig.~\ref{fig:ts_autom}. 
\begin{figure}[htb]
\centering
\scalebox{0.5}{ 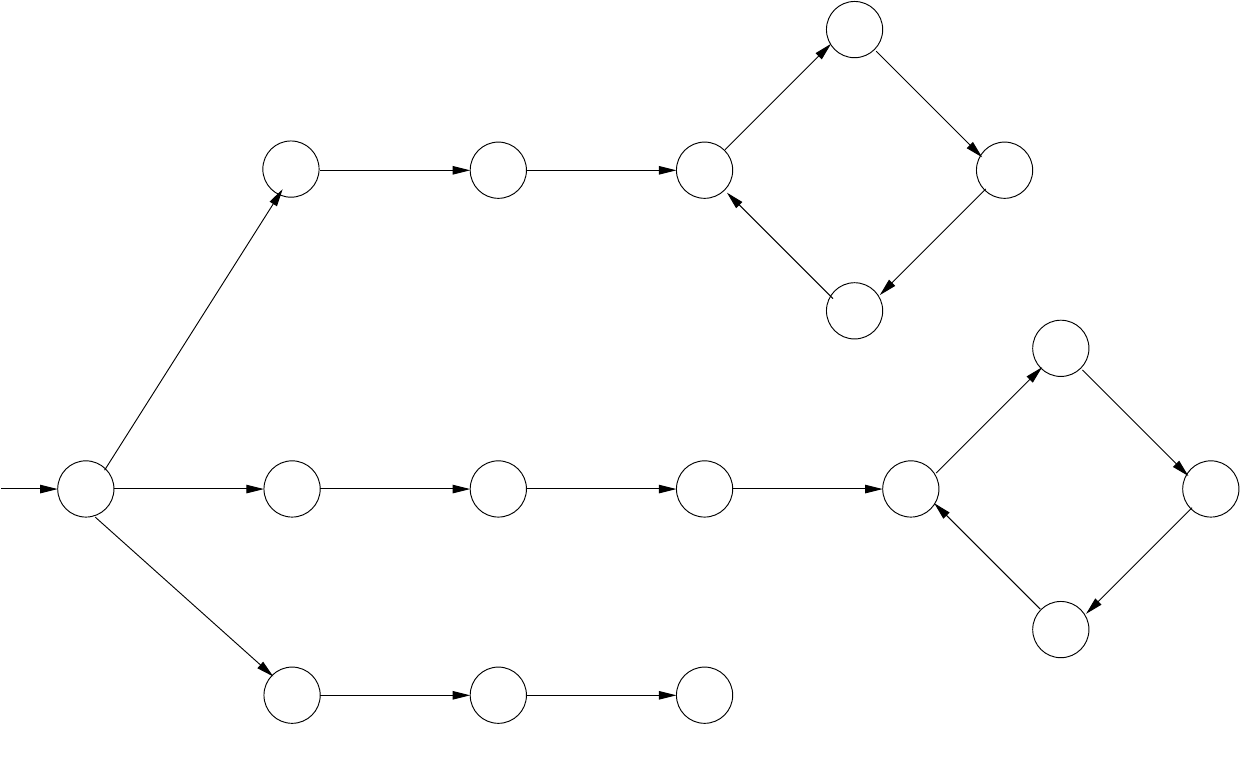 }
\caption{Timestamp automata of a) $\TTTSSS(A_a)$; b) $\TTTSSS(A_b)$; c) $\TTTSSS(A_c)$}
\label{fig:ts_autom}
\end{figure}
\end{example} 
\begin{figure}[htb]
\centering
\scalebox{0.5}{ 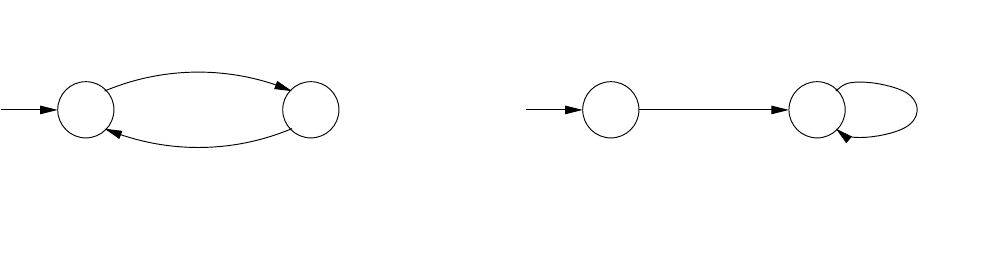 }
\caption{a) A non-determinizable $A \in \ntaeps$ ; b) A timestamp automaton $\tilde{A}$}
\label{fig:time_stamp1}
\end{figure} 
\begin{example}
The language of the TA $A \in \ntaeps$ of Fig.~\ref{fig:time_stamp1} (a) is
$$\LLL(A) = \{ (t_0, a), (t_1, a),\ldots, (t_n, a)\, | \, i < t_i < i+1, i=0,\ldots,n-1, n \in \ZNaturals \}$$ (supposing all locations are `accepting').
The timestamp of $A$ is the set of all positive non-integral reals: $\TTTSSS(A) = \PReals \smallsetminus \ZNaturals$.
$A$ is not determinizable. Each transition occurs between the next pair of successive natural numbers.
The guard of each such transition must refer to a clock which was reset on some previous integral time.
But since all transitions occur on non-integral time, the only clock that can be referred to is a clock $x$ that is reset at time $0$ and hence the transition guards need to be of the form $n < x < n+1$ for each $n \in \ZNaturals$, which makes the automaton infinite.
Nevertheless, the timestamp automaton associated with $A$, seen in Fig.~\ref{fig:time_stamp1} (b), is deterministic.
\end{example}
\section{Conclusion and Future Research}
The timestamp of a non-deterministic timed automaton with silent transitions ($\ntaeps$) consists of the set of all action-labeled times at which locations can be reached by observable transitions.
The problem of computing the timestamp is a generalization of the basic reachability problem, a fundamental problem in model checking, thus being of interest from the theoretical as well as from the practical point of view.
In this paper we showed that the timestamp can be effectively computed, also when the timed automata are non-deterministic and include silent transitions.

One of the major problems in testing and verification of abstract models of real-time systems is the inclusion of the language of one timed automaton in the language of another timed automaton.
This problem is, in general, undecidable.
Thus, since (non)-inclusion of timestamps of timed automata is a decidable problem, we have a tool which provides a sufficient condition for language non-inclusion in timed automata. However, the timestamp may be seen as overly abstract since it does not take into account the order in which events occur.
Another property to be considered is complexity.
We did not try to find here an efficient algorithm for the construction of the timestamp, e.g. by replacing regions with time-periodic structures like zones or other symbolic representations \cite{MPS11}
and this can be the subject of possible future research.

\subsection*{Acknowledgements.} 
\begin{small}
	This research was partly supported by the Austrian Science Fund (FWF) Project P29355-N35. 
\end{small}
\bibliography{ta}
\bibliographystyle{amsplain}
\end{document}

%% file: trail_and_ts_path.pdf_t
\begin{picture}(0,0)%
\includegraphics{trail_and_ts_path.pdf}%
\end{picture}%
\setlength{\unitlength}{3947sp}%
\begingroup\makeatletter\ifx\SetFigFont\undefined%
\gdef\SetFigFont#1#2#3#4#5{%
  \reset@font\fontsize{#1}{#2pt}%
  \fontfamily{#3}\fontseries{#4}\fontshape{#5}%
  \selectfont}%
\fi\endgroup%
\begin{picture}(11280,7912)(-4739,-10550)
\put(151,-6436){\makebox(0,0)[lb]{\smash{{\SetFigFont{14}{16.8}{\rmdefault}{\bfdefault}{\updefault}{\color[rgb]{0,0,0}$b$-timestamp}%
}}}}
\put(-2249,-7186){\makebox(0,0)[lb]{\smash{{\SetFigFont{17}{20.4}{\rmdefault}{\mddefault}{\updefault}{\color[rgb]{0,0,0}$(a)$}%
}}}}
\put(1726,-2761){\makebox(0,0)[lb]{\smash{{\SetFigFont{12}{14.4}{\rmdefault}{\mddefault}{\updefault}{\color[rgb]{0,0,0}$x$}%
}}}}
\put(6526,-5236){\makebox(0,0)[lb]{\smash{{\SetFigFont{12}{14.4}{\rmdefault}{\mddefault}{\updefault}{\color[rgb]{0,0,0}$t$}%
}}}}
\put(-3674,-5761){\makebox(0,0)[lb]{\smash{{\SetFigFont{14}{16.8}{\rmdefault}{\mddefault}{\updefault}{\color[rgb]{0,0,0}$1$}%
}}}}
\put(-2024,-5761){\makebox(0,0)[lb]{\smash{{\SetFigFont{14}{16.8}{\rmdefault}{\mddefault}{\updefault}{\color[rgb]{0,0,0}$2$}%
}}}}
\put(-374,-5761){\makebox(0,0)[lb]{\smash{{\SetFigFont{14}{16.8}{\rmdefault}{\mddefault}{\updefault}{\color[rgb]{0,0,0}$3$}%
}}}}
\put(-1649,-5311){\makebox(0,0)[lb]{\smash{{\SetFigFont{14}{16.8}{\rmdefault}{\mddefault}{\updefault}{\color[rgb]{0,0,0}$1 < x < 2$}%
}}}}
\put(-1124,-6211){\makebox(0,0)[lb]{\smash{{\SetFigFont{14}{16.8}{\rmdefault}{\mddefault}{\updefault}{\color[rgb]{0,0,0}$a$}%
}}}}
\put(-3674,-4111){\makebox(0,0)[lb]{\smash{{\SetFigFont{14}{16.8}{\rmdefault}{\mddefault}{\updefault}{\color[rgb]{0,0,0}$0$}%
}}}}
\put(-1199,-5086){\makebox(0,0)[lb]{\smash{{\SetFigFont{14}{16.8}{\rmdefault}{\mddefault}{\updefault}{\color[rgb]{0,0,0}$a$}%
}}}}
\put(4113,-9706){\makebox(0,0)[lb]{\smash{{\SetFigFont{11}{13.2}{\rmdefault}{\mddefault}{\updefault}{\color[rgb]{0,0,0}$t$}%
}}}}
\put(-644,-7361){\makebox(0,0)[lb]{\smash{{\SetFigFont{11}{13.2}{\rmdefault}{\mddefault}{\updefault}{\color[rgb]{0,0,0}$x$}%
}}}}
\put(1351,-10486){\makebox(0,0)[lb]{\smash{{\SetFigFont{17}{20.4}{\rmdefault}{\mddefault}{\updefault}{\color[rgb]{0,0,0}$(c)$}%
}}}}
\put(3751,-7186){\makebox(0,0)[lb]{\smash{{\SetFigFont{17}{20.4}{\rmdefault}{\mddefault}{\updefault}{\color[rgb]{0,0,0}$(b)$}%
}}}}
\put(-2924,-5986){\makebox(0,0)[lb]{\smash{{\SetFigFont{14}{16.8}{\rmdefault}{\mddefault}{\updefault}{\color[rgb]{0,0,0}$b$}%
}}}}
\put(-3524,-6211){\makebox(0,0)[lb]{\smash{{\SetFigFont{14}{16.8}{\rmdefault}{\mddefault}{\updefault}{\color[rgb]{0,0,0}$1 \leq x \leq 3, \{x\}$}%
}}}}
\put(-4724,-4936){\makebox(0,0)[lb]{\smash{{\SetFigFont{14}{16.8}{\rmdefault}{\mddefault}{\updefault}{\color[rgb]{0,0,0}$x = 1, \{x\}$}%
}}}}
\put(-4349,-4711){\makebox(0,0)[lb]{\smash{{\SetFigFont{14}{16.8}{\rmdefault}{\mddefault}{\updefault}{\color[rgb]{0,0,0}$a$}%
}}}}
\put(-1574,-6436){\makebox(0,0)[lb]{\smash{{\SetFigFont{14}{16.8}{\rmdefault}{\mddefault}{\updefault}{\color[rgb]{0,0,0}$x = 3, \{x\}$}%
}}}}
\put(151,-5836){\makebox(0,0)[lb]{\smash{{\SetFigFont{14}{16.8}{\rmdefault}{\bfdefault}{\updefault}{\color[rgb]{0,0,0}$a$-timestamp}%
}}}}
\end{picture}%

%% file: APTA_a.pdf_t
\begin{picture}(0,0)%
\includegraphics{APTA_a.pdf}%
\end{picture}%
\setlength{\unitlength}{3947sp}%
\begingroup\makeatletter\ifx\SetFigFont\undefined%
\gdef\SetFigFont#1#2#3#4#5{%
  \reset@font\fontsize{#1}{#2pt}%
  \fontfamily{#3}\fontseries{#4}\fontshape{#5}%
  \selectfont}%
\fi\endgroup%
\begin{picture}(13077,9300)(1261,-4366)
\put(7876,3764){\makebox(0,0)[lb]{\smash{{\SetFigFont{17}{20.4}{\rmdefault}{\mddefault}{\updefault}{\color[rgb]{0,0,0}$(a)$}%
}}}}
\put(10051,-1111){\makebox(0,0)[lb]{\smash{{\SetFigFont{14}{16.8}{\rmdefault}{\mddefault}{\updefault}{\color[rgb]{0,0,0}$0=\{t\}=\{x\}$}%
}}}}
\put(9676,-961){\makebox(0,0)[lb]{\smash{{\SetFigFont{14}{16.8}{\rmdefault}{\mddefault}{\updefault}{\color[rgb]{0,0,0}$0$}%
}}}}
\put(9676,-2161){\makebox(0,0)[lb]{\smash{{\SetFigFont{14}{16.8}{\rmdefault}{\mddefault}{\updefault}{\color[rgb]{0,0,0}$1$}%
}}}}
\put(10351,-1561){\makebox(0,0)[lb]{\smash{{\SetFigFont{14}{16.8}{\rmdefault}{\mddefault}{\updefault}{\color[rgb]{0,0,0}$a$}%
}}}}
\put(10651,-1561){\makebox(0,0)[lb]{\smash{{\SetFigFont{14}{16.8}{\rmdefault}{\mddefault}{\updefault}{\color[rgb]{0,0,0}$(*)$}%
}}}}
\put(10426,-811){\makebox(0,0)[lb]{\smash{{\SetFigFont{14}{16.8}{\rmdefault}{\mddefault}{\updefault}{\color[rgb]{0,0,0}$(0,0)$}%
}}}}
\put(10351,-2311){\makebox(0,0)[lb]{\smash{{\SetFigFont{14}{16.8}{\rmdefault}{\mddefault}{\updefault}{\color[rgb]{0,0,0}$0<\{t\}$}%
}}}}
\put(11701,-4261){\makebox(0,0)[lb]{\smash{{\SetFigFont{17}{20.4}{\rmdefault}{\mddefault}{\updefault}{\color[rgb]{0,0,0}$(d)$}%
}}}}
\put(11701,-2161){\makebox(0,0)[lb]{\smash{{\SetFigFont{14}{16.8}{\rmdefault}{\mddefault}{\updefault}{\color[rgb]{0,0,0}$a$}%
}}}}
\put(9676,-3361){\makebox(0,0)[lb]{\smash{{\SetFigFont{14}{16.8}{\rmdefault}{\mddefault}{\updefault}{\color[rgb]{0,0,0}$1$}%
}}}}
\put(11851,-811){\makebox(0,0)[lb]{\smash{{\SetFigFont{14}{16.8}{\rmdefault}{\mddefault}{\updefault}{\color[rgb]{0,0,0}$a$}%
}}}}
\put(12751,-1111){\makebox(0,0)[lb]{\smash{{\SetFigFont{14}{16.8}{\rmdefault}{\mddefault}{\updefault}{\color[rgb]{0,0,0}$0=\{t\}=\{x\}$}%
}}}}
\put(12376,-961){\makebox(0,0)[lb]{\smash{{\SetFigFont{14}{16.8}{\rmdefault}{\mddefault}{\updefault}{\color[rgb]{0,0,0}$1$}%
}}}}
\put(13126,-811){\makebox(0,0)[lb]{\smash{{\SetFigFont{14}{16.8}{\rmdefault}{\mddefault}{\updefault}{\color[rgb]{0,0,0}$(0,0)$}%
}}}}
\put(12001,-2161){\makebox(0,0)[lb]{\smash{{\SetFigFont{14}{16.8}{\rmdefault}{\mddefault}{\updefault}{\color[rgb]{0,0,0}$(*)$}%
}}}}
\put(10351,-3511){\makebox(0,0)[lb]{\smash{{\SetFigFont{14}{16.8}{\rmdefault}{\mddefault}{\updefault}{\color[rgb]{0,0,0}$0=\{t\}$}%
}}}}
\put(10201,-2011){\makebox(0,0)[lb]{\smash{{\SetFigFont{14}{16.8}{\rmdefault}{\mddefault}{\updefault}{\color[rgb]{0,0,0}$(0+\Naturals,\top)$}%
}}}}
\put(10201,-3211){\makebox(0,0)[lb]{\smash{{\SetFigFont{14}{16.8}{\rmdefault}{\mddefault}{\updefault}{\color[rgb]{0,0,0}$(1+\Naturals,\top)$}%
}}}}
\put(8326,-2161){\makebox(0,0)[lb]{\smash{{\SetFigFont{14}{16.8}{\rmdefault}{\mddefault}{\updefault}{\color[rgb]{0,0,0}$(0,1)+\Naturals$}%
}}}}
\put(8476,-3361){\makebox(0,0)[lb]{\smash{{\SetFigFont{14}{16.8}{\rmdefault}{\mddefault}{\updefault}{\color[rgb]{0,0,0}$1+\Naturals$}%
}}}}
\put(8476,-961){\makebox(0,0)[lb]{\smash{{\SetFigFont{14}{16.8}{\rmdefault}{\mddefault}{\updefault}{\color[rgb]{0,0,0}$0$}%
}}}}
\put(8476,-286){\makebox(0,0)[lb]{\smash{{\SetFigFont{20}{24.0}{\rmdefault}{\mddefault}{\updefault}{\color[rgb]{0,0,0}$t$}%
}}}}
\put(10051,3014){\makebox(0,0)[lb]{\smash{{\SetFigFont{14}{16.8}{\rmdefault}{\mddefault}{\updefault}{\color[rgb]{0,0,0}$0=\{t\}=\{x\}$}%
}}}}
\put(9676,3164){\makebox(0,0)[lb]{\smash{{\SetFigFont{14}{16.8}{\rmdefault}{\mddefault}{\updefault}{\color[rgb]{0,0,0}$0$}%
}}}}
\put(9676,1964){\makebox(0,0)[lb]{\smash{{\SetFigFont{14}{16.8}{\rmdefault}{\mddefault}{\updefault}{\color[rgb]{0,0,0}$1$}%
}}}}
\put(10351,2564){\makebox(0,0)[lb]{\smash{{\SetFigFont{14}{16.8}{\rmdefault}{\mddefault}{\updefault}{\color[rgb]{0,0,0}$a$}%
}}}}
\put(10651,2564){\makebox(0,0)[lb]{\smash{{\SetFigFont{14}{16.8}{\rmdefault}{\mddefault}{\updefault}{\color[rgb]{0,0,0}$0^*$}%
}}}}
\put(10426,3314){\makebox(0,0)[lb]{\smash{{\SetFigFont{14}{16.8}{\rmdefault}{\mddefault}{\updefault}{\color[rgb]{0,0,0}$(-,0)$}%
}}}}
\put(10351,1814){\makebox(0,0)[lb]{\smash{{\SetFigFont{14}{16.8}{\rmdefault}{\mddefault}{\updefault}{\color[rgb]{0,0,0}$0<\{t\}$}%
}}}}
\put(11851,3314){\makebox(0,0)[lb]{\smash{{\SetFigFont{14}{16.8}{\rmdefault}{\mddefault}{\updefault}{\color[rgb]{0,0,0}$a$}%
}}}}
\put(12751,3014){\makebox(0,0)[lb]{\smash{{\SetFigFont{14}{16.8}{\rmdefault}{\mddefault}{\updefault}{\color[rgb]{0,0,0}$0=\{t\}=\{x\}$}%
}}}}
\put(12376,3164){\makebox(0,0)[lb]{\smash{{\SetFigFont{14}{16.8}{\rmdefault}{\mddefault}{\updefault}{\color[rgb]{0,0,0}$1$}%
}}}}
\put(13126,3314){\makebox(0,0)[lb]{\smash{{\SetFigFont{14}{16.8}{\rmdefault}{\mddefault}{\updefault}{\color[rgb]{0,0,0}$(-,0)$}%
}}}}
\put(10201,2114){\makebox(0,0)[lb]{\smash{{\SetFigFont{14}{16.8}{\rmdefault}{\mddefault}{\updefault}{\color[rgb]{0,0,0}$(-,\top)$}%
}}}}
\put(12376,1964){\makebox(0,0)[lb]{\smash{{\SetFigFont{14}{16.8}{\rmdefault}{\mddefault}{\updefault}{\color[rgb]{0,0,0}$1$}%
}}}}
\put(13051,1814){\makebox(0,0)[lb]{\smash{{\SetFigFont{14}{16.8}{\rmdefault}{\mddefault}{\updefault}{\color[rgb]{0,0,0}$0=\{t\}$}%
}}}}
\put(12901,2114){\makebox(0,0)[lb]{\smash{{\SetFigFont{14}{16.8}{\rmdefault}{\mddefault}{\updefault}{\color[rgb]{0,0,0}$(-,\top)$}%
}}}}
\put(11626,1064){\makebox(0,0)[lb]{\smash{{\SetFigFont{17}{20.4}{\rmdefault}{\mddefault}{\updefault}{\color[rgb]{0,0,0}$(c)$}%
}}}}
\put(11851,3014){\makebox(0,0)[lb]{\smash{{\SetFigFont{14}{16.8}{\rmdefault}{\mddefault}{\updefault}{\color[rgb]{0,0,0}$0$}%
}}}}
\put(12001,2489){\makebox(0,0)[lb]{\smash{{\SetFigFont{14}{16.8}{\rmdefault}{\mddefault}{\updefault}{\color[rgb]{0,0,0}$a$}%
}}}}
\put(12301,2639){\makebox(0,0)[lb]{\smash{{\SetFigFont{14}{16.8}{\rmdefault}{\mddefault}{\updefault}{\color[rgb]{0,0,0}$1^*$}%
}}}}
\put(4051,2564){\makebox(0,0)[lb]{\smash{{\SetFigFont{14}{16.8}{\rmdefault}{\mddefault}{\updefault}{\color[rgb]{0,0,0}$a$}%
}}}}
\put(4126,1364){\makebox(0,0)[lb]{\smash{{\SetFigFont{14}{16.8}{\rmdefault}{\mddefault}{\updefault}{\color[rgb]{0,0,0}$a$}%
}}}}
\put(4126,164){\makebox(0,0)[lb]{\smash{{\SetFigFont{14}{16.8}{\rmdefault}{\mddefault}{\updefault}{\color[rgb]{0,0,0}$a$}%
}}}}
\put(4126,-1036){\makebox(0,0)[lb]{\smash{{\SetFigFont{14}{16.8}{\rmdefault}{\mddefault}{\updefault}{\color[rgb]{0,0,0}$a$}%
}}}}
\put(4126,-2236){\makebox(0,0)[lb]{\smash{{\SetFigFont{14}{16.8}{\rmdefault}{\mddefault}{\updefault}{\color[rgb]{0,0,0}$a$}%
}}}}
\put(4126,-3436){\makebox(0,0)[lb]{\smash{{\SetFigFont{14}{16.8}{\rmdefault}{\mddefault}{\updefault}{\color[rgb]{0,0,0}$a$}%
}}}}
\put(2176,-2836){\makebox(0,0)[lb]{\smash{{\SetFigFont{14}{16.8}{\rmdefault}{\mddefault}{\updefault}{\color[rgb]{0,0,0}$1$}%
}}}}
\put(2176,-1636){\makebox(0,0)[lb]{\smash{{\SetFigFont{14}{16.8}{\rmdefault}{\mddefault}{\updefault}{\color[rgb]{0,0,0}$1$}%
}}}}
\put(2176,764){\makebox(0,0)[lb]{\smash{{\SetFigFont{14}{16.8}{\rmdefault}{\mddefault}{\updefault}{\color[rgb]{0,0,0}$1$}%
}}}}
\put(2176,-436){\makebox(0,0)[lb]{\smash{{\SetFigFont{14}{16.8}{\rmdefault}{\mddefault}{\updefault}{\color[rgb]{0,0,0}$1$}%
}}}}
\put(2926,-2686){\makebox(0,0)[lb]{\smash{{\SetFigFont{14}{16.8}{\rmdefault}{\mddefault}{\updefault}{\color[rgb]{0,0,0}$(2,\top)$}%
}}}}
\put(2926,-1486){\makebox(0,0)[lb]{\smash{{\SetFigFont{14}{16.8}{\rmdefault}{\mddefault}{\updefault}{\color[rgb]{0,0,0}$(2,\top)$}%
}}}}
\put(2926,-286){\makebox(0,0)[lb]{\smash{{\SetFigFont{14}{16.8}{\rmdefault}{\mddefault}{\updefault}{\color[rgb]{0,0,0}$(1,\top)$}%
}}}}
\put(2926,914){\makebox(0,0)[lb]{\smash{{\SetFigFont{14}{16.8}{\rmdefault}{\mddefault}{\updefault}{\color[rgb]{0,0,0}$(1,\top)$}%
}}}}
\put(2176,3164){\makebox(0,0)[lb]{\smash{{\SetFigFont{14}{16.8}{\rmdefault}{\mddefault}{\updefault}{\color[rgb]{0,0,0}$0$}%
}}}}
\put(2851,-1786){\makebox(0,0)[lb]{\smash{{\SetFigFont{14}{16.8}{\rmdefault}{\mddefault}{\updefault}{\color[rgb]{0,0,0}$0=\{t\}$}%
}}}}
\put(2851,614){\makebox(0,0)[lb]{\smash{{\SetFigFont{14}{16.8}{\rmdefault}{\mddefault}{\updefault}{\color[rgb]{0,0,0}$0=\{t\}$}%
}}}}
\put(1276,1964){\makebox(0,0)[lb]{\smash{{\SetFigFont{14}{16.8}{\rmdefault}{\mddefault}{\updefault}{\color[rgb]{0,0,0}$(0,1)$}%
}}}}
\put(1276,-436){\makebox(0,0)[lb]{\smash{{\SetFigFont{14}{16.8}{\rmdefault}{\mddefault}{\updefault}{\color[rgb]{0,0,0}$(1,2)$}%
}}}}
\put(1276,-2836){\makebox(0,0)[lb]{\smash{{\SetFigFont{14}{16.8}{\rmdefault}{\mddefault}{\updefault}{\color[rgb]{0,0,0}$(2,3)$}%
}}}}
\put(2176,1964){\makebox(0,0)[lb]{\smash{{\SetFigFont{14}{16.8}{\rmdefault}{\mddefault}{\updefault}{\color[rgb]{0,0,0}$1$}%
}}}}
\put(4351,3314){\makebox(0,0)[lb]{\smash{{\SetFigFont{14}{16.8}{\rmdefault}{\mddefault}{\updefault}{\color[rgb]{0,0,0}$a$}%
}}}}
\put(5626,3314){\makebox(0,0)[lb]{\smash{{\SetFigFont{14}{16.8}{\rmdefault}{\mddefault}{\updefault}{\color[rgb]{0,0,0}$(0,0)$}%
}}}}
\put(5251,3014){\makebox(0,0)[lb]{\smash{{\SetFigFont{14}{16.8}{\rmdefault}{\mddefault}{\updefault}{\color[rgb]{0,0,0}$0=\{t\}=\{x\}$}%
}}}}
\put(4876,3164){\makebox(0,0)[lb]{\smash{{\SetFigFont{14}{16.8}{\rmdefault}{\mddefault}{\updefault}{\color[rgb]{0,0,0}$1$}%
}}}}
\put(1426,764){\makebox(0,0)[lb]{\smash{{\SetFigFont{14}{16.8}{\rmdefault}{\mddefault}{\updefault}{\color[rgb]{0,0,0}$1$}%
}}}}
\put(1426,-1636){\makebox(0,0)[lb]{\smash{{\SetFigFont{14}{16.8}{\rmdefault}{\mddefault}{\updefault}{\color[rgb]{0,0,0}$2$}%
}}}}
\put(2926,2114){\makebox(0,0)[lb]{\smash{{\SetFigFont{14}{16.8}{\rmdefault}{\mddefault}{\updefault}{\color[rgb]{0,0,0}$(0,\top)$}%
}}}}
\put(2851,1814){\makebox(0,0)[lb]{\smash{{\SetFigFont{14}{16.8}{\rmdefault}{\mddefault}{\updefault}{\color[rgb]{0,0,0}$0<\{t\}$}%
}}}}
\put(2926,3314){\makebox(0,0)[lb]{\smash{{\SetFigFont{14}{16.8}{\rmdefault}{\mddefault}{\updefault}{\color[rgb]{0,0,0}$(0,0)$}%
}}}}
\put(2551,3014){\makebox(0,0)[lb]{\smash{{\SetFigFont{14}{16.8}{\rmdefault}{\mddefault}{\updefault}{\color[rgb]{0,0,0}$0=\{t\}=\{x\}$}%
}}}}
\put(2851,-586){\makebox(0,0)[lb]{\smash{{\SetFigFont{14}{16.8}{\rmdefault}{\mddefault}{\updefault}{\color[rgb]{0,0,0}$0<\{t\}$}%
}}}}
\put(2851,-2986){\makebox(0,0)[lb]{\smash{{\SetFigFont{14}{16.8}{\rmdefault}{\mddefault}{\updefault}{\color[rgb]{0,0,0}$0<\{t\}$}%
}}}}
\put(1426,3164){\makebox(0,0)[lb]{\smash{{\SetFigFont{14}{16.8}{\rmdefault}{\mddefault}{\updefault}{\color[rgb]{0,0,0}$0$}%
}}}}
\put(1426,3839){\makebox(0,0)[lb]{\smash{{\SetFigFont{20}{24.0}{\rmdefault}{\mddefault}{\updefault}{\color[rgb]{0,0,0}$t$}%
}}}}
\put(2926,-4261){\makebox(0,0)[lb]{\smash{{\SetFigFont{17}{20.4}{\rmdefault}{\mddefault}{\updefault}{\color[rgb]{0,0,0}$(b)$}%
}}}}
\put(8926,4289){\makebox(0,0)[lb]{\smash{{\SetFigFont{14}{16.8}{\rmdefault}{\mddefault}{\updefault}{\color[rgb]{0,0,0}$1$}%
}}}}
\put(7276,4289){\makebox(0,0)[lb]{\smash{{\SetFigFont{14}{16.8}{\rmdefault}{\mddefault}{\updefault}{\color[rgb]{0,0,0}$0$}%
}}}}
\put(7801,4514){\makebox(0,0)[lb]{\smash{{\SetFigFont{14}{16.8}{\rmdefault}{\mddefault}{\updefault}{\color[rgb]{0,0,0}$x \geq 0$}%
}}}}
\put(8026,4739){\makebox(0,0)[lb]{\smash{{\SetFigFont{14}{16.8}{\rmdefault}{\mddefault}{\updefault}{\color[rgb]{0,0,0}$a$}%
}}}}
\end{picture}%

%% file: complex.pdf_t
\begin{picture}(0,0)%
\includegraphics{complex.pdf}%
\end{picture}%
\setlength{\unitlength}{3947sp}%
\begingroup\makeatletter\ifx\SetFigFont\undefined%
\gdef\SetFigFont#1#2#3#4#5{%
  \reset@font\fontsize{#1}{#2pt}%
  \fontfamily{#3}\fontseries{#4}\fontshape{#5}%
  \selectfont}%
\fi\endgroup%
\begin{picture}(12224,19133)(4486,-14451)
\put(13201,-7486){\makebox(0,0)[lb]{\smash{{\SetFigFont{14}{16.8}{\rmdefault}{\mddefault}{\updefault}{\color[rgb]{0,0,0}$(b)$}%
}}}}
\put(4576,2564){\makebox(0,0)[lb]{\smash{{\SetFigFont{14}{16.8}{\rmdefault}{\mddefault}{\updefault}{\color[rgb]{0,0,0}$(0,1)$}%
}}}}
\put(4726,1364){\makebox(0,0)[lb]{\smash{{\SetFigFont{14}{16.8}{\rmdefault}{\mddefault}{\updefault}{\color[rgb]{0,0,0}$1$}%
}}}}
\put(4726,3764){\makebox(0,0)[lb]{\smash{{\SetFigFont{14}{16.8}{\rmdefault}{\mddefault}{\updefault}{\color[rgb]{0,0,0}$0$}%
}}}}
\put(4576,164){\makebox(0,0)[lb]{\smash{{\SetFigFont{14}{16.8}{\rmdefault}{\mddefault}{\updefault}{\color[rgb]{0,0,0}$(1,2)$}%
}}}}
\put(4726,-1036){\makebox(0,0)[lb]{\smash{{\SetFigFont{14}{16.8}{\rmdefault}{\mddefault}{\updefault}{\color[rgb]{0,0,0}$2$}%
}}}}
\put(4726,-3436){\makebox(0,0)[lb]{\smash{{\SetFigFont{14}{16.8}{\rmdefault}{\mddefault}{\updefault}{\color[rgb]{0,0,0}$4$}%
}}}}
\put(4726,4439){\makebox(0,0)[lb]{\smash{{\SetFigFont{20}{24.0}{\rmdefault}{\mddefault}{\updefault}{\color[rgb]{0,0,0}$t$}%
}}}}
\put(4576,-4636){\makebox(0,0)[lb]{\smash{{\SetFigFont{14}{16.8}{\rmdefault}{\mddefault}{\updefault}{\color[rgb]{0,0,0}$(4,5)$}%
}}}}
\put(4576,-7036){\makebox(0,0)[lb]{\smash{{\SetFigFont{14}{16.8}{\rmdefault}{\mddefault}{\updefault}{\color[rgb]{0,0,0}$(6,7)$}%
}}}}
\put(4726,-5836){\makebox(0,0)[lb]{\smash{{\SetFigFont{14}{16.8}{\rmdefault}{\mddefault}{\updefault}{\color[rgb]{0,0,0}$6$}%
}}}}
\put(4576,-9436){\makebox(0,0)[lb]{\smash{{\SetFigFont{14}{16.8}{\rmdefault}{\mddefault}{\updefault}{\color[rgb]{0,0,0}$(7,8)$}%
}}}}
\put(4726,-8236){\makebox(0,0)[lb]{\smash{{\SetFigFont{14}{16.8}{\rmdefault}{\mddefault}{\updefault}{\color[rgb]{0,0,0}$7$}%
}}}}
\put(4801,-2236){\makebox(0,0)[lb]{\smash{{\SetFigFont{14}{16.8}{\rmdefault}{\mddefault}{\updefault}{\color[rgb]{0,0,0}$\vdots$}%
}}}}
\put(4801,-5236){\makebox(0,0)[lb]{\smash{{\SetFigFont{14}{16.8}{\rmdefault}{\mddefault}{\updefault}{\color[rgb]{0,0,0}$\vdots$}%
}}}}
\put(4576,-10636){\makebox(0,0)[lb]{\smash{{\SetFigFont{14}{16.8}{\rmdefault}{\mddefault}{\updefault}{\color[rgb]{0,0,0}$(8,9)$}%
}}}}
\put(4501,-11836){\makebox(0,0)[lb]{\smash{{\SetFigFont{14}{16.8}{\rmdefault}{\mddefault}{\updefault}{\color[rgb]{0,0,0}$(10,11)$}%
}}}}
\put(4501,-13036){\makebox(0,0)[lb]{\smash{{\SetFigFont{14}{16.8}{\rmdefault}{\mddefault}{\updefault}{\color[rgb]{0,0,0}$(12,13)$}%
}}}}
\put(4801,-10036){\makebox(0,0)[lb]{\smash{{\SetFigFont{14}{16.8}{\rmdefault}{\mddefault}{\updefault}{\color[rgb]{0,0,0}$\vdots$}%
}}}}
\put(4801,-11236){\makebox(0,0)[lb]{\smash{{\SetFigFont{14}{16.8}{\rmdefault}{\mddefault}{\updefault}{\color[rgb]{0,0,0}$\vdots$}%
}}}}
\put(4801,-12436){\makebox(0,0)[lb]{\smash{{\SetFigFont{14}{16.8}{\rmdefault}{\mddefault}{\updefault}{\color[rgb]{0,0,0}$\vdots$}%
}}}}
\put(15751,-6886){\makebox(0,0)[lb]{\smash{{\SetFigFont{14}{16.8}{\rmdefault}{\mddefault}{\updefault}{\color[rgb]{0,0,0}$(\top,\top)$}%
}}}}
\put(15676,-7186){\makebox(0,0)[lb]{\smash{{\SetFigFont{14}{16.8}{\rmdefault}{\mddefault}{\updefault}{\color[rgb]{0,0,0}$2$}%
}}}}
\put(16051,-7186){\makebox(0,0)[lb]{\smash{{\SetFigFont{14}{16.8}{\rmdefault}{\mddefault}{\updefault}{\color[rgb]{0,0,0}$\{t\}>0$}%
}}}}
\put(16426,-6286){\makebox(0,0)[lb]{\smash{{\SetFigFont{14}{16.8}{\rmdefault}{\mddefault}{\updefault}{\color[rgb]{0,0,0}$c$}%
}}}}
\put(12601,-2236){\makebox(0,0)[lb]{\smash{{\SetFigFont{14}{16.8}{\rmdefault}{\mddefault}{\updefault}{\color[rgb]{0,0,0}$\epsilon$}%
}}}}
\put(13951,-886){\makebox(0,0)[lb]{\smash{{\SetFigFont{14}{16.8}{\rmdefault}{\mddefault}{\updefault}{\color[rgb]{0,0,0}$(0,2)$}%
}}}}
\put(13876,-1186){\makebox(0,0)[lb]{\smash{{\SetFigFont{14}{16.8}{\rmdefault}{\mddefault}{\updefault}{\color[rgb]{0,0,0}$2$}%
}}}}
\put(14251,-1186){\makebox(0,0)[lb]{\smash{{\SetFigFont{14}{16.8}{\rmdefault}{\mddefault}{\updefault}{\color[rgb]{0,0,0}$\Delta_0$}%
}}}}
\put(12181,-886){\makebox(0,0)[lb]{\smash{{\SetFigFont{14}{16.8}{\rmdefault}{\mddefault}{\updefault}{\color[rgb]{0,0,0}$(2,0)$}%
}}}}
\put(12481,-1186){\makebox(0,0)[lb]{\smash{{\SetFigFont{14}{16.8}{\rmdefault}{\mddefault}{\updefault}{\color[rgb]{0,0,0}$\Delta_0$}%
}}}}
\put(12106,-1186){\makebox(0,0)[lb]{\smash{{\SetFigFont{14}{16.8}{\rmdefault}{\mddefault}{\updefault}{\color[rgb]{0,0,0}$1$}%
}}}}
\put(10351,3914){\makebox(0,0)[lb]{\smash{{\SetFigFont{14}{16.8}{\rmdefault}{\mddefault}{\updefault}{\color[rgb]{0,0,0}$(0,0)$}%
}}}}
\put(10276,3614){\makebox(0,0)[lb]{\smash{{\SetFigFont{14}{16.8}{\rmdefault}{\mddefault}{\updefault}{\color[rgb]{0,0,0}$0$}%
}}}}
\put(10651,3614){\makebox(0,0)[lb]{\smash{{\SetFigFont{14}{16.8}{\rmdefault}{\mddefault}{\updefault}{\color[rgb]{0,0,0}$\Delta_0$}%
}}}}
\put(12151,3914){\makebox(0,0)[lb]{\smash{{\SetFigFont{14}{16.8}{\rmdefault}{\mddefault}{\updefault}{\color[rgb]{0,0,0}$(0,0)$}%
}}}}
\put(12076,3614){\makebox(0,0)[lb]{\smash{{\SetFigFont{14}{16.8}{\rmdefault}{\mddefault}{\updefault}{\color[rgb]{0,0,0}$3$}%
}}}}
\put(12451,3614){\makebox(0,0)[lb]{\smash{{\SetFigFont{14}{16.8}{\rmdefault}{\mddefault}{\updefault}{\color[rgb]{0,0,0}$\Delta_0$}%
}}}}
\put(10801,-2236){\makebox(0,0)[lb]{\smash{{\SetFigFont{14}{16.8}{\rmdefault}{\mddefault}{\updefault}{\color[rgb]{0,0,0}$\epsilon$}%
}}}}
\put(10351,-3286){\makebox(0,0)[lb]{\smash{{\SetFigFont{14}{16.8}{\rmdefault}{\mddefault}{\updefault}{\color[rgb]{0,0,0}$(0,4)$}%
}}}}
\put(10651,-3586){\makebox(0,0)[lb]{\smash{{\SetFigFont{14}{16.8}{\rmdefault}{\mddefault}{\updefault}{\color[rgb]{0,0,0}$\Delta_0$}%
}}}}
\put(10276,-3586){\makebox(0,0)[lb]{\smash{{\SetFigFont{14}{16.8}{\rmdefault}{\mddefault}{\updefault}{\color[rgb]{0,0,0}$2$}%
}}}}
\put(10351,314){\makebox(0,0)[lb]{\smash{{\SetFigFont{14}{16.8}{\rmdefault}{\mddefault}{\updefault}{\color[rgb]{0,0,0}$(1,1)$}%
}}}}
\put(10651, 14){\makebox(0,0)[lb]{\smash{{\SetFigFont{14}{16.8}{\rmdefault}{\mddefault}{\updefault}{\color[rgb]{0,0,0}$\Delta_5$}%
}}}}
\put(10276, 14){\makebox(0,0)[lb]{\smash{{\SetFigFont{14}{16.8}{\rmdefault}{\mddefault}{\updefault}{\color[rgb]{0,0,0}$1$}%
}}}}
\put(10351,-886){\makebox(0,0)[lb]{\smash{{\SetFigFont{14}{16.8}{\rmdefault}{\mddefault}{\updefault}{\color[rgb]{0,0,0}$(2,2)$}%
}}}}
\put(10651,-1186){\makebox(0,0)[lb]{\smash{{\SetFigFont{14}{16.8}{\rmdefault}{\mddefault}{\updefault}{\color[rgb]{0,0,0}$\Delta_0$}%
}}}}
\put(10276,-1186){\makebox(0,0)[lb]{\smash{{\SetFigFont{14}{16.8}{\rmdefault}{\mddefault}{\updefault}{\color[rgb]{0,0,0}$1$}%
}}}}
\put(9826,1364){\makebox(0,0)[lb]{\smash{{\SetFigFont{14}{16.8}{\rmdefault}{\mddefault}{\updefault}{\color[rgb]{0,0,0}$a$}%
}}}}
\put(12151,3164){\makebox(0,0)[lb]{\smash{{\SetFigFont{14}{16.8}{\rmdefault}{\mddefault}{\updefault}{\color[rgb]{0,0,0}$a$}%
}}}}
\put(12451,2414){\makebox(0,0)[lb]{\smash{{\SetFigFont{14}{16.8}{\rmdefault}{\mddefault}{\updefault}{\color[rgb]{0,0,0}$\Delta_5$}%
}}}}
\put(12151,2714){\makebox(0,0)[lb]{\smash{{\SetFigFont{14}{16.8}{\rmdefault}{\mddefault}{\updefault}{\color[rgb]{0,0,0}$(0,0)$}%
}}}}
\put(12076,2414){\makebox(0,0)[lb]{\smash{{\SetFigFont{14}{16.8}{\rmdefault}{\mddefault}{\updefault}{\color[rgb]{0,0,0}$3$}%
}}}}
\put(12751,3314){\makebox(0,0)[lb]{\smash{{\SetFigFont{14}{16.8}{\rmdefault}{\mddefault}{\updefault}{\color[rgb]{0,0,0}$a$}%
}}}}
\put(11551,3914){\makebox(0,0)[lb]{\smash{{\SetFigFont{14}{16.8}{\rmdefault}{\mddefault}{\updefault}{\color[rgb]{0,0,0}$\epsilon$}%
}}}}
\put(10801,1964){\makebox(0,0)[lb]{\smash{{\SetFigFont{14}{16.8}{\rmdefault}{\mddefault}{\updefault}{\color[rgb]{0,0,0}$a$}%
}}}}
\put(11851,1364){\makebox(0,0)[lb]{\smash{{\SetFigFont{14}{16.8}{\rmdefault}{\mddefault}{\updefault}{\color[rgb]{0,0,0}$a$}%
}}}}
\put(12601,764){\makebox(0,0)[lb]{\smash{{\SetFigFont{14}{16.8}{\rmdefault}{\mddefault}{\updefault}{\color[rgb]{0,0,0}$\epsilon$}%
}}}}
\put(13801,1364){\makebox(0,0)[lb]{\smash{{\SetFigFont{14}{16.8}{\rmdefault}{\mddefault}{\updefault}{\color[rgb]{0,0,0}$\epsilon$}%
}}}}
\put(14551,-286){\makebox(0,0)[lb]{\smash{{\SetFigFont{14}{16.8}{\rmdefault}{\mddefault}{\updefault}{\color[rgb]{0,0,0}$b$}%
}}}}
\put(8626,-2836){\makebox(0,0)[lb]{\smash{{\SetFigFont{14}{16.8}{\rmdefault}{\mddefault}{\updefault}{\color[rgb]{0,0,0}$\epsilon$}%
}}}}
\put(7126,-7186){\makebox(0,0)[lb]{\smash{{\SetFigFont{14}{16.8}{\rmdefault}{\mddefault}{\updefault}{\color[rgb]{0,0,0}$\Delta_8$}%
}}}}
\put(8851,-7186){\makebox(0,0)[lb]{\smash{{\SetFigFont{14}{16.8}{\rmdefault}{\mddefault}{\updefault}{\color[rgb]{0,0,0}$\Delta_3$}%
}}}}
\put(8551,-6886){\makebox(0,0)[lb]{\smash{{\SetFigFont{14}{16.8}{\rmdefault}{\mddefault}{\updefault}{\color[rgb]{0,0,0}$(0,0)$}%
}}}}
\put(8476,-7186){\makebox(0,0)[lb]{\smash{{\SetFigFont{14}{16.8}{\rmdefault}{\mddefault}{\updefault}{\color[rgb]{0,0,0}$3$}%
}}}}
\put(8551,-8086){\makebox(0,0)[lb]{\smash{{\SetFigFont{14}{16.8}{\rmdefault}{\mddefault}{\updefault}{\color[rgb]{0,0,0}$(0,0)$}%
}}}}
\put(8476,-8386){\makebox(0,0)[lb]{\smash{{\SetFigFont{14}{16.8}{\rmdefault}{\mddefault}{\updefault}{\color[rgb]{0,0,0}$3$}%
}}}}
\put(8851,-8386){\makebox(0,0)[lb]{\smash{{\SetFigFont{14}{16.8}{\rmdefault}{\mddefault}{\updefault}{\color[rgb]{0,0,0}$\Delta_2$}%
}}}}
\put(8551,-9286){\makebox(0,0)[lb]{\smash{{\SetFigFont{14}{16.8}{\rmdefault}{\mddefault}{\updefault}{\color[rgb]{0,0,0}$(0,0)$}%
}}}}
\put(8476,-9586){\makebox(0,0)[lb]{\smash{{\SetFigFont{14}{16.8}{\rmdefault}{\mddefault}{\updefault}{\color[rgb]{0,0,0}$3$}%
}}}}
\put(8851,-9586){\makebox(0,0)[lb]{\smash{{\SetFigFont{14}{16.8}{\rmdefault}{\mddefault}{\updefault}{\color[rgb]{0,0,0}$\Delta_9$}%
}}}}
\put(8551,-10486){\makebox(0,0)[lb]{\smash{{\SetFigFont{14}{16.8}{\rmdefault}{\mddefault}{\updefault}{\color[rgb]{0,0,0}$(2,0)$}%
}}}}
\put(8476,-10786){\makebox(0,0)[lb]{\smash{{\SetFigFont{14}{16.8}{\rmdefault}{\mddefault}{\updefault}{\color[rgb]{0,0,0}$1$}%
}}}}
\put(8851,-10786){\makebox(0,0)[lb]{\smash{{\SetFigFont{14}{16.8}{\rmdefault}{\mddefault}{\updefault}{\color[rgb]{0,0,0}$\Delta_3$}%
}}}}
\put(6226,-8761){\makebox(0,0)[lb]{\smash{{\SetFigFont{14}{16.8}{\rmdefault}{\mddefault}{\updefault}{\color[rgb]{0,0,0}$\epsilon$}%
}}}}
\put(7126,-8761){\makebox(0,0)[lb]{\smash{{\SetFigFont{14}{16.8}{\rmdefault}{\mddefault}{\updefault}{\color[rgb]{0,0,0}$\epsilon$}%
}}}}
\put(9001,-8836){\makebox(0,0)[lb]{\smash{{\SetFigFont{14}{16.8}{\rmdefault}{\mddefault}{\updefault}{\color[rgb]{0,0,0}$a$}%
}}}}
\put(8401,-8686){\makebox(0,0)[lb]{\smash{{\SetFigFont{14}{16.8}{\rmdefault}{\mddefault}{\updefault}{\color[rgb]{0,0,0}$a$}%
}}}}
\put(7951,-6886){\makebox(0,0)[lb]{\smash{{\SetFigFont{14}{16.8}{\rmdefault}{\mddefault}{\updefault}{\color[rgb]{0,0,0}$a$}%
}}}}
\put(7801,-8761){\makebox(0,0)[lb]{\smash{{\SetFigFont{14}{16.8}{\rmdefault}{\mddefault}{\updefault}{\color[rgb]{0,0,0}$\epsilon$}%
}}}}
\put(9376,-2236){\makebox(0,0)[lb]{\smash{{\SetFigFont{14}{16.8}{\rmdefault}{\mddefault}{\updefault}{\color[rgb]{0,0,0}$\epsilon$}%
}}}}
\put(10426,-4636){\makebox(0,0)[lb]{\smash{{\SetFigFont{14}{16.8}{\rmdefault}{\mddefault}{\updefault}{\color[rgb]{0,0,0}$a$}%
}}}}
\put(7126,2414){\makebox(0,0)[lb]{\smash{{\SetFigFont{14}{16.8}{\rmdefault}{\mddefault}{\updefault}{\color[rgb]{0,0,0}$\Delta_7$}%
}}}}
\put(8851,2414){\makebox(0,0)[lb]{\smash{{\SetFigFont{14}{16.8}{\rmdefault}{\mddefault}{\updefault}{\color[rgb]{0,0,0}$\Delta_4$}%
}}}}
\put(8551,2714){\makebox(0,0)[lb]{\smash{{\SetFigFont{14}{16.8}{\rmdefault}{\mddefault}{\updefault}{\color[rgb]{0,0,0}$(0,0)$}%
}}}}
\put(8476,2414){\makebox(0,0)[lb]{\smash{{\SetFigFont{14}{16.8}{\rmdefault}{\mddefault}{\updefault}{\color[rgb]{0,0,0}$3$}%
}}}}
\put(8551,1514){\makebox(0,0)[lb]{\smash{{\SetFigFont{14}{16.8}{\rmdefault}{\mddefault}{\updefault}{\color[rgb]{0,0,0}$(0,1)$}%
}}}}
\put(8476,1214){\makebox(0,0)[lb]{\smash{{\SetFigFont{14}{16.8}{\rmdefault}{\mddefault}{\updefault}{\color[rgb]{0,0,0}$3$}%
}}}}
\put(8851,1214){\makebox(0,0)[lb]{\smash{{\SetFigFont{14}{16.8}{\rmdefault}{\mddefault}{\updefault}{\color[rgb]{0,0,0}$\Delta_1$}%
}}}}
\put(8551,314){\makebox(0,0)[lb]{\smash{{\SetFigFont{14}{16.8}{\rmdefault}{\mddefault}{\updefault}{\color[rgb]{0,0,0}$(0,1)$}%
}}}}
\put(8476, 14){\makebox(0,0)[lb]{\smash{{\SetFigFont{14}{16.8}{\rmdefault}{\mddefault}{\updefault}{\color[rgb]{0,0,0}$3$}%
}}}}
\put(8851, 14){\makebox(0,0)[lb]{\smash{{\SetFigFont{14}{16.8}{\rmdefault}{\mddefault}{\updefault}{\color[rgb]{0,0,0}$\Delta_6$}%
}}}}
\put(8551,-886){\makebox(0,0)[lb]{\smash{{\SetFigFont{14}{16.8}{\rmdefault}{\mddefault}{\updefault}{\color[rgb]{0,0,0}$(1,0)$}%
}}}}
\put(8476,-1186){\makebox(0,0)[lb]{\smash{{\SetFigFont{14}{16.8}{\rmdefault}{\mddefault}{\updefault}{\color[rgb]{0,0,0}$1$}%
}}}}
\put(8851,-1186){\makebox(0,0)[lb]{\smash{{\SetFigFont{14}{16.8}{\rmdefault}{\mddefault}{\updefault}{\color[rgb]{0,0,0}$\Delta_1$}%
}}}}
\put(9001,3539){\makebox(0,0)[lb]{\smash{{\SetFigFont{14}{16.8}{\rmdefault}{\mddefault}{\updefault}{\color[rgb]{0,0,0}$\epsilon$}%
}}}}
\put(6226,839){\makebox(0,0)[lb]{\smash{{\SetFigFont{14}{16.8}{\rmdefault}{\mddefault}{\updefault}{\color[rgb]{0,0,0}$\epsilon$}%
}}}}
\put(7126,839){\makebox(0,0)[lb]{\smash{{\SetFigFont{14}{16.8}{\rmdefault}{\mddefault}{\updefault}{\color[rgb]{0,0,0}$\epsilon$}%
}}}}
\put(8401,914){\makebox(0,0)[lb]{\smash{{\SetFigFont{14}{16.8}{\rmdefault}{\mddefault}{\updefault}{\color[rgb]{0,0,0}$a$}%
}}}}
\put(6751,2714){\makebox(0,0)[lb]{\smash{{\SetFigFont{14}{16.8}{\rmdefault}{\mddefault}{\updefault}{\color[rgb]{0,0,0}$(0,0)$}%
}}}}
\put(6676,2414){\makebox(0,0)[lb]{\smash{{\SetFigFont{14}{16.8}{\rmdefault}{\mddefault}{\updefault}{\color[rgb]{0,0,0}$3$}%
}}}}
\put(6601,3314){\makebox(0,0)[lb]{\smash{{\SetFigFont{14}{16.8}{\rmdefault}{\mddefault}{\updefault}{\color[rgb]{0,0,0}$a$}%
}}}}
\put(7951,2714){\makebox(0,0)[lb]{\smash{{\SetFigFont{14}{16.8}{\rmdefault}{\mddefault}{\updefault}{\color[rgb]{0,0,0}$a$}%
}}}}
\put(7801,839){\makebox(0,0)[lb]{\smash{{\SetFigFont{14}{16.8}{\rmdefault}{\mddefault}{\updefault}{\color[rgb]{0,0,0}$\epsilon$}%
}}}}
\put(8551,1964){\makebox(0,0)[lb]{\smash{{\SetFigFont{14}{16.8}{\rmdefault}{\mddefault}{\updefault}{\color[rgb]{0,0,0}$a$}%
}}}}
\put(9001,764){\makebox(0,0)[lb]{\smash{{\SetFigFont{14}{16.8}{\rmdefault}{\mddefault}{\updefault}{\color[rgb]{0,0,0}$a$}%
}}}}
\put(8626,-436){\makebox(0,0)[lb]{\smash{{\SetFigFont{14}{16.8}{\rmdefault}{\mddefault}{\updefault}{\color[rgb]{0,0,0}$\epsilon$}%
}}}}
\put(8626,-5836){\makebox(0,0)[lb]{\smash{{\SetFigFont{14}{16.8}{\rmdefault}{\mddefault}{\updefault}{\color[rgb]{0,0,0}$a$}%
}}}}
\put(8626,-7636){\makebox(0,0)[lb]{\smash{{\SetFigFont{14}{16.8}{\rmdefault}{\mddefault}{\updefault}{\color[rgb]{0,0,0}$a$}%
}}}}
\put(8626,-10036){\makebox(0,0)[lb]{\smash{{\SetFigFont{14}{16.8}{\rmdefault}{\mddefault}{\updefault}{\color[rgb]{0,0,0}$\epsilon$}%
}}}}
\put(6751,-6886){\makebox(0,0)[lb]{\smash{{\SetFigFont{14}{16.8}{\rmdefault}{\mddefault}{\updefault}{\color[rgb]{0,0,0}$(0,0)$}%
}}}}
\put(6676,-7186){\makebox(0,0)[lb]{\smash{{\SetFigFont{14}{16.8}{\rmdefault}{\mddefault}{\updefault}{\color[rgb]{0,0,0}$3$}%
}}}}
\put(6601,-6286){\makebox(0,0)[lb]{\smash{{\SetFigFont{14}{16.8}{\rmdefault}{\mddefault}{\updefault}{\color[rgb]{0,0,0}$a$}%
}}}}
\put(8656,-12466){\makebox(0,0)[lb]{\smash{{\SetFigFont{14}{16.8}{\rmdefault}{\mddefault}{\updefault}{\color[rgb]{0,0,0}$a$}%
}}}}
\put(8626,-11236){\makebox(0,0)[lb]{\smash{{\SetFigFont{14}{16.8}{\rmdefault}{\mddefault}{\updefault}{\color[rgb]{0,0,0}$\epsilon$}%
}}}}
\put(8581,-12916){\makebox(0,0)[lb]{\smash{{\SetFigFont{14}{16.8}{\rmdefault}{\mddefault}{\updefault}{\color[rgb]{0,0,0}$(0,0)$}%
}}}}
\put(8506,-13216){\makebox(0,0)[lb]{\smash{{\SetFigFont{14}{16.8}{\rmdefault}{\mddefault}{\updefault}{\color[rgb]{0,0,0}$3$}%
}}}}
\put(8881,-13216){\makebox(0,0)[lb]{\smash{{\SetFigFont{14}{16.8}{\rmdefault}{\mddefault}{\updefault}{\color[rgb]{0,0,0}$\Delta_3$}%
}}}}
\put(8581,-4486){\makebox(0,0)[lb]{\smash{{\SetFigFont{14}{16.8}{\rmdefault}{\mddefault}{\updefault}{\color[rgb]{0,0,0}$(0,2)$}%
}}}}
\put(8506,-4786){\makebox(0,0)[lb]{\smash{{\SetFigFont{14}{16.8}{\rmdefault}{\mddefault}{\updefault}{\color[rgb]{0,0,0}$2$}%
}}}}
\put(8881,-4786){\makebox(0,0)[lb]{\smash{{\SetFigFont{14}{16.8}{\rmdefault}{\mddefault}{\updefault}{\color[rgb]{0,0,0}$\Delta_4$}%
}}}}
\put(8581,-11686){\makebox(0,0)[lb]{\smash{{\SetFigFont{14}{16.8}{\rmdefault}{\mddefault}{\updefault}{\color[rgb]{0,0,0}$(0,2)$}%
}}}}
\put(8506,-11986){\makebox(0,0)[lb]{\smash{{\SetFigFont{14}{16.8}{\rmdefault}{\mddefault}{\updefault}{\color[rgb]{0,0,0}$2$}%
}}}}
\put(8881,-11986){\makebox(0,0)[lb]{\smash{{\SetFigFont{14}{16.8}{\rmdefault}{\mddefault}{\updefault}{\color[rgb]{0,0,0}$\Delta_3$}%
}}}}
\put(12151,-3286){\makebox(0,0)[lb]{\smash{{\SetFigFont{14}{16.8}{\rmdefault}{\mddefault}{\updefault}{\color[rgb]{0,0,0}$(0,2)$}%
}}}}
\put(12076,-3586){\makebox(0,0)[lb]{\smash{{\SetFigFont{14}{16.8}{\rmdefault}{\mddefault}{\updefault}{\color[rgb]{0,0,0}$2$}%
}}}}
\put(12451,-3586){\makebox(0,0)[lb]{\smash{{\SetFigFont{14}{16.8}{\rmdefault}{\mddefault}{\updefault}{\color[rgb]{0,0,0}$\Delta_0$}%
}}}}
\put(12151,-5686){\makebox(0,0)[lb]{\smash{{\SetFigFont{14}{16.8}{\rmdefault}{\mddefault}{\updefault}{\color[rgb]{0,0,0}$(0,0)$}%
}}}}
\put(12076,-5986){\makebox(0,0)[lb]{\smash{{\SetFigFont{14}{16.8}{\rmdefault}{\mddefault}{\updefault}{\color[rgb]{0,0,0}$3$}%
}}}}
\put(12451,-5986){\makebox(0,0)[lb]{\smash{{\SetFigFont{14}{16.8}{\rmdefault}{\mddefault}{\updefault}{\color[rgb]{0,0,0}$\Delta_0$}%
}}}}
\put(12601,-4636){\makebox(0,0)[lb]{\smash{{\SetFigFont{14}{16.8}{\rmdefault}{\mddefault}{\updefault}{\color[rgb]{0,0,0}$a$}%
}}}}
\put(13951,-3286){\makebox(0,0)[lb]{\smash{{\SetFigFont{14}{16.8}{\rmdefault}{\mddefault}{\updefault}{\color[rgb]{0,0,0}$(0,0)$}%
}}}}
\put(13876,-3586){\makebox(0,0)[lb]{\smash{{\SetFigFont{14}{16.8}{\rmdefault}{\mddefault}{\updefault}{\color[rgb]{0,0,0}$3$}%
}}}}
\put(14251,-3586){\makebox(0,0)[lb]{\smash{{\SetFigFont{14}{16.8}{\rmdefault}{\mddefault}{\updefault}{\color[rgb]{0,0,0}$\Delta_0$}%
}}}}
\put(13951,-4486){\makebox(0,0)[lb]{\smash{{\SetFigFont{14}{16.8}{\rmdefault}{\mddefault}{\updefault}{\color[rgb]{0,0,0}$(0,0)$}%
}}}}
\put(13876,-4786){\makebox(0,0)[lb]{\smash{{\SetFigFont{14}{16.8}{\rmdefault}{\mddefault}{\updefault}{\color[rgb]{0,0,0}$3$}%
}}}}
\put(14401,-2236){\makebox(0,0)[lb]{\smash{{\SetFigFont{14}{16.8}{\rmdefault}{\mddefault}{\updefault}{\color[rgb]{0,0,0}$a$}%
}}}}
\put(14251,-4786){\makebox(0,0)[lb]{\smash{{\SetFigFont{14}{16.8}{\rmdefault}{\mddefault}{\updefault}{\color[rgb]{0,0,0}$\Delta_5$}%
}}}}
\put(14551,-3886){\makebox(0,0)[lb]{\smash{{\SetFigFont{14}{16.8}{\rmdefault}{\mddefault}{\updefault}{\color[rgb]{0,0,0}$a$}%
}}}}
\put(14026,-4036){\makebox(0,0)[lb]{\smash{{\SetFigFont{14}{16.8}{\rmdefault}{\mddefault}{\updefault}{\color[rgb]{0,0,0}$a$}%
}}}}
\put(12751,-2686){\makebox(0,0)[lb]{\smash{{\SetFigFont{14}{16.8}{\rmdefault}{\mddefault}{\updefault}{\color[rgb]{0,0,0}$b$}%
}}}}
\put(9151,-11086){\makebox(0,0)[lb]{\smash{{\SetFigFont{14}{16.8}{\rmdefault}{\mddefault}{\updefault}{\color[rgb]{0,0,0}$b$}%
}}}}
\put(5476,-1936){\makebox(0,0)[lb]{\smash{{\SetFigFont{14}{16.8}{\rmdefault}{\mddefault}{\updefault}{\color[rgb]{0,0,0}$\Delta_1: 0=\{t\}=\{y\}<\{x\}$}%
}}}}
\put(5476,-1636){\makebox(0,0)[lb]{\smash{{\SetFigFont{14}{16.8}{\rmdefault}{\mddefault}{\updefault}{\color[rgb]{0,0,0}$\Delta_0: 0=\{t\}=\{x\}=\{y\}$}%
}}}}
\put(5476,-3736){\makebox(0,0)[lb]{\smash{{\SetFigFont{14}{16.8}{\rmdefault}{\mddefault}{\updefault}{\color[rgb]{0,0,0}$\Delta_7: 0<\{x\}<\{t\}=\{y\}$}%
}}}}
\put(5476,-3436){\makebox(0,0)[lb]{\smash{{\SetFigFont{14}{16.8}{\rmdefault}{\mddefault}{\updefault}{\color[rgb]{0,0,0}$\Delta_6: 0<\{t\}=\{y\}<\{x\}$}%
}}}}
\put(5476,-3136){\makebox(0,0)[lb]{\smash{{\SetFigFont{14}{16.8}{\rmdefault}{\mddefault}{\updefault}{\color[rgb]{0,0,0}$\Delta_5: 0<\{t\}=\{x\}=\{y\}$}%
}}}}
\put(5476,-2836){\makebox(0,0)[lb]{\smash{{\SetFigFont{14}{16.8}{\rmdefault}{\mddefault}{\updefault}{\color[rgb]{0,0,0}$\Delta_4: 0=\{x\}<\{t\}=\{y\}$}%
}}}}
\put(5476,-2236){\makebox(0,0)[lb]{\smash{{\SetFigFont{14}{16.8}{\rmdefault}{\mddefault}{\updefault}{\color[rgb]{0,0,0}$\Delta_2: 0=\{t\}<\{x\}=\{y\}$}%
}}}}
\put(5476,-4036){\makebox(0,0)[lb]{\smash{{\SetFigFont{14}{16.8}{\rmdefault}{\mddefault}{\updefault}{\color[rgb]{0,0,0}$\Delta_8: 0<\{x\}=\{y\}<\{t\}$}%
}}}}
\put(5476,-2536){\makebox(0,0)[lb]{\smash{{\SetFigFont{14}{16.8}{\rmdefault}{\mddefault}{\updefault}{\color[rgb]{0,0,0}$\Delta_3: 0=\{x\}=\{y\}<\{t\}$}%
}}}}
\put(5476,-4336){\makebox(0,0)[lb]{\smash{{\SetFigFont{14}{16.8}{\rmdefault}{\mddefault}{\updefault}{\color[rgb]{0,0,0}$\Delta_9: 0<\{t\}<\{x\}=\{y\}$}%
}}}}
\put(15751,1364){\makebox(0,0)[lb]{\smash{{\SetFigFont{14}{16.8}{\rmdefault}{\mddefault}{\updefault}{\color[rgb]{0,0,0}$a$}%
}}}}
\put(9826,-8236){\makebox(0,0)[lb]{\smash{{\SetFigFont{14}{16.8}{\rmdefault}{\mddefault}{\updefault}{\color[rgb]{0,0,0}$a$}%
}}}}
\put(8851,-14086){\makebox(0,0)[lb]{\smash{{\SetFigFont{14}{16.8}{\rmdefault}{\mddefault}{\updefault}{\color[rgb]{0,0,0}$\vdots$}%
}}}}
\put(12451,-6886){\makebox(0,0)[lb]{\smash{{\SetFigFont{14}{16.8}{\rmdefault}{\mddefault}{\updefault}{\color[rgb]{0,0,0}$\vdots$}%
}}}}
\put(14251,-5686){\makebox(0,0)[lb]{\smash{{\SetFigFont{14}{16.8}{\rmdefault}{\mddefault}{\updefault}{\color[rgb]{0,0,0}$\vdots$}%
}}}}
\put(6976,-11161){\makebox(0,0)[lb]{\smash{{\SetFigFont{20}{24.0}{\rmdefault}{\bfdefault}{\updefault}{\color[rgb]{0,0,0}$C_2$}%
}}}}
\put(12826,1589){\makebox(0,0)[lb]{\smash{{\SetFigFont{20}{24.0}{\rmdefault}{\bfdefault}{\updefault}{\color[rgb]{0,0,0}$C_1$}%
}}}}
\put(16126,-4036){\makebox(0,0)[lb]{\smash{{\SetFigFont{14}{16.8}{\rmdefault}{\mddefault}{\updefault}{\color[rgb]{0,0,0}$c$}%
}}}}
\put(14326,-6136){\makebox(0,0)[lb]{\smash{{\SetFigFont{20}{24.0}{\rmdefault}{\bfdefault}{\updefault}{\color[rgb]{0,0,0}$C_1$}%
}}}}
\put(13306,-10846){\makebox(0,0)[lb]{\smash{{\SetFigFont{14}{16.8}{\rmdefault}{\mddefault}{\updefault}{\color[rgb]{0,0,0}$1$}%
}}}}
\put(11656,-10846){\makebox(0,0)[lb]{\smash{{\SetFigFont{14}{16.8}{\rmdefault}{\mddefault}{\updefault}{\color[rgb]{0,0,0}$0$}%
}}}}
\put(14131,-10396){\makebox(0,0)[lb]{\smash{{\SetFigFont{14}{16.8}{\rmdefault}{\mddefault}{\updefault}{\color[rgb]{0,0,0}$\epsilon$}%
}}}}
\put(12031,-10621){\makebox(0,0)[lb]{\smash{{\SetFigFont{14}{16.8}{\rmdefault}{\mddefault}{\updefault}{\color[rgb]{0,0,0}$1 < x \leq 2$}%
}}}}
\put(14956,-10846){\makebox(0,0)[lb]{\smash{{\SetFigFont{14}{16.8}{\rmdefault}{\mddefault}{\updefault}{\color[rgb]{0,0,0}$2$}%
}}}}
\put(13681,-10621){\makebox(0,0)[lb]{\smash{{\SetFigFont{14}{16.8}{\rmdefault}{\mddefault}{\updefault}{\color[rgb]{0,0,0}$x=4, \{x\}$}%
}}}}
\put(13306,-12946){\makebox(0,0)[lb]{\smash{{\SetFigFont{14}{16.8}{\rmdefault}{\mddefault}{\updefault}{\color[rgb]{0,0,0}$a$}%
}}}}
\put(13381,-9271){\makebox(0,0)[lb]{\smash{{\SetFigFont{14}{16.8}{\rmdefault}{\mddefault}{\updefault}{\color[rgb]{0,0,0}$a$}%
}}}}
\put(12931,-9496){\makebox(0,0)[lb]{\smash{{\SetFigFont{14}{16.8}{\rmdefault}{\mddefault}{\updefault}{\color[rgb]{0,0,0}$x=2, \{x\}$}%
}}}}
\put(13306,-11971){\makebox(0,0)[lb]{\smash{{\SetFigFont{14}{16.8}{\rmdefault}{\mddefault}{\updefault}{\color[rgb]{0,0,0}$3$}%
}}}}
\put(10531,-11671){\makebox(0,0)[lb]{\smash{{\SetFigFont{14}{16.8}{\rmdefault}{\mddefault}{\updefault}{\color[rgb]{0,0,0}$0 \leq x < 1, \{x\}$}%
}}}}
\put(12481,-10396){\makebox(0,0)[lb]{\smash{{\SetFigFont{14}{16.8}{\rmdefault}{\mddefault}{\updefault}{\color[rgb]{0,0,0}$a$}%
}}}}
\put(12856,-13171){\makebox(0,0)[lb]{\smash{{\SetFigFont{14}{16.8}{\rmdefault}{\mddefault}{\updefault}{\color[rgb]{0,0,0}$0 < x < 1$}%
}}}}
\put(11131,-11446){\makebox(0,0)[lb]{\smash{{\SetFigFont{14}{16.8}{\rmdefault}{\mddefault}{\updefault}{\color[rgb]{0,0,0}$\epsilon$}%
}}}}
\put(15706,-11596){\makebox(0,0)[lb]{\smash{{\SetFigFont{14}{16.8}{\rmdefault}{\mddefault}{\updefault}{\color[rgb]{0,0,0}$y=2$}%
}}}}
\put(15931,-11371){\makebox(0,0)[lb]{\smash{{\SetFigFont{14}{16.8}{\rmdefault}{\mddefault}{\updefault}{\color[rgb]{0,0,0}$b$}%
}}}}
\put(13681,-11221){\makebox(0,0)[lb]{\smash{{\SetFigFont{14}{16.8}{\rmdefault}{\mddefault}{\updefault}{\color[rgb]{0,0,0}$\epsilon$}%
}}}}
\put(13456,-11446){\makebox(0,0)[lb]{\smash{{\SetFigFont{14}{16.8}{\rmdefault}{\mddefault}{\updefault}{\color[rgb]{0,0,0}$y=2, \{y\}$}%
}}}}
\put(14581,-11746){\makebox(0,0)[lb]{\smash{{\SetFigFont{14}{16.8}{\rmdefault}{\mddefault}{\updefault}{\color[rgb]{0,0,0}$a$}%
}}}}
\put(14206,-11971){\makebox(0,0)[lb]{\smash{{\SetFigFont{14}{16.8}{\rmdefault}{\mddefault}{\updefault}{\color[rgb]{0,0,0}$x=2, \{x,y\}$}%
}}}}
\put(15376,-9661){\makebox(0,0)[lb]{\smash{{\SetFigFont{14}{16.8}{\rmdefault}{\mddefault}{\updefault}{\color[rgb]{0,0,0}$c$}%
}}}}
\put(14551,-9961){\makebox(0,0)[lb]{\smash{{\SetFigFont{14}{16.8}{\rmdefault}{\mddefault}{\updefault}{\color[rgb]{0,0,0}$(x>4) \wedge (y \geq 4)$}%
}}}}
\put(13201,-13711){\makebox(0,0)[lb]{\smash{{\SetFigFont{14}{16.8}{\rmdefault}{\mddefault}{\updefault}{\color[rgb]{0,0,0}$(a)$}%
}}}}
\end{picture}%

%% file: complex2.pdf_t
\begin{picture}(0,0)%
\includegraphics{complex2.pdf}%
\end{picture}%
\setlength{\unitlength}{3947sp}%
\begingroup\makeatletter\ifx\SetFigFont\undefined%
\gdef\SetFigFont#1#2#3#4#5{%
  \reset@font\fontsize{#1}{#2pt}%
  \fontfamily{#3}\fontseries{#4}\fontshape{#5}%
  \selectfont}%
\fi\endgroup%
\begin{picture}(11565,19900)(4411,-16418)
\put(9301,-7036){\makebox(0,0)[lb]{\smash{{\SetFigFont{20}{24.0}{\rmdefault}{\bfdefault}{\updefault}{\color[rgb]{0,0,0}$C_4$}%
}}}}
\put(10351,-2086){\makebox(0,0)[lb]{\smash{{\SetFigFont{14}{16.8}{\rmdefault}{\mddefault}{\updefault}{\color[rgb]{0,0,0}$(0,3)$}%
}}}}
\put(10651,-2386){\makebox(0,0)[lb]{\smash{{\SetFigFont{14}{16.8}{\rmdefault}{\mddefault}{\updefault}{\color[rgb]{0,0,0}$\Delta_4$}%
}}}}
\put(10276,-2386){\makebox(0,0)[lb]{\smash{{\SetFigFont{14}{16.8}{\rmdefault}{\mddefault}{\updefault}{\color[rgb]{0,0,0}$2$}%
}}}}
\put(4576,1364){\makebox(0,0)[lb]{\smash{{\SetFigFont{14}{16.8}{\rmdefault}{\mddefault}{\updefault}{\color[rgb]{0,0,0}$(0,1)$}%
}}}}
\put(4726,2564){\makebox(0,0)[lb]{\smash{{\SetFigFont{14}{16.8}{\rmdefault}{\mddefault}{\updefault}{\color[rgb]{0,0,0}$0$}%
}}}}
\put(4726,3239){\makebox(0,0)[lb]{\smash{{\SetFigFont{20}{24.0}{\rmdefault}{\mddefault}{\updefault}{\color[rgb]{0,0,0}$t$}%
}}}}
\put(12151,-2086){\makebox(0,0)[lb]{\smash{{\SetFigFont{14}{16.8}{\rmdefault}{\mddefault}{\updefault}{\color[rgb]{0,0,0}$(0,1)$}%
}}}}
\put(12076,-2386){\makebox(0,0)[lb]{\smash{{\SetFigFont{14}{16.8}{\rmdefault}{\mddefault}{\updefault}{\color[rgb]{0,0,0}$2$}%
}}}}
\put(12451,-2386){\makebox(0,0)[lb]{\smash{{\SetFigFont{14}{16.8}{\rmdefault}{\mddefault}{\updefault}{\color[rgb]{0,0,0}$\Delta_4$}%
}}}}
\put(8581,-2686){\makebox(0,0)[lb]{\smash{{\SetFigFont{14}{16.8}{\rmdefault}{\mddefault}{\updefault}{\color[rgb]{0,0,0}$(0,2)$}%
}}}}
\put(8881,-2986){\makebox(0,0)[lb]{\smash{{\SetFigFont{14}{16.8}{\rmdefault}{\mddefault}{\updefault}{\color[rgb]{0,0,0}$\Delta_0$}%
}}}}
\put(8476,-2986){\makebox(0,0)[lb]{\smash{{\SetFigFont{14}{16.8}{\rmdefault}{\mddefault}{\updefault}{\color[rgb]{0,0,0}$2$}%
}}}}
\put(6781,-3286){\makebox(0,0)[lb]{\smash{{\SetFigFont{14}{16.8}{\rmdefault}{\mddefault}{\updefault}{\color[rgb]{0,0,0}$(0,2)$}%
}}}}
\put(7081,-3586){\makebox(0,0)[lb]{\smash{{\SetFigFont{14}{16.8}{\rmdefault}{\mddefault}{\updefault}{\color[rgb]{0,0,0}$\Delta_4$}%
}}}}
\put(6676,-3586){\makebox(0,0)[lb]{\smash{{\SetFigFont{14}{16.8}{\rmdefault}{\mddefault}{\updefault}{\color[rgb]{0,0,0}$2$}%
}}}}
\put(12151,-3886){\makebox(0,0)[lb]{\smash{{\SetFigFont{14}{16.8}{\rmdefault}{\mddefault}{\updefault}{\color[rgb]{0,0,0}$(0,0)$}%
}}}}
\put(12076,-4186){\makebox(0,0)[lb]{\smash{{\SetFigFont{14}{16.8}{\rmdefault}{\mddefault}{\updefault}{\color[rgb]{0,0,0}$3$}%
}}}}
\put(12451,-4186){\makebox(0,0)[lb]{\smash{{\SetFigFont{14}{16.8}{\rmdefault}{\mddefault}{\updefault}{\color[rgb]{0,0,0}$\Delta_3$}%
}}}}
\put(8881,-4786){\makebox(0,0)[lb]{\smash{{\SetFigFont{14}{16.8}{\rmdefault}{\mddefault}{\updefault}{\color[rgb]{0,0,0}$\Delta_0$}%
}}}}
\put(8581,-4486){\makebox(0,0)[lb]{\smash{{\SetFigFont{14}{16.8}{\rmdefault}{\mddefault}{\updefault}{\color[rgb]{0,0,0}$(0,0)$}%
}}}}
\put(8506,-4786){\makebox(0,0)[lb]{\smash{{\SetFigFont{14}{16.8}{\rmdefault}{\mddefault}{\updefault}{\color[rgb]{0,0,0}$3$}%
}}}}
\put(7051,-5386){\makebox(0,0)[lb]{\smash{{\SetFigFont{14}{16.8}{\rmdefault}{\mddefault}{\updefault}{\color[rgb]{0,0,0}$\Delta_3$}%
}}}}
\put(6751,-5086){\makebox(0,0)[lb]{\smash{{\SetFigFont{14}{16.8}{\rmdefault}{\mddefault}{\updefault}{\color[rgb]{0,0,0}$(0,0)$}%
}}}}
\put(6676,-5386){\makebox(0,0)[lb]{\smash{{\SetFigFont{14}{16.8}{\rmdefault}{\mddefault}{\updefault}{\color[rgb]{0,0,0}$3$}%
}}}}
\put(12181,-5686){\makebox(0,0)[lb]{\smash{{\SetFigFont{14}{16.8}{\rmdefault}{\mddefault}{\updefault}{\color[rgb]{0,0,0}$(2,0)$}%
}}}}
\put(12481,-5986){\makebox(0,0)[lb]{\smash{{\SetFigFont{14}{16.8}{\rmdefault}{\mddefault}{\updefault}{\color[rgb]{0,0,0}$\Delta_3$}%
}}}}
\put(12076,-5986){\makebox(0,0)[lb]{\smash{{\SetFigFont{14}{16.8}{\rmdefault}{\mddefault}{\updefault}{\color[rgb]{0,0,0}$1$}%
}}}}
\put(8581,-8086){\makebox(0,0)[lb]{\smash{{\SetFigFont{14}{16.8}{\rmdefault}{\mddefault}{\updefault}{\color[rgb]{0,0,0}$(0,1)$}%
}}}}
\put(8881,-8386){\makebox(0,0)[lb]{\smash{{\SetFigFont{14}{16.8}{\rmdefault}{\mddefault}{\updefault}{\color[rgb]{0,0,0}$\Delta_4$}%
}}}}
\put(8476,-8386){\makebox(0,0)[lb]{\smash{{\SetFigFont{14}{16.8}{\rmdefault}{\mddefault}{\updefault}{\color[rgb]{0,0,0}$2$}%
}}}}
\put(12181,-8086){\makebox(0,0)[lb]{\smash{{\SetFigFont{14}{16.8}{\rmdefault}{\mddefault}{\updefault}{\color[rgb]{0,0,0}$(0,1)$}%
}}}}
\put(12106,-8386){\makebox(0,0)[lb]{\smash{{\SetFigFont{14}{16.8}{\rmdefault}{\mddefault}{\updefault}{\color[rgb]{0,0,0}$2$}%
}}}}
\put(12481,-8386){\makebox(0,0)[lb]{\smash{{\SetFigFont{14}{16.8}{\rmdefault}{\mddefault}{\updefault}{\color[rgb]{0,0,0}$\Delta_{11}$}%
}}}}
\put(10381,-8086){\makebox(0,0)[lb]{\smash{{\SetFigFont{14}{16.8}{\rmdefault}{\mddefault}{\updefault}{\color[rgb]{0,0,0}$(0,0)$}%
}}}}
\put(10306,-8386){\makebox(0,0)[lb]{\smash{{\SetFigFont{14}{16.8}{\rmdefault}{\mddefault}{\updefault}{\color[rgb]{0,0,0}$2$}%
}}}}
\put(10681,-8386){\makebox(0,0)[lb]{\smash{{\SetFigFont{14}{16.8}{\rmdefault}{\mddefault}{\updefault}{\color[rgb]{0,0,0}$\Delta_3$}%
}}}}
\put(13381,-8686){\makebox(0,0)[lb]{\smash{{\SetFigFont{14}{16.8}{\rmdefault}{\mddefault}{\updefault}{\color[rgb]{0,0,0}$(0,0)$}%
}}}}
\put(13306,-8986){\makebox(0,0)[lb]{\smash{{\SetFigFont{14}{16.8}{\rmdefault}{\mddefault}{\updefault}{\color[rgb]{0,0,0}$3$}%
}}}}
\put(13681,-8986){\makebox(0,0)[lb]{\smash{{\SetFigFont{14}{16.8}{\rmdefault}{\mddefault}{\updefault}{\color[rgb]{0,0,0}$\Delta_3$}%
}}}}
\put(14581,-9286){\makebox(0,0)[lb]{\smash{{\SetFigFont{14}{16.8}{\rmdefault}{\mddefault}{\updefault}{\color[rgb]{0,0,0}$(0,0)$}%
}}}}
\put(14506,-9586){\makebox(0,0)[lb]{\smash{{\SetFigFont{14}{16.8}{\rmdefault}{\mddefault}{\updefault}{\color[rgb]{0,0,0}$3$}%
}}}}
\put(14881,-9586){\makebox(0,0)[lb]{\smash{{\SetFigFont{14}{16.8}{\rmdefault}{\mddefault}{\updefault}{\color[rgb]{0,0,0}$\Delta_0$}%
}}}}
\put(14551,-10486){\makebox(0,0)[lb]{\smash{{\SetFigFont{14}{16.8}{\rmdefault}{\mddefault}{\updefault}{\color[rgb]{0,0,0}$(2,0)$}%
}}}}
\put(14476,-10786){\makebox(0,0)[lb]{\smash{{\SetFigFont{14}{16.8}{\rmdefault}{\mddefault}{\updefault}{\color[rgb]{0,0,0}$1$}%
}}}}
\put(14851,-10786){\makebox(0,0)[lb]{\smash{{\SetFigFont{14}{16.8}{\rmdefault}{\mddefault}{\updefault}{\color[rgb]{0,0,0}$\Delta_0$}%
}}}}
\put(12151,-9886){\makebox(0,0)[lb]{\smash{{\SetFigFont{14}{16.8}{\rmdefault}{\mddefault}{\updefault}{\color[rgb]{0,0,0}$(0,0)$}%
}}}}
\put(12076,-10186){\makebox(0,0)[lb]{\smash{{\SetFigFont{14}{16.8}{\rmdefault}{\mddefault}{\updefault}{\color[rgb]{0,0,0}$3$}%
}}}}
\put(12451,-10186){\makebox(0,0)[lb]{\smash{{\SetFigFont{14}{16.8}{\rmdefault}{\mddefault}{\updefault}{\color[rgb]{0,0,0}$\Delta_3$}%
}}}}
\put(14536,-7486){\makebox(0,0)[lb]{\smash{{\SetFigFont{14}{16.8}{\rmdefault}{\mddefault}{\updefault}{\color[rgb]{0,0,0}$(0,1)$}%
}}}}
\put(14461,-7786){\makebox(0,0)[lb]{\smash{{\SetFigFont{14}{16.8}{\rmdefault}{\mddefault}{\updefault}{\color[rgb]{0,0,0}$2$}%
}}}}
\put(14836,-7786){\makebox(0,0)[lb]{\smash{{\SetFigFont{14}{16.8}{\rmdefault}{\mddefault}{\updefault}{\color[rgb]{0,0,0}$\Delta_{10}$}%
}}}}
\put(6751,-6886){\makebox(0,0)[lb]{\smash{{\SetFigFont{14}{16.8}{\rmdefault}{\mddefault}{\updefault}{\color[rgb]{0,0,0}$(2,0)$}%
}}}}
\put(6676,-7186){\makebox(0,0)[lb]{\smash{{\SetFigFont{14}{16.8}{\rmdefault}{\mddefault}{\updefault}{\color[rgb]{0,0,0}$1$}%
}}}}
\put(7051,-7186){\makebox(0,0)[lb]{\smash{{\SetFigFont{14}{16.8}{\rmdefault}{\mddefault}{\updefault}{\color[rgb]{0,0,0}$\Delta_3$}%
}}}}
\put(12136,-6886){\makebox(0,0)[lb]{\smash{{\SetFigFont{14}{16.8}{\rmdefault}{\mddefault}{\updefault}{\color[rgb]{0,0,0}$(0,1)$}%
}}}}
\put(12061,-7186){\makebox(0,0)[lb]{\smash{{\SetFigFont{14}{16.8}{\rmdefault}{\mddefault}{\updefault}{\color[rgb]{0,0,0}$2$}%
}}}}
\put(12436,-7186){\makebox(0,0)[lb]{\smash{{\SetFigFont{14}{16.8}{\rmdefault}{\mddefault}{\updefault}{\color[rgb]{0,0,0}$\Delta_{12}$}%
}}}}
\put(5476,2864){\makebox(0,0)[lb]{\smash{{\SetFigFont{14}{16.8}{\rmdefault}{\mddefault}{\updefault}{\color[rgb]{0,0,0}$\Delta_1: 0=\{t\}=\{y\}<\{x\}$}%
}}}}
\put(5476,3164){\makebox(0,0)[lb]{\smash{{\SetFigFont{14}{16.8}{\rmdefault}{\mddefault}{\updefault}{\color[rgb]{0,0,0}$\Delta_0: 0=\{t\}=\{x\}=\{y\}$}%
}}}}
\put(5476,1064){\makebox(0,0)[lb]{\smash{{\SetFigFont{14}{16.8}{\rmdefault}{\mddefault}{\updefault}{\color[rgb]{0,0,0}$\Delta_7: 0<\{x\}<\{t\}=\{y\}$}%
}}}}
\put(5476,1364){\makebox(0,0)[lb]{\smash{{\SetFigFont{14}{16.8}{\rmdefault}{\mddefault}{\updefault}{\color[rgb]{0,0,0}$\Delta_6: 0<\{t\}=\{y\}<\{x\}$}%
}}}}
\put(5476,1664){\makebox(0,0)[lb]{\smash{{\SetFigFont{14}{16.8}{\rmdefault}{\mddefault}{\updefault}{\color[rgb]{0,0,0}$\Delta_5: 0<\{t\}=\{x\}=\{y\}$}%
}}}}
\put(5476,1964){\makebox(0,0)[lb]{\smash{{\SetFigFont{14}{16.8}{\rmdefault}{\mddefault}{\updefault}{\color[rgb]{0,0,0}$\Delta_4: 0=\{x\}<\{t\}=\{y\}$}%
}}}}
\put(5476,2564){\makebox(0,0)[lb]{\smash{{\SetFigFont{14}{16.8}{\rmdefault}{\mddefault}{\updefault}{\color[rgb]{0,0,0}$\Delta_2: 0=\{t\}<\{x\}=\{y\}$}%
}}}}
\put(5476,764){\makebox(0,0)[lb]{\smash{{\SetFigFont{14}{16.8}{\rmdefault}{\mddefault}{\updefault}{\color[rgb]{0,0,0}$\Delta_8: 0<\{x\}=\{y\}<\{t\}$}%
}}}}
\put(5476,2264){\makebox(0,0)[lb]{\smash{{\SetFigFont{14}{16.8}{\rmdefault}{\mddefault}{\updefault}{\color[rgb]{0,0,0}$\Delta_3: 0=\{x\}=\{y\}<\{t\}$}%
}}}}
\put(5476,464){\makebox(0,0)[lb]{\smash{{\SetFigFont{14}{16.8}{\rmdefault}{\mddefault}{\updefault}{\color[rgb]{0,0,0}$\Delta_9: 0<\{t\}<\{x\}=\{y\}$}%
}}}}
\put(5476,164){\makebox(0,0)[lb]{\smash{{\SetFigFont{14}{16.8}{\rmdefault}{\mddefault}{\updefault}{\color[rgb]{0,0,0}$\Delta_{10}: 0=\{t\}=\{x\}<\{y\}$}%
}}}}
\put(5476,-136){\makebox(0,0)[lb]{\smash{{\SetFigFont{14}{16.8}{\rmdefault}{\mddefault}{\updefault}{\color[rgb]{0,0,0}$\Delta_{11}: 0=\{x\}<\{t\}<\{y\}$}%
}}}}
\put(5476,-436){\makebox(0,0)[lb]{\smash{{\SetFigFont{14}{16.8}{\rmdefault}{\mddefault}{\updefault}{\color[rgb]{0,0,0}$\Delta_{12}: 0=\{x\}<\{y\}<\{t\}$}%
}}}}
\put(12181,-886){\makebox(0,0)[lb]{\smash{{\SetFigFont{14}{16.8}{\rmdefault}{\mddefault}{\updefault}{\color[rgb]{0,0,0}$(2,0)$}%
}}}}
\put(12481,-1186){\makebox(0,0)[lb]{\smash{{\SetFigFont{14}{16.8}{\rmdefault}{\mddefault}{\updefault}{\color[rgb]{0,0,0}$\Delta_0$}%
}}}}
\put(12106,-1186){\makebox(0,0)[lb]{\smash{{\SetFigFont{14}{16.8}{\rmdefault}{\mddefault}{\updefault}{\color[rgb]{0,0,0}$1$}%
}}}}
\put(10351,314){\makebox(0,0)[lb]{\smash{{\SetFigFont{14}{16.8}{\rmdefault}{\mddefault}{\updefault}{\color[rgb]{0,0,0}$(1,1)$}%
}}}}
\put(10651, 14){\makebox(0,0)[lb]{\smash{{\SetFigFont{14}{16.8}{\rmdefault}{\mddefault}{\updefault}{\color[rgb]{0,0,0}$\Delta_5$}%
}}}}
\put(10276, 14){\makebox(0,0)[lb]{\smash{{\SetFigFont{14}{16.8}{\rmdefault}{\mddefault}{\updefault}{\color[rgb]{0,0,0}$1$}%
}}}}
\put(10351,-886){\makebox(0,0)[lb]{\smash{{\SetFigFont{14}{16.8}{\rmdefault}{\mddefault}{\updefault}{\color[rgb]{0,0,0}$(2,2)$}%
}}}}
\put(10651,-1186){\makebox(0,0)[lb]{\smash{{\SetFigFont{14}{16.8}{\rmdefault}{\mddefault}{\updefault}{\color[rgb]{0,0,0}$\Delta_0$}%
}}}}
\put(10276,-1186){\makebox(0,0)[lb]{\smash{{\SetFigFont{14}{16.8}{\rmdefault}{\mddefault}{\updefault}{\color[rgb]{0,0,0}$1$}%
}}}}
\put(8551,-886){\makebox(0,0)[lb]{\smash{{\SetFigFont{14}{16.8}{\rmdefault}{\mddefault}{\updefault}{\color[rgb]{0,0,0}$(1,0)$}%
}}}}
\put(8851,-1186){\makebox(0,0)[lb]{\smash{{\SetFigFont{14}{16.8}{\rmdefault}{\mddefault}{\updefault}{\color[rgb]{0,0,0}$\Delta_1$}%
}}}}
\put(4576,164){\makebox(0,0)[lb]{\smash{{\SetFigFont{14}{16.8}{\rmdefault}{\mddefault}{\updefault}{\color[rgb]{0,0,0}$(1,2)$}%
}}}}
\put(4726,-1036){\makebox(0,0)[lb]{\smash{{\SetFigFont{14}{16.8}{\rmdefault}{\mddefault}{\updefault}{\color[rgb]{0,0,0}$2$}%
}}}}
\put(8476,-1186){\makebox(0,0)[lb]{\smash{{\SetFigFont{14}{16.8}{\rmdefault}{\mddefault}{\updefault}{\color[rgb]{0,0,0}$1$}%
}}}}
\put(4576,-2236){\makebox(0,0)[lb]{\smash{{\SetFigFont{14}{16.8}{\rmdefault}{\mddefault}{\updefault}{\color[rgb]{0,0,0}$(3,4)$}%
}}}}
\put(4801,-1636){\makebox(0,0)[lb]{\smash{{\SetFigFont{14}{16.8}{\rmdefault}{\mddefault}{\updefault}{\color[rgb]{0,0,0}$\vdots$}%
}}}}
\put(10801,-1636){\makebox(0,0)[lb]{\smash{{\SetFigFont{14}{16.8}{\rmdefault}{\mddefault}{\updefault}{\color[rgb]{0,0,0}$\epsilon$}%
}}}}
\put(9676,-1036){\makebox(0,0)[lb]{\smash{{\SetFigFont{14}{16.8}{\rmdefault}{\mddefault}{\updefault}{\color[rgb]{0,0,0}$\epsilon$}%
}}}}
\put(10351,2714){\makebox(0,0)[lb]{\smash{{\SetFigFont{14}{16.8}{\rmdefault}{\mddefault}{\updefault}{\color[rgb]{0,0,0}$(0,0)$}%
}}}}
\put(10276,2414){\makebox(0,0)[lb]{\smash{{\SetFigFont{14}{16.8}{\rmdefault}{\mddefault}{\updefault}{\color[rgb]{0,0,0}$0$}%
}}}}
\put(10651,2414){\makebox(0,0)[lb]{\smash{{\SetFigFont{14}{16.8}{\rmdefault}{\mddefault}{\updefault}{\color[rgb]{0,0,0}$\Delta_0$}%
}}}}
\put(12151,2714){\makebox(0,0)[lb]{\smash{{\SetFigFont{14}{16.8}{\rmdefault}{\mddefault}{\updefault}{\color[rgb]{0,0,0}$(0,0)$}%
}}}}
\put(12076,2414){\makebox(0,0)[lb]{\smash{{\SetFigFont{14}{16.8}{\rmdefault}{\mddefault}{\updefault}{\color[rgb]{0,0,0}$3$}%
}}}}
\put(12451,2414){\makebox(0,0)[lb]{\smash{{\SetFigFont{14}{16.8}{\rmdefault}{\mddefault}{\updefault}{\color[rgb]{0,0,0}$\Delta_0$}%
}}}}
\put(11551,2714){\makebox(0,0)[lb]{\smash{{\SetFigFont{14}{16.8}{\rmdefault}{\mddefault}{\updefault}{\color[rgb]{0,0,0}$\epsilon$}%
}}}}
\put(8851,1214){\makebox(0,0)[lb]{\smash{{\SetFigFont{14}{16.8}{\rmdefault}{\mddefault}{\updefault}{\color[rgb]{0,0,0}$\Delta_4$}%
}}}}
\put(8551,1514){\makebox(0,0)[lb]{\smash{{\SetFigFont{14}{16.8}{\rmdefault}{\mddefault}{\updefault}{\color[rgb]{0,0,0}$(0,0)$}%
}}}}
\put(8476,1214){\makebox(0,0)[lb]{\smash{{\SetFigFont{14}{16.8}{\rmdefault}{\mddefault}{\updefault}{\color[rgb]{0,0,0}$3$}%
}}}}
\put(9001,2339){\makebox(0,0)[lb]{\smash{{\SetFigFont{14}{16.8}{\rmdefault}{\mddefault}{\updefault}{\color[rgb]{0,0,0}$\epsilon$}%
}}}}
\put(10801,1364){\makebox(0,0)[lb]{\smash{{\SetFigFont{14}{16.8}{\rmdefault}{\mddefault}{\updefault}{\color[rgb]{0,0,0}$a$}%
}}}}
\put(12601,764){\makebox(0,0)[lb]{\smash{{\SetFigFont{14}{16.8}{\rmdefault}{\mddefault}{\updefault}{\color[rgb]{0,0,0}$\epsilon$}%
}}}}
\put(11701,764){\makebox(0,0)[lb]{\smash{{\SetFigFont{14}{16.8}{\rmdefault}{\mddefault}{\updefault}{\color[rgb]{0,0,0}$a$}%
}}}}
\put(8626,164){\makebox(0,0)[lb]{\smash{{\SetFigFont{14}{16.8}{\rmdefault}{\mddefault}{\updefault}{\color[rgb]{0,0,0}$\epsilon$}%
}}}}
\put(8626,-1936){\makebox(0,0)[lb]{\smash{{\SetFigFont{14}{16.8}{\rmdefault}{\mddefault}{\updefault}{\color[rgb]{0,0,0}$\epsilon$}%
}}}}
\put(4726,-2911){\makebox(0,0)[lb]{\smash{{\SetFigFont{14}{16.8}{\rmdefault}{\mddefault}{\updefault}{\color[rgb]{0,0,0}$4$}%
}}}}
\put(4576,-3436){\makebox(0,0)[lb]{\smash{{\SetFigFont{14}{16.8}{\rmdefault}{\mddefault}{\updefault}{\color[rgb]{0,0,0}$(4,5)$}%
}}}}
\put(4576,-4036){\makebox(0,0)[lb]{\smash{{\SetFigFont{14}{16.8}{\rmdefault}{\mddefault}{\updefault}{\color[rgb]{0,0,0}$(5,6)$}%
}}}}
\put(4726,-4636){\makebox(0,0)[lb]{\smash{{\SetFigFont{14}{16.8}{\rmdefault}{\mddefault}{\updefault}{\color[rgb]{0,0,0}$6$}%
}}}}
\put(4801,-3736){\makebox(0,0)[lb]{\smash{{\SetFigFont{14}{16.8}{\rmdefault}{\mddefault}{\updefault}{\color[rgb]{0,0,0}$\vdots$}%
}}}}
\put(4576,-5236){\makebox(0,0)[lb]{\smash{{\SetFigFont{14}{16.8}{\rmdefault}{\mddefault}{\updefault}{\color[rgb]{0,0,0}$(6,7)$}%
}}}}
\put(4801,-5536){\makebox(0,0)[lb]{\smash{{\SetFigFont{14}{16.8}{\rmdefault}{\mddefault}{\updefault}{\color[rgb]{0,0,0}$\vdots$}%
}}}}
\put(7126,-1936){\makebox(0,0)[lb]{\smash{{\SetFigFont{14}{16.8}{\rmdefault}{\mddefault}{\updefault}{\color[rgb]{0,0,0}$\epsilon$}%
}}}}
\put(10726,-3286){\makebox(0,0)[lb]{\smash{{\SetFigFont{14}{16.8}{\rmdefault}{\mddefault}{\updefault}{\color[rgb]{0,0,0}$a$}%
}}}}
\put(6826,-4336){\makebox(0,0)[lb]{\smash{{\SetFigFont{14}{16.8}{\rmdefault}{\mddefault}{\updefault}{\color[rgb]{0,0,0}$a$}%
}}}}
\put(8626,-3736){\makebox(0,0)[lb]{\smash{{\SetFigFont{14}{16.8}{\rmdefault}{\mddefault}{\updefault}{\color[rgb]{0,0,0}$a$}%
}}}}
\put(4576,-5836){\makebox(0,0)[lb]{\smash{{\SetFigFont{14}{16.8}{\rmdefault}{\mddefault}{\updefault}{\color[rgb]{0,0,0}$(7,8)$}%
}}}}
\put(6826,-6136){\makebox(0,0)[lb]{\smash{{\SetFigFont{14}{16.8}{\rmdefault}{\mddefault}{\updefault}{\color[rgb]{0,0,0}$\epsilon$}%
}}}}
\put(8626,-5536){\makebox(0,0)[lb]{\smash{{\SetFigFont{14}{16.8}{\rmdefault}{\mddefault}{\updefault}{\color[rgb]{0,0,0}$\epsilon$}%
}}}}
\put(12601,-4936){\makebox(0,0)[lb]{\smash{{\SetFigFont{14}{16.8}{\rmdefault}{\mddefault}{\updefault}{\color[rgb]{0,0,0}$\epsilon$}%
}}}}
\put(10726,-6736){\makebox(0,0)[lb]{\smash{{\SetFigFont{14}{16.8}{\rmdefault}{\mddefault}{\updefault}{\color[rgb]{0,0,0}$\epsilon$}%
}}}}
\put(11176,-7411){\makebox(0,0)[lb]{\smash{{\SetFigFont{14}{16.8}{\rmdefault}{\mddefault}{\updefault}{\color[rgb]{0,0,0}$\epsilon$}%
}}}}
\put(12601,-3136){\makebox(0,0)[lb]{\smash{{\SetFigFont{14}{16.8}{\rmdefault}{\mddefault}{\updefault}{\color[rgb]{0,0,0}$a$}%
}}}}
\put(12601,-6436){\makebox(0,0)[lb]{\smash{{\SetFigFont{14}{16.8}{\rmdefault}{\mddefault}{\updefault}{\color[rgb]{0,0,0}$\epsilon$}%
}}}}
\put(4801,-8536){\makebox(0,0)[lb]{\smash{{\SetFigFont{14}{16.8}{\rmdefault}{\mddefault}{\updefault}{\color[rgb]{0,0,0}$\vdots$}%
}}}}
\put(4651,-10636){\makebox(0,0)[lb]{\smash{{\SetFigFont{14}{16.8}{\rmdefault}{\mddefault}{\updefault}{\color[rgb]{0,0,0}$13$}%
}}}}
\put(4501,-8236){\makebox(0,0)[lb]{\smash{{\SetFigFont{14}{16.8}{\rmdefault}{\mddefault}{\updefault}{\color[rgb]{0,0,0}$(9,10)$}%
}}}}
\put(4426,-10036){\makebox(0,0)[lb]{\smash{{\SetFigFont{14}{16.8}{\rmdefault}{\mddefault}{\updefault}{\color[rgb]{0,0,0}$(11,12)$}%
}}}}
\put(4801,-10336){\makebox(0,0)[lb]{\smash{{\SetFigFont{14}{16.8}{\rmdefault}{\mddefault}{\updefault}{\color[rgb]{0,0,0}$\vdots$}%
}}}}
\put(8851,-9286){\makebox(0,0)[lb]{\smash{{\SetFigFont{14}{16.8}{\rmdefault}{\mddefault}{\updefault}{\color[rgb]{0,0,0}$\vdots$}%
}}}}
\put(12451,-11086){\makebox(0,0)[lb]{\smash{{\SetFigFont{14}{16.8}{\rmdefault}{\mddefault}{\updefault}{\color[rgb]{0,0,0}$\vdots$}%
}}}}
\put(14851,-11686){\makebox(0,0)[lb]{\smash{{\SetFigFont{14}{16.8}{\rmdefault}{\mddefault}{\updefault}{\color[rgb]{0,0,0}$\vdots$}%
}}}}
\put(15001,-10036){\makebox(0,0)[lb]{\smash{{\SetFigFont{14}{16.8}{\rmdefault}{\mddefault}{\updefault}{\color[rgb]{0,0,0}$\epsilon$}%
}}}}
\put(10726,-9436){\makebox(0,0)[lb]{\smash{{\SetFigFont{14}{16.8}{\rmdefault}{\mddefault}{\updefault}{\color[rgb]{0,0,0}$a$}%
}}}}
\put(4726,-7636){\makebox(0,0)[lb]{\smash{{\SetFigFont{14}{16.8}{\rmdefault}{\mddefault}{\updefault}{\color[rgb]{0,0,0}$9$}%
}}}}
\put(15001,-8536){\makebox(0,0)[lb]{\smash{{\SetFigFont{14}{16.8}{\rmdefault}{\mddefault}{\updefault}{\color[rgb]{0,0,0}$a$}%
}}}}
\put(4501,-7036){\makebox(0,0)[lb]{\smash{{\SetFigFont{14}{16.8}{\rmdefault}{\mddefault}{\updefault}{\color[rgb]{0,0,0}$(8,9)$}%
}}}}
\put(7051,-8086){\makebox(0,0)[lb]{\smash{{\SetFigFont{14}{16.8}{\rmdefault}{\mddefault}{\updefault}{\color[rgb]{0,0,0}$\vdots$}%
}}}}
\put(13651,-7786){\makebox(0,0)[lb]{\smash{{\SetFigFont{14}{16.8}{\rmdefault}{\mddefault}{\updefault}{\color[rgb]{0,0,0}$a$}%
}}}}
\put(4726,-6436){\makebox(0,0)[lb]{\smash{{\SetFigFont{14}{16.8}{\rmdefault}{\mddefault}{\updefault}{\color[rgb]{0,0,0}$8$}%
}}}}
\put(8626,-7336){\makebox(0,0)[lb]{\smash{{\SetFigFont{14}{16.8}{\rmdefault}{\mddefault}{\updefault}{\color[rgb]{0,0,0}$\epsilon$}%
}}}}
\put(4426,-8836){\makebox(0,0)[lb]{\smash{{\SetFigFont{14}{16.8}{\rmdefault}{\mddefault}{\updefault}{\color[rgb]{0,0,0}$(10,11)$}%
}}}}
\put(4651,-9436){\makebox(0,0)[lb]{\smash{{\SetFigFont{14}{16.8}{\rmdefault}{\mddefault}{\updefault}{\color[rgb]{0,0,0}$11$}%
}}}}
\put(8581,-6286){\makebox(0,0)[lb]{\smash{{\SetFigFont{14}{16.8}{\rmdefault}{\mddefault}{\updefault}{\color[rgb]{0,0,0}$(2,0)$}%
}}}}
\put(8506,-6586){\makebox(0,0)[lb]{\smash{{\SetFigFont{14}{16.8}{\rmdefault}{\mddefault}{\updefault}{\color[rgb]{0,0,0}$1$}%
}}}}
\put(8881,-6586){\makebox(0,0)[lb]{\smash{{\SetFigFont{14}{16.8}{\rmdefault}{\mddefault}{\updefault}{\color[rgb]{0,0,0}$\Delta_0$}%
}}}}
\put(12601,-9136){\makebox(0,0)[lb]{\smash{{\SetFigFont{14}{16.8}{\rmdefault}{\mddefault}{\updefault}{\color[rgb]{0,0,0}$a$}%
}}}}
\put(14626,-6436){\makebox(0,0)[lb]{\smash{{\SetFigFont{14}{16.8}{\rmdefault}{\mddefault}{\updefault}{\color[rgb]{0,0,0}$\epsilon$}%
}}}}
\put(11731,-12991){\makebox(0,0)[lb]{\smash{{\SetFigFont{14}{16.8}{\rmdefault}{\mddefault}{\updefault}{\color[rgb]{0,0,0}$(x>4) \wedge (y \geq 4)$}%
}}}}
\put(6451,-14386){\makebox(0,0)[lb]{\smash{{\SetFigFont{14}{16.8}{\rmdefault}{\mddefault}{\updefault}{\color[rgb]{0,0,0}$0 \leq x < 1, \{x\}$}%
}}}}
\put(9256,-13516){\makebox(0,0)[lb]{\smash{{\SetFigFont{14}{16.8}{\rmdefault}{\mddefault}{\updefault}{\color[rgb]{0,0,0}$1$}%
}}}}
\put(7606,-13516){\makebox(0,0)[lb]{\smash{{\SetFigFont{14}{16.8}{\rmdefault}{\mddefault}{\updefault}{\color[rgb]{0,0,0}$0$}%
}}}}
\put(7981,-13291){\makebox(0,0)[lb]{\smash{{\SetFigFont{14}{16.8}{\rmdefault}{\mddefault}{\updefault}{\color[rgb]{0,0,0}$1 < x \leq 2$}%
}}}}
\put(10906,-13516){\makebox(0,0)[lb]{\smash{{\SetFigFont{14}{16.8}{\rmdefault}{\mddefault}{\updefault}{\color[rgb]{0,0,0}$2$}%
}}}}
\put(9256,-15616){\makebox(0,0)[lb]{\smash{{\SetFigFont{14}{16.8}{\rmdefault}{\mddefault}{\updefault}{\color[rgb]{0,0,0}$a$}%
}}}}
\put(9331,-11941){\makebox(0,0)[lb]{\smash{{\SetFigFont{14}{16.8}{\rmdefault}{\mddefault}{\updefault}{\color[rgb]{0,0,0}$a$}%
}}}}
\put(8881,-12166){\makebox(0,0)[lb]{\smash{{\SetFigFont{14}{16.8}{\rmdefault}{\mddefault}{\updefault}{\color[rgb]{0,0,0}$x=2, \{x\}$}%
}}}}
\put(9256,-14641){\makebox(0,0)[lb]{\smash{{\SetFigFont{14}{16.8}{\rmdefault}{\mddefault}{\updefault}{\color[rgb]{0,0,0}$3$}%
}}}}
\put(8431,-13066){\makebox(0,0)[lb]{\smash{{\SetFigFont{14}{16.8}{\rmdefault}{\mddefault}{\updefault}{\color[rgb]{0,0,0}$a$}%
}}}}
\put(8806,-15841){\makebox(0,0)[lb]{\smash{{\SetFigFont{14}{16.8}{\rmdefault}{\mddefault}{\updefault}{\color[rgb]{0,0,0}$0 < x < 1$}%
}}}}
\put(7081,-14116){\makebox(0,0)[lb]{\smash{{\SetFigFont{14}{16.8}{\rmdefault}{\mddefault}{\updefault}{\color[rgb]{0,0,0}$\epsilon$}%
}}}}
\put(11656,-14266){\makebox(0,0)[lb]{\smash{{\SetFigFont{14}{16.8}{\rmdefault}{\mddefault}{\updefault}{\color[rgb]{0,0,0}$y=2$}%
}}}}
\put(11881,-14041){\makebox(0,0)[lb]{\smash{{\SetFigFont{14}{16.8}{\rmdefault}{\mddefault}{\updefault}{\color[rgb]{0,0,0}$b$}%
}}}}
\put(9631,-13891){\makebox(0,0)[lb]{\smash{{\SetFigFont{14}{16.8}{\rmdefault}{\mddefault}{\updefault}{\color[rgb]{0,0,0}$\epsilon$}%
}}}}
\put(9406,-14116){\makebox(0,0)[lb]{\smash{{\SetFigFont{14}{16.8}{\rmdefault}{\mddefault}{\updefault}{\color[rgb]{0,0,0}$y=2, \{y\}$}%
}}}}
\put(10531,-14416){\makebox(0,0)[lb]{\smash{{\SetFigFont{14}{16.8}{\rmdefault}{\mddefault}{\updefault}{\color[rgb]{0,0,0}$a$}%
}}}}
\put(12181,-12766){\makebox(0,0)[lb]{\smash{{\SetFigFont{14}{16.8}{\rmdefault}{\mddefault}{\updefault}{\color[rgb]{0,0,0}$c$}%
}}}}
\put(10156,-14641){\makebox(0,0)[lb]{\smash{{\SetFigFont{14}{16.8}{\rmdefault}{\mddefault}{\updefault}{\color[rgb]{0,0,0}$x=2, \{x,y\}$}%
}}}}
\put(9376,-13111){\makebox(0,0)[lb]{\smash{{\SetFigFont{14}{16.8}{\rmdefault}{\mddefault}{\updefault}{\color[rgb]{0,0,0}$3<x\leq4, \{x\}$}%
}}}}
\put(9976,-13336){\makebox(0,0)[lb]{\smash{{\SetFigFont{14}{16.8}{\rmdefault}{\mddefault}{\updefault}{\color[rgb]{0,0,0}$\epsilon$}%
}}}}
\put(9076,-16336){\makebox(0,0)[lb]{\smash{{\SetFigFont{14}{16.8}{\rmdefault}{\mddefault}{\updefault}{\color[rgb]{0,0,0}$(a)$}%
}}}}
\put(9076,-10786){\makebox(0,0)[lb]{\smash{{\SetFigFont{14}{16.8}{\rmdefault}{\mddefault}{\updefault}{\color[rgb]{0,0,0}$(b)$}%
}}}}
\put(12751,-4711){\makebox(0,0)[lb]{\smash{{\SetFigFont{20}{24.0}{\rmdefault}{\bfdefault}{\updefault}{\color[rgb]{0,0,0}$C_3$}%
}}}}
\put(12601,-1636){\makebox(0,0)[lb]{\smash{{\SetFigFont{14}{16.8}{\rmdefault}{\mddefault}{\updefault}{\color[rgb]{0,0,0}$\epsilon$}%
}}}}
\put(13276,-1636){\makebox(0,0)[lb]{\smash{{\SetFigFont{20}{24.0}{\rmdefault}{\bfdefault}{\updefault}{\color[rgb]{0,0,0}$C_4$}%
}}}}
\put(7351,-5911){\makebox(0,0)[lb]{\smash{{\SetFigFont{20}{24.0}{\rmdefault}{\bfdefault}{\updefault}{\color[rgb]{0,0,0}$C_3$}%
}}}}
\end{picture}%

%% file: ta3.pdf_t
\begin{picture}(0,0)%
\includegraphics{ta3.pdf}%
\end{picture}%
\setlength{\unitlength}{3947sp}%
\begingroup\makeatletter\ifx\SetFigFont\undefined%
\gdef\SetFigFont#1#2#3#4#5{%
  \reset@font\fontsize{#1}{#2pt}%
  \fontfamily{#3}\fontseries{#4}\fontshape{#5}%
  \selectfont}%
\fi\endgroup%
\begin{picture}(11652,10875)(6961,-4216)
\put(7126,4439){\makebox(0,0)[lb]{\smash{{\SetFigFont{20}{24.0}{\rmdefault}{\mddefault}{\updefault}{\color[rgb]{0,0,0}$t$}%
}}}}
\put(10951,5789){\makebox(0,0)[lb]{\smash{{\SetFigFont{14}{16.8}{\rmdefault}{\mddefault}{\updefault}{\color[rgb]{0,0,0}$1$}%
}}}}
\put(9301,5789){\makebox(0,0)[lb]{\smash{{\SetFigFont{14}{16.8}{\rmdefault}{\mddefault}{\updefault}{\color[rgb]{0,0,0}$0$}%
}}}}
\put(12601,5789){\makebox(0,0)[lb]{\smash{{\SetFigFont{14}{16.8}{\rmdefault}{\mddefault}{\updefault}{\color[rgb]{0,0,0}$2$}%
}}}}
\put(14251,5789){\makebox(0,0)[lb]{\smash{{\SetFigFont{14}{16.8}{\rmdefault}{\mddefault}{\updefault}{\color[rgb]{0,0,0}$3$}%
}}}}
\put(12976,6239){\makebox(0,0)[lb]{\smash{{\SetFigFont{14}{16.8}{\rmdefault}{\mddefault}{\updefault}{\color[rgb]{0,0,0}$x = 1, \{x\}$}%
}}}}
\put(10051,6239){\makebox(0,0)[lb]{\smash{{\SetFigFont{14}{16.8}{\rmdefault}{\mddefault}{\updefault}{\color[rgb]{0,0,0}$a$}%
}}}}
\put(11776,6239){\makebox(0,0)[lb]{\smash{{\SetFigFont{14}{16.8}{\rmdefault}{\mddefault}{\updefault}{\color[rgb]{0,0,0}$b$}%
}}}}
\put(13351,6464){\makebox(0,0)[lb]{\smash{{\SetFigFont{14}{16.8}{\rmdefault}{\mddefault}{\updefault}{\color[rgb]{0,0,0}$a$}%
}}}}
\put(9676,6014){\makebox(0,0)[lb]{\smash{{\SetFigFont{14}{16.8}{\rmdefault}{\mddefault}{\updefault}{\color[rgb]{0,0,0}$x = 1, \{x\}$}%
}}}}
\put(13501,5339){\makebox(0,0)[lb]{\smash{{\SetFigFont{14}{16.8}{\rmdefault}{\mddefault}{\updefault}{\color[rgb]{0,0,0}$b$}%
}}}}
\put(12526,5114){\makebox(0,0)[lb]{\smash{{\SetFigFont{14}{16.8}{\rmdefault}{\mddefault}{\updefault}{\color[rgb]{0,0,0}$(0 < x) \wedge (y < 1), \{y\}$}%
}}}}
\put(11626,6014){\makebox(0,0)[lb]{\smash{{\SetFigFont{14}{16.8}{\rmdefault}{\mddefault}{\updefault}{\color[rgb]{0,0,0}$\{y\}$}%
}}}}
\put(11551,5264){\makebox(0,0)[lb]{\smash{{\SetFigFont{17}{20.4}{\rmdefault}{\mddefault}{\updefault}{\color[rgb]{0,0,0}$(a)$}%
}}}}
\put(9451,-1036){\makebox(0,0)[lb]{\smash{{\SetFigFont{14}{16.8}{\rmdefault}{\mddefault}{\updefault}{\color[rgb]{0,0,0}$2$}%
}}}}
\put(9451,1364){\makebox(0,0)[lb]{\smash{{\SetFigFont{14}{16.8}{\rmdefault}{\mddefault}{\updefault}{\color[rgb]{0,0,0}$1$}%
}}}}
\put(9451,3764){\makebox(0,0)[lb]{\smash{{\SetFigFont{14}{16.8}{\rmdefault}{\mddefault}{\updefault}{\color[rgb]{0,0,0}$0$}%
}}}}
\put(9826,3614){\makebox(0,0)[lb]{\smash{{\SetFigFont{14}{16.8}{\rmdefault}{\mddefault}{\updefault}{\color[rgb]{0,0,0}$0=\{t\}=\{x\}=\{y\}$}%
}}}}
\put(7126,-3436){\makebox(0,0)[lb]{\smash{{\SetFigFont{14}{16.8}{\rmdefault}{\mddefault}{\updefault}{\color[rgb]{0,0,0}$3+\ZNaturals$}%
}}}}
\put(9451,-3436){\makebox(0,0)[lb]{\smash{{\SetFigFont{14}{16.8}{\rmdefault}{\mddefault}{\updefault}{\color[rgb]{0,0,0}$2$}%
}}}}
\put(9451,164){\makebox(0,0)[lb]{\smash{{\SetFigFont{14}{16.8}{\rmdefault}{\mddefault}{\updefault}{\color[rgb]{0,0,0}$2$}%
}}}}
\put(9151,-436){\makebox(0,0)[lb]{\smash{{\SetFigFont{14}{16.8}{\rmdefault}{\mddefault}{\updefault}{\color[rgb]{0,0,0}$b$}%
}}}}
\put(9226,764){\makebox(0,0)[lb]{\smash{{\SetFigFont{14}{16.8}{\rmdefault}{\mddefault}{\updefault}{\color[rgb]{0,0,0}$b$}%
}}}}
\put(9826,1214){\makebox(0,0)[lb]{\smash{{\SetFigFont{14}{16.8}{\rmdefault}{\mddefault}{\updefault}{\color[rgb]{0,0,0}$0=\{t\}=\{x\}=\{y\}$}%
}}}}
\put(12751,1364){\makebox(0,0)[lb]{\smash{{\SetFigFont{14}{16.8}{\rmdefault}{\mddefault}{\updefault}{\color[rgb]{0,0,0}$2$}%
}}}}
\put(9826, 14){\makebox(0,0)[lb]{\smash{{\SetFigFont{14}{16.8}{\rmdefault}{\mddefault}{\updefault}{\color[rgb]{0,0,0}$0=\{y\}<\{t\}=\{x\}$}%
}}}}
\put(10126,-2386){\makebox(0,0)[lb]{\smash{{\SetFigFont{14}{16.8}{\rmdefault}{\mddefault}{\updefault}{\color[rgb]{0,0,0}$0=\{y\}<\{t\}$}%
}}}}
\put(10126,-2086){\makebox(0,0)[lb]{\smash{{\SetFigFont{14}{16.8}{\rmdefault}{\mddefault}{\updefault}{\color[rgb]{0,0,0}$(2+\ZNaturals,\top,0)$}%
}}}}
\put(10126,-3286){\makebox(0,0)[lb]{\smash{{\SetFigFont{14}{16.8}{\rmdefault}{\mddefault}{\updefault}{\color[rgb]{0,0,0}$(3+\ZNaturals,\top,0)$}%
}}}}
\put(10351,1514){\makebox(0,0)[lb]{\smash{{\SetFigFont{14}{16.8}{\rmdefault}{\mddefault}{\updefault}{\color[rgb]{0,0,0}$(1,0,1)$}%
}}}}
\put(10351,314){\makebox(0,0)[lb]{\smash{{\SetFigFont{14}{16.8}{\rmdefault}{\mddefault}{\updefault}{\color[rgb]{0,0,0}$(1,0,0)$}%
}}}}
\put(10351,-886){\makebox(0,0)[lb]{\smash{{\SetFigFont{14}{16.8}{\rmdefault}{\mddefault}{\updefault}{\color[rgb]{0,0,0}$(2,1,0)$}%
}}}}
\put(10351,3914){\makebox(0,0)[lb]{\smash{{\SetFigFont{14}{16.8}{\rmdefault}{\mddefault}{\updefault}{\color[rgb]{0,0,0}$(0,0,0)$}%
}}}}
\put(10126,-3586){\makebox(0,0)[lb]{\smash{{\SetFigFont{14}{16.8}{\rmdefault}{\mddefault}{\updefault}{\color[rgb]{0,0,0}$0=\{t\}=\{y\}$}%
}}}}
\put(13126,1214){\makebox(0,0)[lb]{\smash{{\SetFigFont{14}{16.8}{\rmdefault}{\mddefault}{\updefault}{\color[rgb]{0,0,0}$0=\{t\}=\{x\}=\{y\}$}%
}}}}
\put(13651,1514){\makebox(0,0)[lb]{\smash{{\SetFigFont{14}{16.8}{\rmdefault}{\mddefault}{\updefault}{\color[rgb]{0,0,0}$(1,0,0)$}%
}}}}
\put(12226,1514){\makebox(0,0)[lb]{\smash{{\SetFigFont{14}{16.8}{\rmdefault}{\mddefault}{\updefault}{\color[rgb]{0,0,0}$b$}%
}}}}
\put(12751,-1036){\makebox(0,0)[lb]{\smash{{\SetFigFont{14}{16.8}{\rmdefault}{\mddefault}{\updefault}{\color[rgb]{0,0,0}$3$}%
}}}}
\put(16051,-1036){\makebox(0,0)[lb]{\smash{{\SetFigFont{14}{16.8}{\rmdefault}{\mddefault}{\updefault}{\color[rgb]{0,0,0}$3$}%
}}}}
\put(12751,-3436){\makebox(0,0)[lb]{\smash{{\SetFigFont{14}{16.8}{\rmdefault}{\mddefault}{\updefault}{\color[rgb]{0,0,0}$3$}%
}}}}
\put(12751,-2236){\makebox(0,0)[lb]{\smash{{\SetFigFont{14}{16.8}{\rmdefault}{\mddefault}{\updefault}{\color[rgb]{0,0,0}$2$}%
}}}}
\put(11926,-4111){\makebox(0,0)[lb]{\smash{{\SetFigFont{17}{20.4}{\rmdefault}{\mddefault}{\updefault}{\color[rgb]{0,0,0}$(b)$}%
}}}}
\put(13651,-886){\makebox(0,0)[lb]{\smash{{\SetFigFont{14}{16.8}{\rmdefault}{\mddefault}{\updefault}{\color[rgb]{0,0,0}$(2,0,1)$}%
}}}}
\put(16951,-886){\makebox(0,0)[lb]{\smash{{\SetFigFont{14}{16.8}{\rmdefault}{\mddefault}{\updefault}{\color[rgb]{0,0,0}$(2,0,0)$}%
}}}}
\put(13426,-3286){\makebox(0,0)[lb]{\smash{{\SetFigFont{14}{16.8}{\rmdefault}{\mddefault}{\updefault}{\color[rgb]{0,0,0}$(3+\ZNaturals,0,0)$}%
}}}}
\put(13426,-2086){\makebox(0,0)[lb]{\smash{{\SetFigFont{14}{16.8}{\rmdefault}{\mddefault}{\updefault}{\color[rgb]{0,0,0}$(2+\ZNaturals,0,0)$}%
}}}}
\put(13126,-3586){\makebox(0,0)[lb]{\smash{{\SetFigFont{14}{16.8}{\rmdefault}{\mddefault}{\updefault}{\color[rgb]{0,0,0}$0=\{t\}=\{x\}<\{y\}$}%
}}}}
\put(13126,-2386){\makebox(0,0)[lb]{\smash{{\SetFigFont{14}{16.8}{\rmdefault}{\mddefault}{\updefault}{\color[rgb]{0,0,0}$0=\{y\}<\{t\}=\{x\}$}%
}}}}
\put(9826,-1186){\makebox(0,0)[lb]{\smash{{\SetFigFont{14}{16.8}{\rmdefault}{\mddefault}{\updefault}{\color[rgb]{0,0,0}$0=\{t\}=\{x\}=\{y\}$}%
}}}}
\put(13126,-1186){\makebox(0,0)[lb]{\smash{{\SetFigFont{14}{16.8}{\rmdefault}{\mddefault}{\updefault}{\color[rgb]{0,0,0}$0=\{t\}=\{x\}=\{y\}$}%
}}}}
\put(16426,-1186){\makebox(0,0)[lb]{\smash{{\SetFigFont{14}{16.8}{\rmdefault}{\mddefault}{\updefault}{\color[rgb]{0,0,0}$0=\{t\}=\{x\}<\{y\}$}%
}}}}
\put(8701,-2911){\makebox(0,0)[lb]{\smash{{\SetFigFont{14}{16.8}{\rmdefault}{\mddefault}{\updefault}{\color[rgb]{0,0,0}$(*)$}%
}}}}
\put(8776,-1711){\makebox(0,0)[lb]{\smash{{\SetFigFont{14}{16.8}{\rmdefault}{\mddefault}{\updefault}{\color[rgb]{0,0,0}$(*)$}%
}}}}
\put(9151,-2911){\makebox(0,0)[lb]{\smash{{\SetFigFont{14}{16.8}{\rmdefault}{\mddefault}{\updefault}{\color[rgb]{0,0,0}$b$}%
}}}}
\put(9451,-2236){\makebox(0,0)[lb]{\smash{{\SetFigFont{14}{16.8}{\rmdefault}{\mddefault}{\updefault}{\color[rgb]{0,0,0}$2$}%
}}}}
\put(9151,-1711){\makebox(0,0)[lb]{\smash{{\SetFigFont{14}{16.8}{\rmdefault}{\mddefault}{\updefault}{\color[rgb]{0,0,0}$b$}%
}}}}
\put(10726,2564){\makebox(0,0)[lb]{\smash{{\SetFigFont{14}{16.8}{\rmdefault}{\mddefault}{\updefault}{\color[rgb]{0,0,0}$a$}%
}}}}
\put(13426,-2836){\makebox(0,0)[lb]{\smash{{\SetFigFont{14}{16.8}{\rmdefault}{\mddefault}{\updefault}{\color[rgb]{0,0,0}$a$}%
}}}}
\put(14326,-2836){\makebox(0,0)[lb]{\smash{{\SetFigFont{14}{16.8}{\rmdefault}{\mddefault}{\updefault}{\color[rgb]{0,0,0}$b$}%
}}}}
\put(14026,-1636){\makebox(0,0)[lb]{\smash{{\SetFigFont{14}{16.8}{\rmdefault}{\mddefault}{\updefault}{\color[rgb]{0,0,0}$b$}%
}}}}
\put(16726,-1936){\makebox(0,0)[lb]{\smash{{\SetFigFont{14}{16.8}{\rmdefault}{\mddefault}{\updefault}{\color[rgb]{0,0,0}$b$}%
}}}}
\put(14026, 89){\makebox(0,0)[lb]{\smash{{\SetFigFont{14}{16.8}{\rmdefault}{\mddefault}{\updefault}{\color[rgb]{0,0,0}$a$}%
}}}}
\put(14626,314){\makebox(0,0)[lb]{\smash{{\SetFigFont{14}{16.8}{\rmdefault}{\mddefault}{\updefault}{\color[rgb]{0,0,0}$a$}%
}}}}
\put(12226,-886){\makebox(0,0)[lb]{\smash{{\SetFigFont{14}{16.8}{\rmdefault}{\mddefault}{\updefault}{\color[rgb]{0,0,0}$a$}%
}}}}
\put(6976,2564){\makebox(0,0)[lb]{\smash{{\SetFigFont{14}{16.8}{\rmdefault}{\mddefault}{\updefault}{\color[rgb]{0,0,0}$(0,1)$}%
}}}}
\put(6976,164){\makebox(0,0)[lb]{\smash{{\SetFigFont{14}{16.8}{\rmdefault}{\mddefault}{\updefault}{\color[rgb]{0,0,0}$(1,2)$}%
}}}}
\put(6976,-2236){\makebox(0,0)[lb]{\smash{{\SetFigFont{14}{16.8}{\rmdefault}{\mddefault}{\updefault}{\color[rgb]{0,0,0}$(2,3)+ \ZNaturals$}%
}}}}
\put(7126,1364){\makebox(0,0)[lb]{\smash{{\SetFigFont{14}{16.8}{\rmdefault}{\mddefault}{\updefault}{\color[rgb]{0,0,0}$1$}%
}}}}
\put(7126,-1036){\makebox(0,0)[lb]{\smash{{\SetFigFont{14}{16.8}{\rmdefault}{\mddefault}{\updefault}{\color[rgb]{0,0,0}$2$}%
}}}}
\put(7126,3764){\makebox(0,0)[lb]{\smash{{\SetFigFont{14}{16.8}{\rmdefault}{\mddefault}{\updefault}{\color[rgb]{0,0,0}$0$}%
}}}}
\end{picture}%

%% file: ts_autom.pdf_t
\begin{picture}(0,0)%
\includegraphics{ts_autom.pdf}%
\end{picture}%
\setlength{\unitlength}{3947sp}%
\begingroup\makeatletter\ifx\SetFigFont\undefined%
\gdef\SetFigFont#1#2#3#4#5{%
  \reset@font\fontsize{#1}{#2pt}%
  \fontfamily{#3}\fontseries{#4}\fontshape{#5}%
  \selectfont}%
\fi\endgroup%
\begin{picture}(9920,6113)(1639,884)
\put(4501,5189){\makebox(0,0)[lb]{\smash{{\SetFigFont{17}{20.4}{\rmdefault}{\mddefault}{\updefault}{\color[rgb]{0,0,0}$(a)$}%
}}}}
\put(3901,5564){\makebox(0,0)[lb]{\smash{{\SetFigFont{14}{16.8}{\rmdefault}{\mddefault}{\updefault}{\color[rgb]{0,0,0}$1$}%
}}}}
\put(5551,5564){\makebox(0,0)[lb]{\smash{{\SetFigFont{14}{16.8}{\rmdefault}{\mddefault}{\updefault}{\color[rgb]{0,0,0}$2$}%
}}}}
\put(7201,5564){\makebox(0,0)[lb]{\smash{{\SetFigFont{14}{16.8}{\rmdefault}{\mddefault}{\updefault}{\color[rgb]{0,0,0}$3$}%
}}}}
\put(8401,6689){\makebox(0,0)[lb]{\smash{{\SetFigFont{14}{16.8}{\rmdefault}{\mddefault}{\updefault}{\color[rgb]{0,0,0}$4$}%
}}}}
\put(9601,5564){\makebox(0,0)[lb]{\smash{{\SetFigFont{14}{16.8}{\rmdefault}{\mddefault}{\updefault}{\color[rgb]{0,0,0}$5$}%
}}}}
\put(4426,5789){\makebox(0,0)[lb]{\smash{{\SetFigFont{14}{16.8}{\rmdefault}{\mddefault}{\updefault}{\color[rgb]{0,0,0}$x = 5$}%
}}}}
\put(4651,6014){\makebox(0,0)[lb]{\smash{{\SetFigFont{14}{16.8}{\rmdefault}{\mddefault}{\updefault}{\color[rgb]{0,0,0}$a$}%
}}}}
\put(9076,6239){\makebox(0,0)[lb]{\smash{{\SetFigFont{14}{16.8}{\rmdefault}{\mddefault}{\updefault}{\color[rgb]{0,0,0}$x=3$}%
}}}}
\put(9151,4889){\makebox(0,0)[lb]{\smash{{\SetFigFont{14}{16.8}{\rmdefault}{\mddefault}{\updefault}{\color[rgb]{0,0,0}$8<x<18$}%
}}}}
\put(9601,5114){\makebox(0,0)[lb]{\smash{{\SetFigFont{14}{16.8}{\rmdefault}{\mddefault}{\updefault}{\color[rgb]{0,0,0}$a$}%
}}}}
\put(8401,4439){\makebox(0,0)[lb]{\smash{{\SetFigFont{14}{16.8}{\rmdefault}{\mddefault}{\updefault}{\color[rgb]{0,0,0}$6$}%
}}}}
\put(6826,6239){\makebox(0,0)[lb]{\smash{{\SetFigFont{14}{16.8}{\rmdefault}{\mddefault}{\updefault}{\color[rgb]{0,0,0}$0<x<2$}%
}}}}
\put(7276,6464){\makebox(0,0)[lb]{\smash{{\SetFigFont{14}{16.8}{\rmdefault}{\mddefault}{\updefault}{\color[rgb]{0,0,0}$a$}%
}}}}
\put(9301,6464){\makebox(0,0)[lb]{\smash{{\SetFigFont{14}{16.8}{\rmdefault}{\mddefault}{\updefault}{\color[rgb]{0,0,0}$a$}%
}}}}
\put(5926,5789){\makebox(0,0)[lb]{\smash{{\SetFigFont{14}{16.8}{\rmdefault}{\mddefault}{\updefault}{\color[rgb]{0,0,0}$x=6, \{x\}$}%
}}}}
\put(6301,6014){\makebox(0,0)[lb]{\smash{{\SetFigFont{14}{16.8}{\rmdefault}{\mddefault}{\updefault}{\color[rgb]{0,0,0}$a$}%
}}}}
\put(6751,4889){\makebox(0,0)[lb]{\smash{{\SetFigFont{14}{16.8}{\rmdefault}{\mddefault}{\updefault}{\color[rgb]{0,0,0}$x=21,\{x\}$}%
}}}}
\put(7201,5114){\makebox(0,0)[lb]{\smash{{\SetFigFont{14}{16.8}{\rmdefault}{\mddefault}{\updefault}{\color[rgb]{0,0,0}$a$}%
}}}}
\put(10726,3689){\makebox(0,0)[lb]{\smash{{\SetFigFont{14}{16.8}{\rmdefault}{\mddefault}{\updefault}{\color[rgb]{0,0,0}$x=5$}%
}}}}
\put(10801,2339){\makebox(0,0)[lb]{\smash{{\SetFigFont{14}{16.8}{\rmdefault}{\mddefault}{\updefault}{\color[rgb]{0,0,0}$x=8$}%
}}}}
\put(7351,3239){\makebox(0,0)[lb]{\smash{{\SetFigFont{14}{16.8}{\rmdefault}{\mddefault}{\updefault}{\color[rgb]{0,0,0}$6 < x < 7, \{x\}$}%
}}}}
\put(8326,2339){\makebox(0,0)[lb]{\smash{{\SetFigFont{14}{16.8}{\rmdefault}{\mddefault}{\updefault}{\color[rgb]{0,0,0}$x=10,\{x\}$}%
}}}}
\put(8776,3014){\makebox(0,0)[lb]{\smash{{\SetFigFont{14}{16.8}{\rmdefault}{\mddefault}{\updefault}{\color[rgb]{0,0,0}$10$}%
}}}}
\put(9976,4139){\makebox(0,0)[lb]{\smash{{\SetFigFont{14}{16.8}{\rmdefault}{\mddefault}{\updefault}{\color[rgb]{0,0,0}$11$}%
}}}}
\put(11176,3014){\makebox(0,0)[lb]{\smash{{\SetFigFont{14}{16.8}{\rmdefault}{\mddefault}{\updefault}{\color[rgb]{0,0,0}$12$}%
}}}}
\put(9976,1889){\makebox(0,0)[lb]{\smash{{\SetFigFont{14}{16.8}{\rmdefault}{\mddefault}{\updefault}{\color[rgb]{0,0,0}$13$}%
}}}}
\put(7201,3014){\makebox(0,0)[lb]{\smash{{\SetFigFont{14}{16.8}{\rmdefault}{\mddefault}{\updefault}{\color[rgb]{0,0,0}$9$}%
}}}}
\put(2251,3014){\makebox(0,0)[lb]{\smash{{\SetFigFont{14}{16.8}{\rmdefault}{\mddefault}{\updefault}{\color[rgb]{0,0,0}$0$}%
}}}}
\put(3901,3014){\makebox(0,0)[lb]{\smash{{\SetFigFont{14}{16.8}{\rmdefault}{\mddefault}{\updefault}{\color[rgb]{0,0,0}$7$}%
}}}}
\put(5551,3014){\makebox(0,0)[lb]{\smash{{\SetFigFont{14}{16.8}{\rmdefault}{\mddefault}{\updefault}{\color[rgb]{0,0,0}$8$}%
}}}}
\put(6076,3239){\makebox(0,0)[lb]{\smash{{\SetFigFont{14}{16.8}{\rmdefault}{\mddefault}{\updefault}{\color[rgb]{0,0,0}$x = 5$}%
}}}}
\put(6301,3464){\makebox(0,0)[lb]{\smash{{\SetFigFont{14}{16.8}{\rmdefault}{\mddefault}{\updefault}{\color[rgb]{0,0,0}$b$}%
}}}}
\put(4276,3239){\makebox(0,0)[lb]{\smash{{\SetFigFont{14}{16.8}{\rmdefault}{\mddefault}{\updefault}{\color[rgb]{0,0,0}$2<x<4$}%
}}}}
\put(4726,3464){\makebox(0,0)[lb]{\smash{{\SetFigFont{14}{16.8}{\rmdefault}{\mddefault}{\updefault}{\color[rgb]{0,0,0}$b$}%
}}}}
\put(2626,3239){\makebox(0,0)[lb]{\smash{{\SetFigFont{14}{16.8}{\rmdefault}{\mddefault}{\updefault}{\color[rgb]{0,0,0}$0 \leq x \leq 1$}%
}}}}
\put(3076,3464){\makebox(0,0)[lb]{\smash{{\SetFigFont{14}{16.8}{\rmdefault}{\mddefault}{\updefault}{\color[rgb]{0,0,0}$b$}%
}}}}
\put(8851,3689){\makebox(0,0)[lb]{\smash{{\SetFigFont{14}{16.8}{\rmdefault}{\mddefault}{\updefault}{\color[rgb]{0,0,0}$x=1$}%
}}}}
\put(9076,3914){\makebox(0,0)[lb]{\smash{{\SetFigFont{14}{16.8}{\rmdefault}{\mddefault}{\updefault}{\color[rgb]{0,0,0}$b$}%
}}}}
\put(10951,3914){\makebox(0,0)[lb]{\smash{{\SetFigFont{14}{16.8}{\rmdefault}{\mddefault}{\updefault}{\color[rgb]{0,0,0}$b$}%
}}}}
\put(11026,2564){\makebox(0,0)[lb]{\smash{{\SetFigFont{14}{16.8}{\rmdefault}{\mddefault}{\updefault}{\color[rgb]{0,0,0}$b$}%
}}}}
\put(7876,3464){\makebox(0,0)[lb]{\smash{{\SetFigFont{14}{16.8}{\rmdefault}{\mddefault}{\updefault}{\color[rgb]{0,0,0}$b$}%
}}}}
\put(8851,2564){\makebox(0,0)[lb]{\smash{{\SetFigFont{14}{16.8}{\rmdefault}{\mddefault}{\updefault}{\color[rgb]{0,0,0}$b$}%
}}}}
\put(4726,1814){\makebox(0,0)[lb]{\smash{{\SetFigFont{14}{16.8}{\rmdefault}{\mddefault}{\updefault}{\color[rgb]{0,0,0}$c$}%
}}}}
\put(3826,1364){\makebox(0,0)[lb]{\smash{{\SetFigFont{14}{16.8}{\rmdefault}{\mddefault}{\updefault}{\color[rgb]{0,0,0}$14$}%
}}}}
\put(5476,1364){\makebox(0,0)[lb]{\smash{{\SetFigFont{14}{16.8}{\rmdefault}{\mddefault}{\updefault}{\color[rgb]{0,0,0}$15$}%
}}}}
\put(7126,1364){\makebox(0,0)[lb]{\smash{{\SetFigFont{14}{16.8}{\rmdefault}{\mddefault}{\updefault}{\color[rgb]{0,0,0}$16$}%
}}}}
\put(5851,1589){\makebox(0,0)[lb]{\smash{{\SetFigFont{14}{16.8}{\rmdefault}{\mddefault}{\updefault}{\color[rgb]{0,0,0}$10<x<\infty$}%
}}}}
\put(6451,1814){\makebox(0,0)[lb]{\smash{{\SetFigFont{14}{16.8}{\rmdefault}{\mddefault}{\updefault}{\color[rgb]{0,0,0}$c$}%
}}}}
\put(4501,1589){\makebox(0,0)[lb]{\smash{{\SetFigFont{14}{16.8}{\rmdefault}{\mddefault}{\updefault}{\color[rgb]{0,0,0}$x=6$}%
}}}}
\put(2026,2039){\makebox(0,0)[lb]{\smash{{\SetFigFont{14}{16.8}{\rmdefault}{\mddefault}{\updefault}{\color[rgb]{0,0,0}$1 \leq x \leq 4$}%
}}}}
\put(2476,2264){\makebox(0,0)[lb]{\smash{{\SetFigFont{14}{16.8}{\rmdefault}{\mddefault}{\updefault}{\color[rgb]{0,0,0}$c$}%
}}}}
\put(2026,4364){\makebox(0,0)[lb]{\smash{{\SetFigFont{14}{16.8}{\rmdefault}{\mddefault}{\updefault}{\color[rgb]{0,0,0}$1 < x \leq 3$}%
}}}}
\put(2476,4589){\makebox(0,0)[lb]{\smash{{\SetFigFont{14}{16.8}{\rmdefault}{\mddefault}{\updefault}{\color[rgb]{0,0,0}$a$}%
}}}}
\put(4501,2639){\makebox(0,0)[lb]{\smash{{\SetFigFont{17}{20.4}{\rmdefault}{\mddefault}{\updefault}{\color[rgb]{0,0,0}$(b)$}%
}}}}
\put(4501,989){\makebox(0,0)[lb]{\smash{{\SetFigFont{17}{20.4}{\rmdefault}{\mddefault}{\updefault}{\color[rgb]{0,0,0}$(c)$}%
}}}}
\end{picture}%

%% file: time_stamp1.pdf_t
\begin{picture}(0,0)%
\includegraphics{time_stamp1.pdf}%
\end{picture}%
\setlength{\unitlength}{3947sp}%
\begingroup\makeatletter\ifx\SetFigFont\undefined%
\gdef\SetFigFont#1#2#3#4#5{%
  \reset@font\fontsize{#1}{#2pt}%
  \fontfamily{#3}\fontseries{#4}\fontshape{#5}%
  \selectfont}%
\fi\endgroup%
\begin{picture}(7902,2175)(8389,4634)
\put(13351,6239){\makebox(0,0)[lb]{\smash{{\SetFigFont{14}{16.8}{\rmdefault}{\mddefault}{\updefault}{\color[rgb]{0,0,0}$0<x<1, \{x\}$}%
}}}}
\put(9976,5339){\makebox(0,0)[lb]{\smash{{\SetFigFont{14}{16.8}{\rmdefault}{\mddefault}{\updefault}{\color[rgb]{0,0,0}$\epsilon$}%
}}}}
\put(9001,5864){\makebox(0,0)[lb]{\smash{{\SetFigFont{14}{16.8}{\rmdefault}{\mddefault}{\updefault}{\color[rgb]{0,0,0}$0$}%
}}}}
\put(10801,5864){\makebox(0,0)[lb]{\smash{{\SetFigFont{14}{16.8}{\rmdefault}{\mddefault}{\updefault}{\color[rgb]{0,0,0}$1$}%
}}}}
\put(14851,5864){\makebox(0,0)[lb]{\smash{{\SetFigFont{14}{16.8}{\rmdefault}{\mddefault}{\updefault}{\color[rgb]{0,0,0}$1$}%
}}}}
\put(13201,5864){\makebox(0,0)[lb]{\smash{{\SetFigFont{14}{16.8}{\rmdefault}{\mddefault}{\updefault}{\color[rgb]{0,0,0}$0$}%
}}}}
\put(14251,5264){\makebox(0,0)[lb]{\smash{{\SetFigFont{17}{20.4}{\rmdefault}{\mddefault}{\updefault}{\color[rgb]{0,0,0}$(b)$}%
}}}}
\put(16276,6014){\makebox(0,0)[lb]{\smash{{\SetFigFont{14}{16.8}{\rmdefault}{\mddefault}{\updefault}{\color[rgb]{0,0,0}$a$}%
}}}}
\put(15901,5789){\makebox(0,0)[lb]{\smash{{\SetFigFont{14}{16.8}{\rmdefault}{\mddefault}{\updefault}{\color[rgb]{0,0,0}$x = 1, \{x\}$}%
}}}}
\put(9526,6389){\makebox(0,0)[lb]{\smash{{\SetFigFont{14}{16.8}{\rmdefault}{\mddefault}{\updefault}{\color[rgb]{0,0,0}$0<x<1$}%
}}}}
\put(9976,6614){\makebox(0,0)[lb]{\smash{{\SetFigFont{14}{16.8}{\rmdefault}{\mddefault}{\updefault}{\color[rgb]{0,0,0}$a$}%
}}}}
\put(9601,5114){\makebox(0,0)[lb]{\smash{{\SetFigFont{14}{16.8}{\rmdefault}{\mddefault}{\updefault}{\color[rgb]{0,0,0}$x=1, \{x\}$}%
}}}}
\put(9826,4739){\makebox(0,0)[lb]{\smash{{\SetFigFont{17}{20.4}{\rmdefault}{\mddefault}{\updefault}{\color[rgb]{0,0,0}$(a)$}%
}}}}
\put(13951,6464){\makebox(0,0)[lb]{\smash{{\SetFigFont{14}{16.8}{\rmdefault}{\mddefault}{\updefault}{\color[rgb]{0,0,0}$a$}%
}}}}
\end{picture}%